\documentclass[showpacs,amsmath,amssymb,onecolumn,superscriptaddress,notitlepage,preprintnumbers,pra]{revtex4-1}
\usepackage{bm}
\usepackage{times}
\usepackage{multirow}
\usepackage{dcolumn}
\usepackage{tabularx}
\usepackage{graphicx}
\usepackage{subfigure}
\usepackage{ae}
\usepackage{lettrine}
\usepackage{color}
\usepackage{amsmath,amsthm,amssymb,amsfonts,boxedminipage}
\usepackage{epstopdf}
\usepackage{booktabs}
\newtheorem{theorem}{Theorem}
\newtheorem{lemma}{Lemma}
\newtheorem{corollary}{Corollary}
\newtheorem{definition}{Definition}

\newtheorem{fact}{Fact}
\newtheorem{problem}{Problem}

\makeatletter
\renewcommand\@biblabel[1]{#1.}
\renewcommand{\thesubfigure}

\usepackage[noend,noline,ruled,linesnumbered]{algorithm2e}
\usepackage{hyperref}

\hypersetup{colorlinks=true, linkcolor=cyan, citecolor=cyan, urlcolor=blue }

\usepackage{braket}

\newcommand{\Ucal}{\mathcal{U}}

\newcommand{\abs}[1]{\left| #1 \right|}

\newcommand{\pku}{Center on Frontiers of Computing Studies, School of Computer Science, Peking University, Beijing 100871, China}

\newcommand{\bnu}{School of Artificial Intelligence,
 Beijing Normal University, Beijing,
 100875, China}

\begin{document}

\title{\Large Classical Algorithms for Hamiltonian Dynamics Mean Value \\and Guided Local Hamiltonian Problem}

\author{Yusen Wu}
\thanks{These authors contributed equally.}
\affiliation{\bnu}

\author{Yukun Zhang}
\thanks{These authors contributed equally.}
\affiliation{\pku}

\author{Chuan Wang}
\affiliation{\bnu}

\author{Xiao Yuan}
\email{xiaoyuan@pku.edu.cn}
\affiliation{\pku}

\begin{abstract}

The efficient simulation of quantum dynamics and ground states is a central challenge in physics and a key frontier for quantum advantage. While short-time evolution in one-dimensional systems can often be simulated classically, extending this to higher dimensions remains difficult.  
Here, we introduce an efficient classical algorithm for simulating the short-time dynamics of arbitrary local quantum systems. For any local Hamiltonian $H$ and constant evolution time $t$, our method estimates expectation values of the form $\langle\psi|e^{iHt} O e^{-iHt}|\psi\rangle$ for global Pauli observables $O$ and stabilizer states $|\psi\rangle$, with high precision and exponentially small failure probability. Furthermore, we present a classical dequantization of a tailored quantum algorithm that efficiently solves the guided local Hamiltonian~(GLH) problem to constant additive error—previously considered classically hard and hence a promising candidate for quantum computational advantage.
These results reveal unexpected classical tractability in constant-time quantum dynamics and fundamental connections between Hamiltonian dynamics mean value and the GLH problem.
Our work refines the boundary between classical and quantum computational power, identifying sharper criteria for regimes where quantum advantage may genuinely emerge.

\end{abstract}

\maketitle

\noindent\textbf{Introduction---}Understanding the classical complexity of quantum many-body systems remains a long-standing open challenge. Two of the most fundamental tasks in this domain—simulating Hamiltonian dynamics, and determining ground-state energies (given a good initial state)—admit efficient quantum algorithms and are known to be BQP-complete when the evolution time and inverse energy precision scale polynomially with system size~\cite{janzing2005ergodic, gharibian2022improved,cade2022complexity}. These problems are therefore among the most promising candidates for demonstrating quantum computational advantage in the worst-case scenario, assuming ${\rm BPP} \neq {\rm BQP}$. Advancing classical algorithms for such tasks is essential for refining the boundary between classical and quantum computational power.

Although simulating general quantum mean value by Hamiltonian dynamics is classically intractable in the worst case~\cite{janzing2005ergodic,bravyi2021classical}, certain structured quantum systems with short evolution times allow for efficient classical simulation. Notably, the short-time dynamics of (quasi-)one-dimensional systems that satisfy the area law can be efficiently simulated using matrix product state methods~\cite{schollwock2005density, schollwock2011density, orus2014practical, bridgeman2017hand}. However, extending these methods to higher-dimensional systems remains a major challenge~\cite{eisert2010colloquium}. Even for frustration-free, gapped Hamiltonians that satisfy the area law~\cite{anshu2022area}, it remains unknown whether their short-time dynamics (the associated quantum expectation values) can be simulated efficiently on a classical computer. Consequently, establishing the classical tractability of the quantum mean-value problem in higher dimensions remains a central open problem in quantum computational complexity.

In parallel, quantum algorithms for the ground-state problem often assume access to a guiding state with non-negligible overlap with the true ground state—a scenario formalized as the guided local Hamiltonian (GLH) problem~\cite{cade2022improved, cade2022complexity, gharibian2022dequantizing}. While the GLH problem is BQP-complete under high-precision demands, it has been shown that such quantum algorithms can be dequantized to yield efficient classical algorithms for constant relative accuracy $\epsilon$~\cite{gharibian2022dequantizing}, thereby eliminating exponential quantum advantage in that regime. However, constant relative error is typically insufficient for practical purposes, as the corresponding additive error $\epsilon\|H\|$ scales with system size, representing a relatively low accuracy for large systems, especially in the scenario $\|H\|={\rm poly}(n)$. Developing dequantized classical algorithms that achieve constant additive accuracy remains an important and open challenge with significant practical implications.

In this paper, we address these challenges by presenting polynomial classical algorithms for simulating Hamiltonian dynamics mean value and solving the GLH problem. 
For the first result, our classical algorithm provides an $\epsilon$-approximation to the mean value in polynomial running time in terms of $n$ and $1/\epsilon$ with an exponentially small failure probability. This approach significantly extends prior works, including Ref.~\cite{wild2023classical} from local to general global observables and Ref.~\cite{bravyi2021classical} from $2$D discrete circuits to high-dimensional continuous Hamiltonian dynamics. For the second result, leveraging the filter‑function formalism~\cite{lin2022heisenberg,wang2023quantum} and dequantization of a tailored quantum algorithm, we reduce the GLH problem to a Hamiltonian dynamics mean value problem, and prove that any classical algorithm efficient for local‑Hamiltonian dynamics can also be used to solve the GLH problem with constant additive error.

Our results depict the quantum advantage requirement for the quantum mean value problem and the GLH problem, see Fig.~\ref{fig:ResultsSummary}. For the Hamiltonian dynamics mean value problem, quantum advantage is only possible when the evolution time exceeds $t=\Omega(\log n)$, which generally requires error correction to avoid the exponential cost of error mitigation~\cite{takagi2023universal, quek2024exponentially}. On the other hand, for the GLH problem, quantum advantage occurs only when the guided state overlap or the additive error is very small, i.e., $\gamma~\textrm{or}~\epsilon\le 1/{\rm poly}(n)$. 
Our results reveal unexpected classical tractability in constant-time quantum dynamics and GLH problems and refine the boundary between classical and quantum computational power.  

\begin{figure}[t]
\centering
\includegraphics[width=\textwidth]{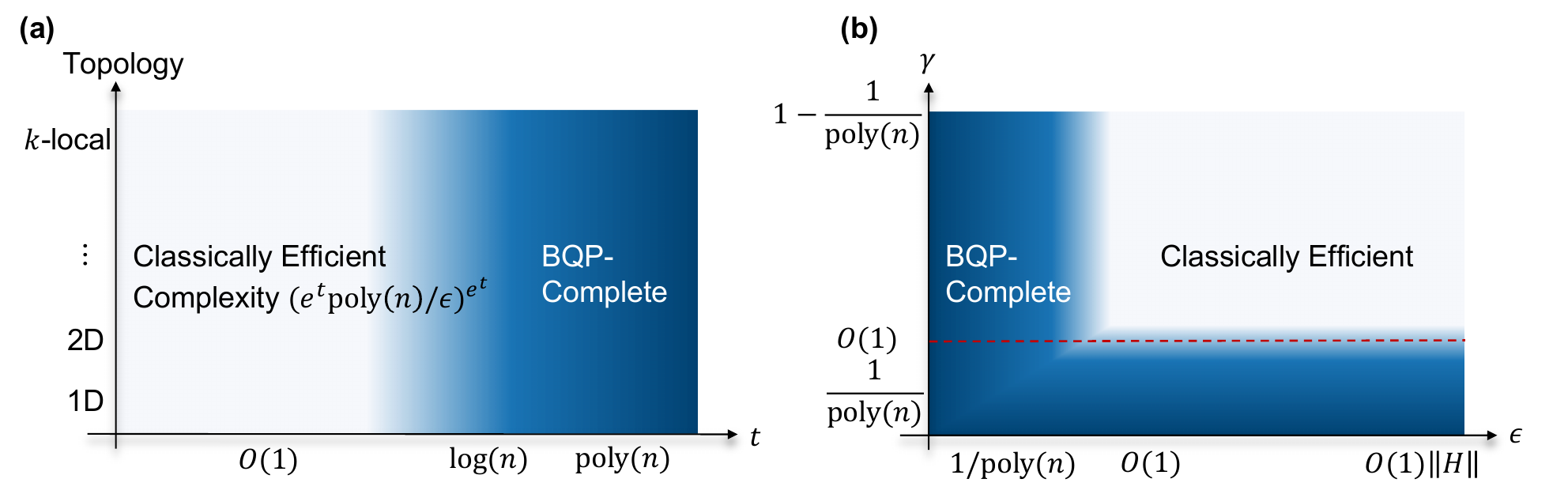}
  \caption{Main Results Summary: (a) Our first result establishes the hardness for estimating the quantum‐dynamics mean value $\langle\psi|e^{iHt}Oe^{-iHt}|\psi\rangle$ for Hamiltonian topologies and and evolution time, where $O$ represents a global operator, and $|\psi\rangle$ is a stabilizer state. Previously, classical efficiency had only been demonstrated for expectation values of shallow 2D quantum circuits \cite{bravyi2021classical}. We generalize this to arbitrary lattice geometries and show that constant‐time Hamiltonian evolution—which is equivalent to a ${\rm poly}\log (n)$-depth quantum circuit. When $t\sim{\rm poly}(n)$, the quantum dynamics mean value problem is proved to be ${\rm BQP}$-complete~\cite{janzing2005ergodic,wild2023classical}. (b) Our second result shows that the quantum-dynamics mean-value algorithm can be applied to compute the ground-state energy of a gapped Hamiltonian in the regime where the overlap satisfies $\gamma=\Omega(1)$ and any constant $\epsilon$. This significantly extendes the previous results which demonstrates the regime $\left(\gamma=\Omega(1), \epsilon=\mathcal{O}(1)\|H\|\right)$ is classically efficient~\cite{gharibian2022dequantizing}. For a normalized Hamiltonian $H/\|H\|$, our classical algorithm achieves an additive error of $\mathcal{O}(1)/\|H\|$ to the ground-state energy. 
  }
  \label{fig:ResultsSummary}
\end{figure}

\begin{figure}[t]
\centering
\includegraphics[width=\textwidth]{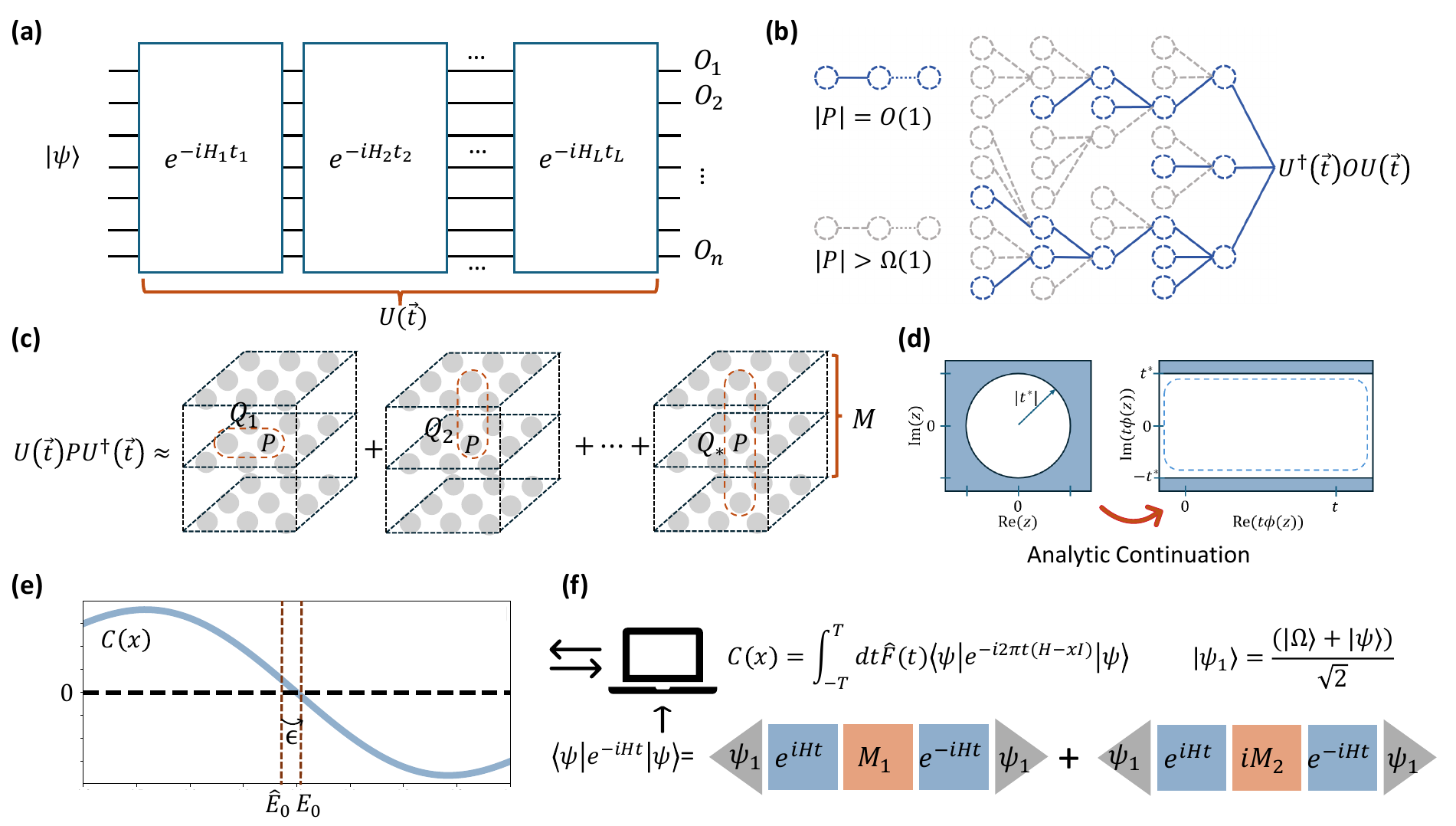}
  \caption{(a) A quantum circuit representation on the quantum dynamics mean value. The quantum dynamics is governed by local Hamiltonians $\{H^{(l)}\}_{l=1}^L$ and corresponding time $\{t_l\}_{l=1}^L$, which acts on the initial Clifford state $|\psi\rangle\sim{\rm Cl}(2^n)
  $. (b) Visualization on approximating $U^{\dagger}(\vec{t})OU(\vec{t})$ by computing low-weight Pauli paths (dark-blue paths) while cutting all large weight Pauli paths $\abs{P}>\Omega(1)$ (light-gray paths). (c) Each grey point represents a single qubit and the red dot circle represents a path of the interaction graph (induced by $\{H^{(l)}\}_{l=1}^L$) connected to $P$. The approximation is essentially a linear combination of ${\rm poly}(n)$ matrices induced by connected clusters, and it is applied to $M\leq\tilde{\mathcal{O}}(e^{L\mathfrak{d}t}\log(e^{L\mathfrak{d}t}/\epsilon))$ qubits. 
  (d) The analytic continuation method provides a paradigm to approximate $U(\vec{t})PU^{\dagger}(\vec{t})$ for general $\max\{\abs{t_l}\}\leq\mathcal{O}(1)$. (e) Visualization of using the filter function $C(x)$ to approximate the ground state energy $E_0$, in which $C(x)$ has a unique zero point around $E_0$. (e) The Loschmidt echo $\langle\psi|e^{-iHt}|\psi\rangle$ can be transformed into linear combinations of quantum mean values, where Hermitian operators $M_1=|\psi_c\rangle\langle\Omega|+|\Omega\rangle\langle\psi_c|$, $M_2=i(|\Omega\rangle\langle\psi_c|-|\psi_c\rangle\langle\Omega|)$, and $|\Omega\rangle$ represents the vacuum state.
  }
  \label{fig:Main_Results}
\end{figure}

\vspace{0.2cm}
\noindent\textbf{Hamiltonian Dynamics Mean Value---}We consider $n$-qubit local Hamiltonian $H=\sum_{X\subset S}\lambda_Xh_X$, where $S$ represents a set of subsystems, real-valued coefficient $\abs{\lambda_X}\leq 1$ and $h_X$ represents a Hermitian operator non-trivially acting on local qubits $X\subset S$. Without loss of generality, we assume the operator norm of each $h_X$ satisfies $\|h_X\|\leq 1$, and $X$ is not necessarily to be geometrically local. 
To characterize the correlation strength within the Hamiltonian, we introduce the associated interaction graph $G$ to depict overlaps of operators contained in $H$~\cite{haah2024learning,wild2023classical, bakshi2024high}. Specifically, given the Hamiltonian terms $\{h_X\}_{X\subset S}$, the interaction graph $G$ is a simple graph with vertex set $\{h_X\}_{X\subset S}$. An edge exists between $h_X$ and $h_{X^{\prime}}$ if $X\cap X^{\prime}\neq\emptyset$, and we denote the degree $\mathfrak{d}(h_X)$ of a vertex $h_X$, which is the number of edges incident to it. The maximum degree among all vertexes within the interaction graph $G$ is denoted by $\mathfrak{d}=\max_{h_X\in H}\left\{\mathfrak{d}(h_X)\right\}$. 

First, we consider the computation of expectation values at the output of an $L$-step Hamiltonian dynamics.

\begin{problem}
    [Hamiltonian Dynamics Mean Value Problem]
\label{problem1}
Consider $L$ local Hamiltonians $\{H^{(1)},H^{(2)},\cdots, H^{(L)}\}$ defined on a $n$-qubit system, and a global observable $O=O_1\otimes\cdots\otimes O_n$ with the operator norm $\|O_i\|\leq 1$ for $i\in[n]$. Let $U(\vec{t})=\prod_{l=1}^Le^{-iH^{(l)}t_l}$, the quantum dynamics mean value is defined by 
\begin{align}
\mu(\vec{t})=\langle\psi_C|U(\vec{t})^{\dagger}O U(\vec{t})|\psi_C\rangle,
    \label{Eq:meanvalue}
\end{align}
 where evolution time series $\vec{t}=\{t_1,\cdots,t_L\}$, input state $|\psi_C\rangle=C|0^n\rangle$, with $C$ a random Clifford circuit. The target is to provide an estimation $\hat{\mu}(\vec{t})$ such that $\abs{\mu(\vec{t})-\hat{\mu}(\vec{t})}\leq\epsilon$.
\end{problem}

For two-dimensional constant-depth quantum circuits, Ref.~\cite{bravyi2021classical} demonstrated that the quantum mean value problem can be efficiently solved on a classical computer. However, extending this result to higher-dimensional lattices or to general Hamiltonian dynamics remains challenging. In the case of higher-dimensional lattices, all known classical simulation algorithms incur high computational overhead.
Assuming each quantum gate satisfies a local ``scrambling'' property, Ref.~\cite{angrisani2024classically} recently introduces a classical method for estimating quantum mean values on high-dimensional lattices. This method operates in time $\mathcal{O}(n^{\log (n/\epsilon)})$ with a success probability of at least $1 - 1/n$ when the quantum circuit depth $d={\rm poly}\log(n)$. Nevertheless, achieving a truly polynomial-time algorithm for Hamiltonian dynamics without the ``scrambling'' assumption of the quantum process remains a significant open question.
For general Hamiltonian dynamics, Ref.~\cite{wild2023classical} employs a cluster expansion technique to approximate quantum mean values for local observables up to a specified additive error. For global observables, the cluster expansion method may fail to converge, and extending such techniques to global observables thus remains unsolved. Finally, the most advanced quantum simulation algorithms for $D$-dimensional local Hamiltonians require circuit depth $\mathcal{O}(t\,\mathrm{poly}\log(nt/\epsilon))$~\cite{haah2021quantum}. As a result, classical evaluation of the local Heisenberg evolution $e^{iHt} O_j e^{-iHt}$ necessitates time $\mathcal{O}(2^{t^D\,\mathrm{poly}\log(nt/\epsilon)})$, leading to only quasi-polynomial-time classical algorithms. 
Therefore, it remains a formidable task to develop polynomial classical algorithms for the Hamiltonian dynamics problem.

Here, we present a novel method to resolve the challenges and provide the first polynomial classical algorithm for solving Problem~\ref{problem1}. 
First, we rewrite the Heisenberg picture time evolved operator $U_O(\vec{t})=U^{\dagger}(\vec{t})OU(\vec{t})$ onto the Pauli basis as $$U_O(\vec{t})=\sum\limits_{P\in\{I,X,Y,Z\}^{\otimes n}}{\rm Tr}\left[OU(\vec{t})PU^{\dagger}(\vec{t})\right]P/{2^n}.$$
This reframes the problem from evolving $O$ directly to evolving each Pauli string $P$. However, handling all $4^n$ Pauli terms remains intractable. Due to the interference of high-weight Pauli operators, we truncate the summation to Pauli strings of constant weight $|P|<k$ with $k=\mathcal{O}(1)$. Here, the Pauli weight $|P|$ represents the number of qubits on which $P$ acts non-trivially. As a result, the truncated operator $U_{\rm cut}=\sum_{|P|<k}{\rm Tr}\left[OU(\vec{t})PU^{\dagger}(\vec{t})\right]P/2^n$ serves as an estimator for $U_O(\vec{t})$ (see Fig.~\ref{fig:Main_Results}.~(b)). Since $P$ is a local operator, its Heisenberg evolution $U(\vec{t})PU^{\dagger}(\vec{t})$ can be approximated by an operator $V_P(\vec{t})=\sum_{|{\rm supp}(Q)|\leq M}\alpha_QQ$ with real coefficients $\alpha_Q\in\mathbb{R}$, such that $\|V_P(\vec{t})-U(\vec{t})PU^{\dagger}(\vec{t})\|\leq \epsilon^{\prime}$, and the support size is 
$M=\tilde{\mathcal{O}}\left(e^{tL\mathfrak{d}}\log\left[e^{tL\mathfrak{d}}/\epsilon^{\prime}\right]\right)$ as visualized in Fig.~\ref{fig:Main_Results}~(c) and (d). This approximation follows from the cluster expansion method~\cite{haah2024learning,wild2023classical}, which is essentially a multi-variable Taylor expansion to $U(\vec{t})PU^{\dagger}(\vec{t})$. The total number of terms within $V_P(\vec{t})$ is determined by the number of interaction graphs of size $M$ that are connected to $P$. One may view this as a tree‑generation process: starting at the root node $P$, each step expands only to its neighboring nodes. Since each node has only $\mathcal{O}(\mathfrak{d})$ connected neighbors for local Hamiltonians, the number of connected paths grows like $\mathcal{O}((L\mathfrak{d})^M)$, independently of the operator norm $\|H\|$. Consequently, the cost of estimating $V_P(\vec{t})$ scales as $\mathcal{O}(\exp(M))$, thereby yielding an efficient classical simulation approach whenever $t=\max_{l\in[L]}\{\abs{t_l}\}$ and $L$ are constant values.

The above observation implies that $U_{\rm cut}$ admits a classically efficient approximation of the form
\begin{align}
    V(\vec{t}) = \sum_{|P| < k} {\rm Tr}\left[ O V_P(\vec{t}) \right] P/2^n,
\end{align}
such that $\| V(\vec{t}) - U_{\rm cut}(\vec{t}) \| \leq \epsilon$, by setting $\epsilon' = \epsilon n^{-k}$. It remains to characterize the difference between $U_{\rm cut}$ and the true Hamiltonian dynamics operator $U_O(\vec{t})$. While their operator norm difference may not be tightly bounded, this norm quantifies the worst-case bias across all input states. However, most quantum states $|\psi\rangle$ incur significantly smaller deviations, as the expectation value of high-weight Pauli operators tends to be exponentially suppressed on typical states $|\psi\rangle$.
To capture this typical behavior, we consider the average-case deviation over the Clifford-state ensemble (or computational basis state ensemble):
\begin{align}
    {\rm dist}_{\psi}(U_{\rm cut}, U_O(\vec{t})) = \mathbb{E}_{|\psi\rangle \sim {\rm Cl}(2^n)} \left| \langle \psi | (U_{\rm cut} - U_O(\vec{t})) | \psi \rangle \right|.
\end{align}
When the observable $O$ is a low-rank projector, it can be shown that ${\rm dist}_{\psi}(U_{\rm cut}, U_O(\vec{t}))$ decays exponentially. As a result, for all but an exponentially small fraction of input states $|\psi\rangle$, we have
$
\left| \langle \psi | U_{\rm cut} | \psi \rangle - \langle \psi | U_O(\vec{t}) | \psi \rangle \right| \leq \epsilon$.
Moreover, in many noisy intermediate-scale quantum (NISQ) algorithms, one is often interested in the expectation values of Pauli operators. 
We prove that for random Pauli operators, the deviation satisfies ${\rm dist}_{\psi}(U_{\rm cut}, U_O(\vec{t})) \leq \epsilon$ with exponentially small failure probability over both the state and observable ensembles. Therefore, it indicates that the difference between $U_{\rm cut}$ and $U_O(\vec t)$ is indeed small for most cases. 
We summarize these findings as follows.



\begin{theorem}[Informal]
   There exists a classical algorithm that produces an estimation $\hat{\mu}(\vec{t})$ such that $\abs{\mu(\vec{t})-\hat{\mu}(\vec{t})}\leq\epsilon$ for Problem~\ref{problem1} in runtime $\mathcal{O}((e^{tL\mathfrak{d}}{\rm poly}(n)/\epsilon)^{e^{tL\mathfrak{d}}})$, with $t=\max\{|t_l|\}_{l=1}^L$, $\epsilon=1/{\rm poly}(n)$ and an exponentially small (${\rm poly}(n)/2^n$) failure probability.
    \label{theorem1}
\end{theorem}

The result demonstrates that any Hamiltonian dynamics with constant evolution time can be efficiently simulated on a classical computer with near-unit success probability.
Our results can be directly applied to simulating constant-depth analogue quantum computation~\cite{arute2019quantum,wu2021strong} and variational algorithms~\cite{peruzzo2014variational,huang2023efficient,wu2023orbital,farhi2014quantum,cerezo2021variational,zhang2022quantum}.
Meanwhile, due to the exponential increase of error mitigation cost with circuit depth~\cite{takagi2023universal, quek2024exponentially}, current noisy quantum hardware platforms—including trapped ions~\cite{smith2016many}, neutral atom arrays~\cite{evered2023high}, and superconducting qubit circuits~\cite{arute2019quantum,morvan2024phase}—can only implement constant-depth circuits. Our result thus implies that quantum computing may not provide exponential quantum advantage within the noisy intermediate-scale quantum regime. The result highlights the critical need for advanced quantum error-correction schemes to enable deep logical circuits, sharpening our understanding of quantum computing advantage in simulating Hamiltonian dynamics.

\vspace{0.2cm}
\noindent\textbf{Guided Local Hamiltonian Problem---}The local Hamiltonian (LH) problem~\cite{kempe2006complexity, gharibian2015quantum} is a central challenge in both physics and computer science. 
In physics, it is crucial for understanding low-energy phenomena such as superconductivity~\cite{wilson1983superconducting}, superfluidity~\cite{wheatley1975experimental}, and topological orders~\cite{wen1995topological, kane2005z}. In computer science, the LH problem has been proven to be QMA-complete for both synthetic~\cite{kitaev2002classical, kempe2006complexity} and physical~\cite{childs2014bose, o2021electronic} quantum systems. As the quantum analogue~\cite{aharonov2002quantum} of classical constraint-satisfaction problems, the LH problem plays a foundational role in quantum complexity theory, analogous to how the Cook–Levin theorem characterizes NP-completeness~\cite{karp2010reducibility}. Moreover, the quantum PCP conjecture~\cite{aharonov2013guest}—which posits that approximating the ground-state energy to constant precision remains QMA-hard—stands as a major open question.

The pursuit of efficiently solving the LH problem has driven much of quantum computing research. Notably, the problem becomes efficiently solvable on a quantum computer when a guiding state with non-trivial overlap with the ground state is provided~\cite{lin2020near, dong2022ground, lin2022heisenberg, wan2022randomized, wang2023quantum, ding2023even, ni2023low}. In theoretical computer science, this setting defines the guided local Hamiltonian (GLH) problem~\cite{gharibian2022dequantizing}.

\begin{problem}
    [Guided Local Hamiltonian Problem]
\label{problem2}
Consider an $n$-qubit local Hamiltonian $H=\sum_{X\subset S}\lambda_Xh_X$, let $E_0<E_1\leq\cdots\leq E_{2^n-1}$ be eigenvalues of $H$ with corresponding eigenstates $|\phi_0\rangle,|\phi_1\rangle,\cdots, |\phi_{2^n-1}\rangle$. Suppose 
the energy gap $\Delta=E_1-E_0$, and classical initial state $|\psi_c\rangle$ such that $p_0=\abs{\langle\psi_c|\phi_0\rangle}^2\geq\gamma$. The target is to provide an estimation $\hat{E}_0$ to the ground state energy $E_0$ such that $|\hat{E}_0-E_0|\leq\epsilon$.
\end{problem}

This problem has been proven to be BQP-complete when the guided-state overlap satisfies $\gamma \in (1/{\rm poly}(n),\, 1 - 1/{\rm poly}(n))$ and the required accuracy is $\epsilon = 1/{\rm poly}(n)$~\cite{cade2022improved,cade2022complexity}, thereby characterizing the ultimate computational capability of quantum computers. These results suggest the potential for exponential quantum advantage in systems with accessible guiding states, assuming ${\rm BPP} \neq {\rm BQP}$. However, this advantage can disappear under relaxed conditions. Specifically, if the guiding state is classically tractable and the estimation is performed in the \emph{constant-relative-accuracy} regime—i.e., with error bounded by $\epsilon\|H\|$ with $\epsilon=\mathcal{O}(1)$—then the GLH problem becomes classically solvable~\cite{gharibian2022dequantizing, gall2024classical}. This follows from dequantized algorithms inspired by recent near-optimal quantum methods~\cite{lin2020near, gilyen2019quantum}. Nonetheless, constant-relative-accuracy is often impractical due to the typically large operator norm $\|H\| = {\rm poly}(n)$ in physical systems, even for the simplest Ising model. Designing efficient classical algorithms for the GLH problem under \emph{constant absolute-accuracy} $\epsilon$ remains an open challenge.


Here, we address this challenge by establishing a direct connection between the quantum dynamics mean value problem and the GLH problem through a dequantization of a tailored quantum algorithm.   Specifically, we consider a filter function constructed from the guiding state $|\psi_c\rangle$ (see Fig.~\ref{fig:Main_Results}.~(e)):
\begin{align}\label{Eq:Cxmain}
    C(x)=(F_\sigma*P)(x)=\sum\limits_{j=0}^{2^n-1}p_jF_{\sigma}(x-E_j)\approx\int_{-T}^{T}\hat{F}_{\sigma}(t)\langle\psi_c|e^{-i2\pi t(H-xI)}|\psi_c\rangle dt,
\end{align}
where $P(x)=\sum_{j=0}^{2^n-1}p_j\delta(x -E_j)$ is the spectral function of the guiding state $|\psi_c\rangle$ with $p_j=\abs{\langle\phi_j|\psi_c\rangle}^2$, and $\delta(\cdot)$ is the Dirac delta function. The Gaussian derivative filter function is $F_{\sigma}(t)=-te^{-t^2/(2\sigma^2)}/(\sigma^3\sqrt{2\pi})$ and $\hat{F}_{\sigma}(t)$ represents its Fourier transform with $\sigma=\mathcal{O}(\Delta/\log(\Delta\epsilon^{-1}\gamma^{-1}))$. Although the filter function $C(x)$ is formally defined over $x\in\mathbb{R}$, its concentration around $E_0$ allows us to truncate the integral to $T=\mathcal{O}(\sqrt{\ln(1/\epsilon_1)}/\Delta)$, for which the deviation from the original function $C(x)$ can be bounded by $\epsilon_1=\mathcal{O}(\gamma\epsilon\sigma^{-3})$~\cite{wang2023quantum}. Since the Gaussian derivative filter has an exponentially-decaying tail, $(F_{\sigma}*P)(x)$ is dominated by $\gamma F_{\sigma}(x-E_0)$ in the vicinity of $E_0$, which decays monotonically to zero when $x\leq E_0$ approaches $E_0$ (see Fig.~\ref{fig:Main_Results}.~(e)). Hence, our scheme proceeds by sampling $C(x)$ beginning at $x=E_a$, a lower bound to $E_0$, and an interval of $\varepsilon$. The algorithm outputs the estimation $\hat{E}_0$ when $C(x)$ decays below the termination threshold $\epsilon_1/2$.


Whether the above quantum algorithm can be dequantized depends on the ability to efficiently evaluate the Loschmidt echo 
\(\langle \psi_c | e^{-iHt} | \psi_c \rangle\) in Eq.~\eqref{Eq:Cxmain} for evolution time \(t \leq T\). In the worst case, estimating the Loschmidt echo for \(t = \mathcal{O}(n)\) is BQP-complete, while efficient computation is only known for the regime \(|t| < t^*\), where \(t^*\) denotes a constant threshold, using the cluster expansion method~\cite{wild2023classical}. In this work, we demonstrate that the Loschmidt echo can in fact be computed classically for any constant \(t\).  

The key idea is to consider the vacuum state $|\Omega\rangle=|0^n\rangle$ of the Hamiltonian $H$, which satisfies $e^{-iHt}|\Omega\rangle=|\Omega\rangle$ due to particle number symmetry preservation. Now, consider a superposed input state $ |\psi_1\rangle=\frac{1}{\sqrt{2}}\left(|\Omega\rangle+|\psi_c\rangle\right)$ and observables $O_1=|\psi_c\rangle\langle\Omega|$ and $O_2=|\Omega\rangle\langle\psi_c|$, the Loschmidt echo can be expressed as (see Fig.~\ref{fig:Main_Results}.~(f)):
\begin{eqnarray}
  \begin{split}
        {\rm Re}\left[\langle\psi_c|e^{-iHt}|\psi_c\rangle\right]&=\langle\psi_1|e^{iHt}{\left(O_1+O_2\right)}e^{-iHt}|\psi_1\rangle,\\
        {\rm Im}\left[\langle\psi_c|e^{-iHt}|\psi_c\rangle\right]&=i\langle\psi_1|e^{iHt}{\left(O_2-O_1\right)}e^{-iHt}|\psi_1\rangle,
  \end{split}
\end{eqnarray}
As a result, simulating $\langle\psi_1|e^{iHt}O_1e^{-iHt}|\psi_1\rangle$ and $\langle\psi_1|e^{iHt}O_2e^{-iHt}|\psi_1\rangle$ suffices to obtain the the Loschmidt echo $\langle\psi_c|e^{-iHt}|\psi_c\rangle$. Since the input state $|\psi_1\rangle$ is the linear combinations of computational basis states drawn from a Clifford ensemble, as a result, we can apply the proposed classical simulation algorithm to compute the mean value for each involved computational basis, and one may obtain $\langle\psi_c|e^{-iHt}|\psi_c\rangle$ with high success probability according to Theorem~\ref{theorem1}.

\begin{theorem}[Informal]
Suppose an $R={\rm poly}(n)$-configurational classical guiding state $\ket{\psi_c}$ is given, which has a constant overlap to the ground state of $H$. Then there exists a polynomial classical algorithm to solve the GLH problem, which outputs an approximation to the ground state energy $E_0$ within a constant additive error and an exponentially small (${\rm poly}(n)/2^n$) failure probability.
\label{theorem:dequantize}
\end{theorem}

This result establishes a fundamental connection between the simulation of quantum dynamics and the solution of the GLH problem, which is expected to have independent applications in other domains. Furthermore, it suggests that existing quantum algorithms~\cite{lin2022heisenberg,wang2023quantum} for the GLH problem do not inherently guarantee exponential quantum advantage when the accuracy requirement is $\epsilon=\mathcal{O}(1)$. The key to realizing such an advantage lies in preparing a nontrivial initial state that does not admit an efficient classical representation. This necessitates the exploration of alternative quantum algorithms, such as adiabatic state preparation~\cite{ge2016rapid,crosson2021prospects} and engineered dissipation~\cite{mi2024stable,harrington2022engineered}, to evaluate their potential for quantum speedups. Consequently, this consideration refines the requirements for achieving quantum computational advantage in LH problems.

Moreover, if the Hamiltonian is restricted to a 2D architecture, the success probability can be boosted to unity. Suppose $|\psi_c\rangle$ has configurations $|\bm j\rangle$, and denote $|\bm j\rangle\langle\Omega|=O_1\otimes \cdots \otimes O_n$, $U_j(t)=e^{iHt}O_je^{-iHt}$ for $j\in[n]$, then our first step aims to approximate $U_j(t)$ by $V_j(t)$ such that $\|U_j(t)-V_j(t)\|\leq \mathcal{O}(\epsilon/2n)$. Here, the operator $V_j(t)$ is a given by the cluster expansion, which nontrivial act on at most $M=\mathcal{O}(e^{\mathfrak{d}t}\log(2ne^{\mathfrak{d}t}/\epsilon))$ qubits, as a result, the mean value can be approximated by $\langle\psi_1|V_1(t)\cdots V_n(t)|\psi_1\rangle$ within $\epsilon/2$-additive error. The second step applies the causality principle and the support of $V_j(t)$ to assign $\{V_j(t)\}_{j=1}^n$ into two different groups, which are denoted by $V(R_1)$ and $V(R_2)$, where regions $R_1$ and $R_2$ are disjoint and $R_1\cup R_2$ compose the whole system. This method is first studied in Ref.~\cite{bravyi2021classical} to simulate constant 2D digital quantum circuits. It is shown that each region ($R_1$ or $R_2$) consists of $\sqrt{n}/4M$ sub-regions which are separated by $\geq 2M$ distance. This property enables operators $V(R_1)$ and $V(R_2)$ are easy to simulate classically, and the quantum dynamics mean value has the form $\hat{\mu}(t)=\langle\psi_1|V(R_1)V(R_2)|\psi_1\rangle$. Then the classical Monte Carlo algorithm can be used to approximate $\hat{\mu}(t)$ in $\mathcal{O}(1/\epsilon^2)$ running time, such that $\abs{\hat{\mu}(t)-\mu(t)}\leq\epsilon$. 

\begin{theorem}[Informal]
Suppose a 2D geometrically local Hamiltonian $H$, and
an $R={\rm poly}(n)$-configurational classical guiding state $\ket{\psi_c}$ which has a constant overlap to the ground state of $H$. There exists a classical algorithm to solve the GLH problem with quasi-polynomial runtime, 
which outputs an approximation to the ground state energy $E_0$ within a constant additive error.
\end{theorem}

Here, the quasi-polynomial complexity $\mathcal{O}(n^{\log n})$ arises from sampling over operators $V(R_1)$ and $V(R_2)$. Suppose $|\psi_1\rangle=\sum_{\bm j}a_j|\bm j\rangle$, then we can write $\langle\psi_1|V(R_1)V(R_2)|\psi_1\rangle=\sum_{\bm j,\bm j^{\prime}}a_ja_{j^{\prime}}^*\langle\bm j^{\prime}|V(R_1)V(R_2)|\bm j\rangle$, where each term $\langle \bm j^{\prime}|V(R_1)V(R_2)|\bm j\rangle$ can be simulated by using the classical Monte Carlo method. We write $\langle\bm j^{\prime}|V(R_1)V(R_2)|\bm j\rangle=\sum_x\abs{\langle\bm j^{\prime}|V(R_1)|x\rangle}^2\langle x|V(R_2)|\bm j\rangle/\langle x|V^{\dagger}(R_1)|\bm j^{\prime}\rangle$, and the amplitude $\langle\bm j^{\prime}|V(R_q)|x\rangle=\prod_{s}\langle \bm j^{\prime}_{{R_q(s)}}|V(R_q(s))|x_{R_q(s)}\rangle$ by the lightcone argument, where $V(R_q(s))$ is a quasi-1D operator acting on a subsystem of size $\sqrt{n}\times M$, indexes $q\in\{1,2\}$ and $s\in\{1,\cdots,[\sqrt{n}/M]\}$. The local amplitude $\langle \bm j^{\prime}_{R_q(s)}|V(R_q(s))|x_{R_q(s)}\rangle$ can be estimated using a row-by-row sampling strategy---each time one samples only the $M$ qubits within an $M\times M$-sized 2D local window. By sliding this window across the $\sqrt{n}\times M$-sized subsystem and repeating the procedure $\sqrt{n}/M$ times, an estimate of the local amplitude is obtained. The total running time is $\sqrt{n}2^{\mathcal{O}(M^2)}/M$, a quasi-polynomial complexity when $M=\mathcal{O}(\log n)$.



\vspace{0.2cm}
\noindent\textbf{Discussion and Outlook---}In this work, we introduce a polynomial-time classical algorithm for simulating constant-time quantum dynamics, where the underlying Hamiltonian could be a general local Hamiltonian without geometry restrictions. We rigorously prove that for any local Hamiltonian $H$ and constant evolution time $t$, our algorithm efficiently achieves a small error $\epsilon=1/{\rm poly}(n)$ in estimating the quantum mean value $\langle\psi|e^{iHt}Oe^{-iHt}|\psi\rangle$ on all global Pauli observables $O$ and stabilizer states $|\psi\rangle$ except for an exponentially small fraction. This implies that the quantum advantage of near-term quantum algorithms—such as VQE, QAOA, and QML algorithms based on Hamiltonian variational ansatz—should be reconsidered, even in noise-free quantum circuits. This extends the quantum advantage boundary from constant noise to the noiseless regime~\cite{stilck2021limitations,aharonov2022polynomial,shao2024simulatingnoisyvariationalquantum}.
Our results further demonstrate a close connection between the quantum mean value problem and the GLH problem, showing that the ground-state energy of a constant-gapped local Hamiltonian can be efficiently computed by a classical algorithm with high probability, provided that both the required accuracy and the guided-state overlap are constants. 
Our result thus indicates the importance of preparing a nontrivial guiding state via other quantum algorithms.

This work leaves substantial opportunities for further research. First, the results demonstrate that the quantum mean value $\langle\psi|e^{iHt}Oe^{-iHt}|\psi\rangle$ can be efficiently simulated by a classical algorithm in the average case; however, the possibility of classically intractable worst-case instances is not excluded. An important direction for future investigation is whether the proposed algorithm can be extended to address worst-case scenarios, or how to identify the exponentially small subsets of states and observables for which classical simulation becomes inefficient.
Additionally, we have shown that the GLH problem, when supplied with a favorable initial state, can be solved in polynomial time under the assumptions of a constant spectral gap and a constant precision requirement. The quantum PCP conjecture asserts that, in the worst case, deciding whether the ground‑state energy of a local Hamiltonian lies within a constant promise gap is QMA‑hard. It remains an open question whether identifying such a “good” initial state is itself QMA‑hard for general quantum systems, or whether the conjecture does not apply to constant‑gap Hamiltonians. Finally, recent studies have also demonstrated that expectation values on noisy quantum circuits with deep circuit depths are classically simulable~\cite{shao2024simulatingnoisyvariationalquantum,schuster2024polynomialtimeclassicalalgorithmnoisy}. These findings highlight the theoretical limitations of near-term quantum computation in the NISQ era. However, the proposed classical algorithms are efficient only in the asymptotic sense—valid for large $n$ and constant error rates—yet they still incur large exponents, leaving room for practical quantum advantages with polynomial speedups.  
Moreover, our work considers expectation values in the absence of mid-circuit measurements and feedforward operations, which are essential for fault-tolerant quantum computing. How to classically simulate quantum processes with mid-circuit measurements and feedforward operations remains an intriguing open question.

\vspace{8pt}

\section*{Data availability}
\noindent 
No datasets were generated or analysed during the current study.

\section*{Competing interests}
\noindent 
The authors declare no competing interests.

\section*{Acknowledgments}
\noindent The authors would like to express gratitude to Xiaoming Zhang, Jinzhao Sun, Zhongxia Shang, Bujiao Wu and Jingbo B. Wang for valuable discussions. This work is supported by the Innovation Program for Quantum Science and Technology (Grant No.~2023ZD0300200), the National Natural Science Foundation of China Grant (No.~12175003, No.~12361161602 and  No.~62131002),  NSAF (Grant No.~U2330201). 
\clearpage
\bibliography{main.bbl}

\begin{thebibliography}{10}

\bibitem{janzing2005ergodic}
Dominik Janzing and Pawel Wocjan.
\newblock Ergodic quantum computing.
\newblock {\em Quantum Information Processing}, 4(2):129--158, 2005.

\bibitem{gharibian2022improved}
Sevag Gharibian, Ryu Hayakawa, Fran{\c{c}}ois~Le Gall, and Tomoyuki Morimae.
\newblock Improved hardness results for the guided local hamiltonian problem.
\newblock {\em arXiv preprint arXiv:2207.10250}, 2022.

\bibitem{cade2022complexity}
Chris Cade, Marten Folkertsma, and Jordi Weggemans.
\newblock Complexity of the guided local hamiltonian problem: Improved parameters and extension to excited states.
\newblock {\em arXiv preprint arXiv:2207.10097}, 2022.

\bibitem{bravyi2021classical}
Sergey Bravyi, David Gosset, and Ramis Movassagh.
\newblock Classical algorithms for quantum mean values.
\newblock {\em Nature Physics}, 17(3):337--341, 2021.

\bibitem{schollwock2005density}
Ulrich Schollw{\"o}ck.
\newblock The density-matrix renormalization group.
\newblock {\em Reviews of modern physics}, 77(1):259--315, 2005.

\bibitem{schollwock2011density}
Ulrich Schollw{\"o}ck.
\newblock The density-matrix renormalization group in the age of matrix product states.
\newblock {\em Annals of physics}, 326(1):96--192, 2011.

\bibitem{orus2014practical}
Rom{\'a}n Or{\'u}s.
\newblock A practical introduction to tensor networks: Matrix product states and projected entangled pair states.
\newblock {\em Annals of physics}, 349:117--158, 2014.

\bibitem{bridgeman2017hand}
Jacob~C Bridgeman and Christopher~T Chubb.
\newblock Hand-waving and interpretive dance: an introductory course on tensor networks.
\newblock {\em Journal of physics A: Mathematical and theoretical}, 50(22):223001, 2017.

\bibitem{eisert2010colloquium}
Jens Eisert, Marcus Cramer, and Martin~B Plenio.
\newblock Colloquium: Area laws for the entanglement entropy.
\newblock {\em Reviews of modern physics}, 82(1):277--306, 2010.

\bibitem{anshu2022area}
Anurag Anshu, Itai Arad, and David Gosset.
\newblock An area law for 2d frustration-free spin systems.
\newblock In {\em Proceedings of the 54th Annual ACM SIGACT Symposium on Theory of Computing}, pages 12--18, 2022.

\bibitem{cade2022improved}
Chris Cade, Marten Folkertsma, Sevag Gharibian, Ryu Hayakawa, Fran{\c{c}}ois~Le Gall, Tomoyuki Morimae, and Jordi Weggemans.
\newblock Improved hardness results for the guided local hamiltonian problem.
\newblock {\em arXiv preprint arXiv:2207.10250}, 2022.

\bibitem{gharibian2022dequantizing}
Sevag Gharibian and Fran{\c{c}}ois Le~Gall.
\newblock Dequantizing the quantum singular value transformation: hardness and applications to quantum chemistry and the quantum pcp conjecture.
\newblock In {\em Proceedings of the 54th Annual ACM SIGACT Symposium on Theory of Computing}, pages 19--32, 2022.

\bibitem{wild2023classical}
Dominik~S Wild and {\'A}lvaro~M Alhambra.
\newblock Classical simulation of short-time quantum dynamics.
\newblock {\em PRX Quantum}, 4(2):020340, 2023.

\bibitem{lin2022heisenberg}
Lin Lin and Yu~Tong.
\newblock Heisenberg-limited ground-state energy estimation for early fault-tolerant quantum computers.
\newblock {\em PRX Quantum}, 3(1):010318, 2022.

\bibitem{wang2023quantum}
Guoming Wang, Daniel~Stilck Fran{\c{c}}a, Ruizhe Zhang, Shuchen Zhu, and Peter~D Johnson.
\newblock Quantum algorithm for ground state energy estimation using circuit depth with exponentially improved dependence on precision.
\newblock {\em Quantum}, 7:1167, 2023.

\bibitem{takagi2023universal}
Ryuji Takagi, Hiroyasu Tajima, and Mile Gu.
\newblock Universal sampling lower bounds for quantum error mitigation.
\newblock {\em Physical Review Letters}, 131(21):210602, 2023.

\bibitem{quek2024exponentially}
Yihui Quek, Daniel Stilck~Fran{\c{c}}a, Sumeet Khatri, Johannes~Jakob Meyer, and Jens Eisert.
\newblock Exponentially tighter bounds on limitations of quantum error mitigation.
\newblock {\em Nature Physics}, 20(10):1648--1658, 2024.

\bibitem{haah2024learning}
Jeongwan Haah, Robin Kothari, and Ewin Tang.
\newblock Learning quantum hamiltonians from high-temperature gibbs states and real-time evolutions.
\newblock {\em Nature Physics}, pages 1--5, 2024.

\bibitem{bakshi2024high}
Ainesh Bakshi, Allen Liu, Ankur Moitra, and Ewin Tang.
\newblock High-temperature gibbs states are unentangled and efficiently preparable.
\newblock In {\em 2024 IEEE 65th Annual Symposium on Foundations of Computer Science (FOCS)}, pages 1027--1036. IEEE, 2024.

\bibitem{angrisani2024classically}
Armando Angrisani, Alexander Schmidhuber, Manuel~S Rudolph, M~Cerezo, Zo{\"e} Holmes, and Hsin-Yuan Huang.
\newblock Classically estimating observables of noiseless quantum circuits.
\newblock {\em arXiv preprint arXiv:2409.01706}, 2024.

\bibitem{haah2021quantum}
Jeongwan Haah, Matthew~B Hastings, Robin Kothari, and Guang~Hao Low.
\newblock Quantum algorithm for simulating real time evolution of lattice hamiltonians.
\newblock {\em SIAM Journal on Computing}, 52(6):FOCS18--250, 2021.

\bibitem{arute2019quantum}
Frank Arute, Kunal Arya, Ryan Babbush, Dave Bacon, Joseph~C Bardin, Rami Barends, Rupak Biswas, Sergio Boixo, Fernando~GSL Brandao, David~A Buell, et~al.
\newblock Quantum supremacy using a programmable superconducting processor.
\newblock {\em Nature}, 574(7779):505--510, 2019.

\bibitem{wu2021strong}
Yulin Wu, Wan-Su Bao, Sirui Cao, Fusheng Chen, Ming-Cheng Chen, Xiawei Chen, Tung-Hsun Chung, Hui Deng, Yajie Du, Daojin Fan, et~al.
\newblock Strong quantum computational advantage using a superconducting quantum processor.
\newblock {\em Physical review letters}, 127(18):180501, 2021.

\bibitem{peruzzo2014variational}
Alberto Peruzzo, Jarrod McClean, Peter Shadbolt, Man-Hong Yung, Xiao-Qi Zhou, Peter~J Love, Al{\'a}n Aspuru-Guzik, and Jeremy~L O’brien.
\newblock A variational eigenvalue solver on a photonic quantum processor.
\newblock {\em Nature Communications}, 5(1):1--7, 2014.

\bibitem{huang2023efficient}
Yifei Huang, Yuguo Shao, Weiluo Ren, Jinzhao Sun, and Dingshun Lv.
\newblock Efficient quantum imaginary time evolution by drifting real-time evolution: an approach with low gate and measurement complexity.
\newblock {\em Journal of Chemical Theory and Computation}, 19(13):3868--3876, 2023.

\bibitem{wu2023orbital}
Yusen Wu, Zigeng Huang, Jinzhao Sun, Xiao Yuan, Jingbo~B Wang, and Dingshun Lv.
\newblock Orbital expansion variational quantum eigensolver.
\newblock {\em Quantum Science and Technology}, 2023.

\bibitem{farhi2014quantum}
Edward Farhi, Jeffrey Goldstone, and Sam Gutmann.
\newblock A quantum approximate optimization algorithm.
\newblock {\em arXiv preprint arXiv:1411.4028}, 2014.

\bibitem{cerezo2021variational}
Marco Cerezo, Andrew Arrasmith, Ryan Babbush, Simon~C Benjamin, Suguru Endo, Keisuke Fujii, Jarrod~R McClean, Kosuke Mitarai, Xiao Yuan, Lukasz Cincio, et~al.
\newblock Variational quantum algorithms.
\newblock {\em Nature Reviews Physics}, 3(9):625--644, 2021.

\bibitem{zhang2022quantum}
Yukun Zhang, Yifei Huang, Jinzhao Sun, Dingshun Lv, and Xiao Yuan.
\newblock Quantum computing quantum monte carlo.
\newblock {\em arXiv preprint arXiv:2206.10431}, 2022.

\bibitem{smith2016many}
Jacob Smith, Aaron Lee, Philip Richerme, Brian Neyenhuis, Paul~W Hess, Philipp Hauke, Markus Heyl, David~A Huse, and Christopher Monroe.
\newblock Many-body localization in a quantum simulator with programmable random disorder.
\newblock {\em Nature Physics}, 12(10):907--911, 2016.

\bibitem{evered2023high}
Simon~J Evered, Dolev Bluvstein, Marcin Kalinowski, Sepehr Ebadi, Tom Manovitz, Hengyun Zhou, Sophie~H Li, Alexandra~A Geim, Tout~T Wang, Nishad Maskara, et~al.
\newblock High-fidelity parallel entangling gates on a neutral-atom quantum computer.
\newblock {\em Nature}, 622(7982):268--272, 2023.

\bibitem{morvan2024phase}
Alexis Morvan, B~Villalonga, X~Mi, S~Mandra, A~Bengtsson, PV~Klimov, Z~Chen, S~Hong, C~Erickson, IK~Drozdov, et~al.
\newblock Phase transitions in random circuit sampling.
\newblock {\em Nature}, 634(8033):328--333, 2024.

\bibitem{kempe2006complexity}
Julia Kempe, Alexei Kitaev, and Oded Regev.
\newblock The complexity of the local hamiltonian problem.
\newblock {\em Siam journal on computing}, 35(5):1070--1097, 2006.

\bibitem{gharibian2015quantum}
Sevag Gharibian, Yichen Huang, Zeph Landau, Seung~Woo Shin, et~al.
\newblock Quantum hamiltonian complexity.
\newblock {\em Foundations and Trends{\textregistered} in Theoretical Computer Science}, 10(3):159--282, 2015.

\bibitem{wilson1983superconducting}
Martin~N Wilson.
\newblock Superconducting magnets.
\newblock 1983.

\bibitem{wheatley1975experimental}
John~C Wheatley.
\newblock Experimental properties of superfluid he 3.
\newblock {\em Reviews of modern physics}, 47(2):415, 1975.

\bibitem{wen1995topological}
Xiao-Gang Wen.
\newblock Topological orders and edge excitations in fractional quantum hall states.
\newblock {\em Advances in Physics}, 44(5):405--473, 1995.

\bibitem{kane2005z}
Charles~L Kane and Eugene~J Mele.
\newblock Z 2 topological order and the quantum spin hall effect.
\newblock {\em Physical review letters}, 95(14):146802, 2005.

\bibitem{kitaev2002classical}
Alexei~Yu Kitaev, Alexander Shen, and Mikhail~N Vyalyi.
\newblock {\em Classical and quantum computation}.
\newblock Number~47. American Mathematical Soc., 2002.

\bibitem{childs2014bose}
Andrew~M Childs, David Gosset, and Zak Webb.
\newblock The bose-hubbard model is qma-complete.
\newblock In {\em Automata, Languages, and Programming: 41st International Colloquium, ICALP 2014, Copenhagen, Denmark, July 8-11, 2014, Proceedings, Part I 41}, pages 308--319. Springer, 2014.

\bibitem{o2021electronic}
Bryan O'Gorman, Sandy Irani, James Whitfield, and Bill Fefferman.
\newblock Electronic structure in a fixed basis is qma-complete.
\newblock {\em arXiv preprint arXiv:2103.08215}, 2021.

\bibitem{aharonov2002quantum}
Dorit Aharonov and Tomer Naveh.
\newblock Quantum np-a survey.
\newblock {\em arXiv preprint quant-ph/0210077}, 2002.

\bibitem{karp2010reducibility}
Richard~M Karp.
\newblock {\em Reducibility among combinatorial problems}.
\newblock Springer, 2010.

\bibitem{aharonov2013guest}
Dorit Aharonov, Itai Arad, and Thomas Vidick.
\newblock Guest column: the quantum pcp conjecture.
\newblock {\em Acm sigact news}, 44(2):47--79, 2013.

\bibitem{lin2020near}
Lin Lin and Yu~Tong.
\newblock Near-optimal ground state preparation.
\newblock {\em Quantum}, 4:372, 2020.

\bibitem{dong2022ground}
Yulong Dong, Lin Lin, and Yu~Tong.
\newblock Ground-state preparation and energy estimation on early fault-tolerant quantum computers via quantum eigenvalue transformation of unitary matrices.
\newblock {\em PRX Quantum}, 3(4):040305, 2022.

\bibitem{wan2022randomized}
Kianna Wan, Mario Berta, and Earl~T Campbell.
\newblock Randomized quantum algorithm for statistical phase estimation.
\newblock {\em Physical Review Letters}, 129(3):030503, 2022.

\bibitem{ding2023even}
Zhiyan Ding and Lin Lin.
\newblock Even shorter quantum circuit for phase estimation on early fault-tolerant quantum computers with applications to ground-state energy estimation.
\newblock {\em PRX Quantum}, 4(2):020331, 2023.

\bibitem{ni2023low}
Hongkang Ni, Haoya Li, and Lexing Ying.
\newblock On low-depth algorithms for quantum phase estimation.
\newblock {\em Quantum}, 7:1165, 2023.

\bibitem{gall2024classical}
Fran{\c{c}}ois~Le Gall.
\newblock Classical algorithms for constant approximation of the ground state energy of local hamiltonians.
\newblock {\em arXiv preprint arXiv:2410.21833}, 2024.

\bibitem{gilyen2019quantum}
Andr{\'a}s Gily{\'e}n, Yuan Su, Guang~Hao Low, and Nathan Wiebe.
\newblock Quantum singular value transformation and beyond: exponential improvements for quantum matrix arithmetics.
\newblock In {\em Proceedings of the 51st Annual ACM SIGACT Symposium on Theory of Computing}, pages 193--204, 2019.

\bibitem{ge2016rapid}
Yimin Ge, Andr{\'a}s Moln{\'a}r, and J~Ignacio Cirac.
\newblock Rapid adiabatic preparation of injective projected entangled pair states and gibbs states.
\newblock {\em Physical review letters}, 116(8):080503, 2016.

\bibitem{crosson2021prospects}
EJ~Crosson and DA~Lidar.
\newblock Prospects for quantum enhancement with diabatic quantum annealing.
\newblock {\em Nature Reviews Physics}, 3(7):466--489, 2021.

\bibitem{mi2024stable}
Xiao Mi, AA~Michailidis, Sara Shabani, KC~Miao, PV~Klimov, J~Lloyd, E~Rosenberg, R~Acharya, I~Aleiner, TI~Andersen, et~al.
\newblock Stable quantum-correlated many-body states through engineered dissipation.
\newblock {\em Science}, 383(6689):1332--1337, 2024.

\bibitem{harrington2022engineered}
Patrick~M Harrington, Erich~J Mueller, and Kater~W Murch.
\newblock Engineered dissipation for quantum information science.
\newblock {\em Nature Reviews Physics}, 4(10):660--671, 2022.

\bibitem{stilck2021limitations}
Daniel Stilck~Fran{\c{c}}a and Raul Garcia-Patron.
\newblock Limitations of optimization algorithms on noisy quantum devices.
\newblock {\em Nature Physics}, 17(11):1221--1227, 2021.

\bibitem{aharonov2022polynomial}
Dorit Aharonov, Xun Gao, Zeph Landau, Yunchao Liu, and Umesh Vazirani.
\newblock A polynomial-time classical algorithm for noisy random circuit sampling.
\newblock {\em arXiv preprint arXiv:2211.03999}, 2022.

\bibitem{shao2024simulatingnoisyvariationalquantum}
Yuguo Shao, Fuchuan Wei, Song Cheng, and Zhengwei Liu.
\newblock Simulating noisy variational quantum algorithms: A polynomial approach, 2024.

\bibitem{schuster2024polynomialtimeclassicalalgorithmnoisy}
Thomas Schuster, Chao Yin, Xun Gao, and Norman~Y. Yao.
\newblock A polynomial-time classical algorithm for noisy quantum circuits, 2024.

\bibitem{terhal2002adaptive}
Barbara~M Terhal and David~P DiVincenzo.
\newblock Adaptive quantum computation, constant depth quantum circuits and arthur-merlin games.
\newblock {\em arXiv preprint quant-ph/0205133}, 2002.

\bibitem{hastings2005quasiadiabatic}
Matthew~B Hastings and Xiao-Gang Wen.
\newblock Quasiadiabatic continuation of quantum states: The stability of topological ground-state degeneracy and emergent gauge invariance.
\newblock {\em Physical Review B—Condensed Matter and Materials Physics}, 72(4):045141, 2005.

\bibitem{osborne2007simulating}
Tobias~J Osborne.
\newblock Simulating adiabatic evolution of gapped spin systems.
\newblock {\em Physical Review A—Atomic, Molecular, and Optical Physics}, 75(3):032321, 2007.

\bibitem{chen2010local}
Xie Chen, Zheng-Cheng Gu, and Xiao-Gang Wen.
\newblock Local unitary transformation, long-range quantum entanglement, wave function renormalization, and topological order.
\newblock {\em Physical Review B—Condensed Matter and Materials Physics}, 82(15):155138, 2010.

\bibitem{lieb1972finite}
Elliott~H Lieb and Derek~W Robinson.
\newblock The finite group velocity of quantum spin systems.
\newblock {\em Communications in mathematical physics}, 28(3):251--257, 1972.

\bibitem{hastings2010locality}
Matthew~B Hastings.
\newblock Locality in quantum systems.
\newblock {\em Quantum Theory from Small to Large Scales}, 95:171--212, 2010.

\bibitem{childs2021theory}
Andrew~M Childs, Yuan Su, Minh~C Tran, Nathan Wiebe, and Shuchen Zhu.
\newblock Theory of trotter error with commutator scaling.
\newblock {\em Physical Review X}, 11(1):011020, 2021.

\bibitem{basso2022performance}
Joao Basso, David Gamarnik, Song Mei, and Leo Zhou.
\newblock Performance and limitations of the qaoa at constant levels on large sparse hypergraphs and spin glass models.
\newblock In {\em 2022 IEEE 63rd Annual Symposium on Foundations of Computer Science (FOCS)}, pages 335--343. IEEE, 2022.

\bibitem{anshu2023concentration}
Anurag Anshu and Tony Metger.
\newblock Concentration bounds for quantum states and limitations on the qaoa from polynomial approximations.
\newblock {\em Quantum}, 7:999, 2023.

\bibitem{low2018hamiltonian}
Guang~Hao Low and Nathan Wiebe.
\newblock Hamiltonian simulation in the interaction picture.
\newblock {\em arXiv preprint arXiv:1805.00675}, 2018.

\bibitem{cerezo2021cost}
Marco Cerezo, Akira Sone, Tyler Volkoff, Lukasz Cincio, and Patrick~J Coles.
\newblock Cost function dependent barren plateaus in shallow parametrized quantum circuits.
\newblock {\em Nature Communications}, 12(1):1--12, 2021.

\bibitem{bjorklund2008fast}
Andreas Bj{\"o}rklund, Thore Husfeldt, Petteri Kaski, and Mikko Koivisto.
\newblock The fast intersection transform with applications to counting paths.
\newblock {\em arXiv preprint arXiv:0809.2489}, 2008.

\bibitem{quek2022exponentially}
Yihui Quek, Stilck~Daniel Franc, Sumeet Khatri, Jakob~Johannes Meyer, and Jens Eisert.
\newblock Exponentially tighter bounds on limitations of quantum error mitigation.
\newblock {\em arXiv preprint arXiv:2210.11505}, 2022.

\bibitem{wang2021noise}
Samson Wang, Enrico Fontana, Marco Cerezo, Kunal Sharma, Akira Sone, Lukasz Cincio, and Patrick~J Coles.
\newblock Noise-induced barren plateaus in variational quantum algorithms.
\newblock {\em Nature communications}, 12(1):1--11, 2021.

\bibitem{haah2016sample}
Jeongwan Haah, Aram~W Harrow, Zhengfeng Ji, Xiaodi Wu, and Nengkun Yu.
\newblock Sample-optimal tomography of quantum states.
\newblock In {\em Proceedings of the forty-eighth annual ACM symposium on Theory of Computing}, pages 913--925, 2016.

\bibitem{setia2019superfast}
Kanav Setia, Sergey Bravyi, Antonio Mezzacapo, and James~D Whitfield.
\newblock Superfast encodings for fermionic quantum simulation.
\newblock {\em Physical Review Research}, 1(3):033033, 2019.

\bibitem{stanisic2022observing}
Stasja Stanisic, Jan~Lukas Bosse, Filippo~Maria Gambetta, Raul~A Santos, Wojciech Mruczkiewicz, Thomas~E O’Brien, Eric Ostby, and Ashley Montanaro.
\newblock Observing ground-state properties of the fermi-hubbard model using a scalable algorithm on a quantum computer.
\newblock {\em Nature communications}, 13(1):5743, 2022.

\bibitem{montorsi2012nonlocal}
Arianna Montorsi and Marco Roncaglia.
\newblock Nonlocal order parameters for the 1d hubbard model.
\newblock {\em Physical review letters}, 109(23):236404, 2012.

\bibitem{barbiero2013hidden}
Luca Barbiero, Arianna Montorsi, and Marco Roncaglia.
\newblock How hidden orders generate gaps in one-dimensional fermionic systems.
\newblock {\em Physical Review B—Condensed Matter and Materials Physics}, 88(3):035109, 2013.

\bibitem{gross2018graph}
Jonathan~L Gross, Jay Yellen, and Mark Anderson.
\newblock {\em Graph theory and its applications}.
\newblock Chapman and Hall/CRC, 2018.

\bibitem{garey1979computers}
Michael~R Garey and David~S Johnson.
\newblock {\em Computers and intractability}, volume 174.
\newblock freeman San Francisco, 1979.

\bibitem{kadowaki1998quantum}
Tadashi Kadowaki and Hidetoshi Nishimori.
\newblock Quantum annealing in the transverse ising model.
\newblock {\em Physical Review E}, 58(5):5355, 1998.

\end{thebibliography}
\clearpage

\section*{Supplementary Information}
\appendix

\section{Comparison to related Results}
In this section, we review recent progress in studying the classical simulation of short-time Hamiltonian simulation and shallow-depth quantum circuits. It has been long known that the simulation of Hamiltonian evolution or quantum circuits is classically intractable. For example, in the seminal work of Ref.~\cite{terhal2002adaptive}, it is shown that the simulation (sample from) of constant-depth quantum circuits is hard unless the polynomial hierarchy collapses. This hardness result persists even for depth-$3$ 2D quantum circuits, indicating the non-simulatability even for constant-depth quantum circuits in low dimensions. Nevertheless, recent efforts have found that generating samples from a constant-depth quantum circuit that supports a geometrical lattice can indeed be simulated efficiently by a classical computer, a counter-intuitive result shows that quantum-classical computational separation may not be fully understood yet. Besides, for more physical-relate purposes, e.g.~estimating expectation values of observables with respect to the quantum circuits, the classical simulation could become tractable in certain scenarios~\cite{bravyi2021classical}.

While it is known that Hamiltonian simulation for a polynomial time is a BQP-complete problem, it is still open that whether decreasing the simulation time to constant or poly-logarithmic the problem remains BQP-complete. To this end, Ref.~\cite{wild2023classical} considered the estimation of the expectation value of $k$-local observables $O$ concerning Hamiltonian evolutional circuit: $\langle O(t) \rangle:=\bra{\psi}e^{iHt} O e^{-iHt}\ket{\psi}$, where $\ket{\psi}$ is an initial state prescribed to be a product state, $H$ is $\mathfrak{d}$-sparse. Through techniques of cluster expansion~\cite{haah2024learning}, which also acts as a key ingredient in this work, a computational complexity scales super-polynomially with $t/t_c$ and $\frac{1}{\epsilon}$ is achieved for $t<t_c$, where $t_c:=1/(2e\mathfrak{d})$ is the critical time for reaching an $\epsilon\|O\|$ error in estimation.
Astute readers may find the results surprising as the complexity is independent of the Hamiltonian's operator norm $\|H\|$, which outperforms the state-of-the-art quantum algorithm~\cite{haah2021quantum} that has polynomial-logarithmic dependence on $\|H\|$. The key to this independence is that disconnected clusters have no contribution to the estimator $\langle O(t) \rangle$ resulting from the geometric locality of the Hamiltonian and sparsity of the observable.
Furthermore, by devising an ingenious analytic continuation of the estimator $\braket{O(t)}$, Ref.~\cite{wild2023classical} manage to extend the classical simulation to arbitrary $O_1$ time with the computational cost blow up to scale doubly exponential with $t/t_c$, that is $\mathrm{poly}\left((\frac{1}{\epsilon}\frac{t}{t_c})^\frac{t}{t_c}\right)$. However, Ref.~\cite{wild2023classical} can only handle the time evolution of local observable $O$.

On the other hand, the concept of `quasi-adiabatic continuation' (QAC)~\cite{hastings2005quasiadiabatic,osborne2007simulating} is developed for extracting physical properties of the adiabatic evolution of constant time. 
Here, constant-time adiabatic evolution is of special interest because it relates to the definition of quantum phases~\cite{chen2010local}. The QAC method is built primarily upon the Lieb-Robinson bound~\cite{lieb1972finite,hastings2010locality}, which utilizes the fact that in the Heisenberg picture, the short-time dynamics $O(t)=e^{iHt}Oe^{-iHt}$ is (approximately) confined in the light cone supported on the interacting graph of $O(0)$. Yet, when the simulation time remains constant, the Lieb-Robinson bound scales super-polynomially with $e^{\mathcal{O}(vt)}$ and $\frac{1}{\epsilon}$, where $v$ is the Lieb-Robinson velocity. As pointed out in Ref.~\cite{wild2023classical}, this result is outperformed by the cluster expansion method with alternatively a polynomial dependence on $\frac{1}{\epsilon}$. While the time-dependence case is not considered in Ref.~\cite{wild2023classical}, here we manage to extend to the classical simulation of time-dependent quantum dynamics, where algorithms with similar complexity are achieved.

S.~Bravyi et al.~\cite{bravyi2021classical} proposed a classical algorithm for simulating the quantum mean value problem for general classes of quantum observables $O=O_1\otimes\cdots\otimes O_n$ and $2$-dimensional constant-depth quantum circuits $U$. Specifically, they divide $UOU^{\dagger}=\prod_{i=1}^nUO_iU^{\dagger}$ into two operators $U_AU_B$, where $U_A$ and $U_B$ can be classically simulated easily. By utilizing the classical Monte Carlo method, they can efficiently simulate the value $\langle0^n|U_AU_B|0^n\rangle$ which approximates the quantum mean value. However, when the unitary $U$ is given by a Hamiltonian dynamics $e^{-iHt}$, the fundamental challenge arises because $e^{-iHt}O_ie^{iHt}$ may not be easily computed by the causal principle which only considers quantum gates in the light-cone of $UO_iU^{\dagger}$. 
If we utilize the Trotter-Suzuki method to translate Hamiltonian dynamics $e^{iHt}$ into quantum circuit model $\tilde{U}$, the resulting quantum circuit depth would be ${\rm poly}(\|H\|,t^{O(1)},(1/\epsilon)^{O(1)})$~\cite{childs2021theory}. When the Hamiltonian norm $\|H\|=\Theta(n)$ or the accuracy $\epsilon=\mathcal{O}(1/n)$, the corresponding quantum circuit $\tilde{U}$ will spread the information to the whole system, and computing $\tilde{U}O_i\tilde{U}^{\dagger}$ by causal principle would be classically hard. In this work, we present a technique able to overcome the above obstacle, by combining the cluster expansion method and analytic continuation. The proposed method can efficiently approximate $e^{-iHt}O_ie^{iHt}$ meanwhile limiting its support size to ${\rm poly}\log n$ rather than ${\rm poly}(n)$. As a result, our method nontrivial extends and generalizes Ref.~\cite{bravyi2021classical} in solving quantum mean value problems when $U$ is given by Hamiltonian dynamics.

From an application perspective, hybrid quantum-classical algorithms, such as quantum approximate optimization algorithms (QAOAs) and variational quantum eigensolvers (VQEs), are paradigmatic protocols for demonstrating the potential quantum advantage on near-term quantum devices. Yet, provable theoretical barriers~\cite{basso2022performance,anshu2023concentration} are found for constant-depth QAOA methods. For solving classical optimization problems, it is discovered that the computational basis measurement of constant-depth QAOA approaches will concentrate regarding the distribution of the output Hamming weight, which can then be used to show their inability to outperform classical algorithms. Our methods on the other hand feature a different direction for classical computation of the output results, which provides a computational-theoretical oriented perspective on the problem.

\section{Theoretical Background}
\label{Appendix:Def}
\subsection{Geometrical Local Hamiltonian}
\begin{definition}[Local Hamiltonian]
    A local Hamiltonian is composed by linear combinations of Hermitian operators $h_X$ which nontrivially acts on the qubit subset $X\subset S$ with the corresponding coefficient $\lambda_X$. Here, the coefficients satisfy $\abs{\lambda_X}\leq 1$ and are chosen such that $\|h_X\|=1$. Here, the subsystem set $X\subset S$ are not necessarily to be geometrical local. We define the associated Hamiltonian as $H=\sum_{X\subset S}\lambda_Xh_X$.
\end{definition}

In this article, we assume Hermitian terms $h_X$ are distinct and non-identity multi-qubit Pauli operators. Such assumption naturally satisfies $\|h_X\|=1$. computation.

\begin{definition}[Operator support~\cite{bravyi2021classical,haah2024learning}]
    The support ${\rm supp}(P)$ of an operator $P$ represents the minimal qubit set such that $P=O_{{\rm supp}(P)}\otimes I_{n\setminus{\rm supp}(P)}$ for some operator $O$.
\end{definition}

\subsection{Cluster and Interaction Graph}

\begin{definition}[Cluster induced by Hamiltonian]
    Given a general local Hamiltonian $$H=\sum_{X\subset S}\lambda_Xh_X,$$ a cluster $\bm{V}$ is defined as a nonempty multi-set of subsystems from $S$, where multi-sets allow an element appearing multiple times. The set of all clusters $\bm V$ with size $m$ is denoted by $\mathcal{C}_m$ and the set of all clusters is represented by $\mathcal{C}=\cup_{m\geq 1}\mathcal{C}_m$.
\end{definition}

For example, if the Hamiltonian $H=X_0X_1+Y_0Y_1$, then possible candidates for $\bm V$ would be $\{X_0X_1\}$, $\{Y_0Y_1\}$, $\{X_0X_1,X_0X_1\},\cdots$. We call the number of times a subsystem $X$ appears in a cluster $\bm{V}$ the multiplicity $\mu_{\bm{V}}(X)$, otherwise we assign $\mu_{\bm{V}}(X)=0$. Enumerate all subsets $X\subset S$ may determine the size of $\bm{V}$, that is $\abs{\bm{V}}=\sum_{X\subset S}\mu_{\bm V}(X)$. In the provided example, when $\bm V=\{X_0X_1,X_0X_1\}$, we have $\mu_{\bm V}(X_0X_1)=2$, $\mu_{\bm V}(Y_0Y_1)=0$ and $\abs{\bm V}=2$.

\begin{definition}[Interaction Graph]
    We associate with every cluster $\bm{V}$ a simple graph $G_{\bm{V}}$ which is also termed as the cluster graph. The vertices of $G_{\bm{V}}$ correspond to the subsystems in $\bm{V}$, with repeated subsystems also appearing as repeated vertices. Two distinct vertices $X$ and $Y$ are connected by an edge if and only if the respective subsystems overlap, that is ${\rm supp}(h_X)\cap{\rm supp}(h_Y)\neq\emptyset$.
\end{definition}

Suppose the cluster $\bm V=\{X_0X_1,X_0X_1\}$, then its corresponding interaction graph $G_{\bm V}$ has two vertices $v_1,v_2$, related to $X_0X_1$ and $X_0X_1$, respectively, and $v_1$ connects to $v_2$ since ${\rm supp}(X_0X_1)\cap {\rm supp}(X_0X_1)\neq \emptyset$. We say a cluster $\bm{V}$ is connected if and only if $G_{\bm{V}}$ is connected. We use the notation $\mathcal{G}_m$ to represent all connected clusters of size $m$ and $\mathcal{G}=\cup_{m\geq 1}\mathcal{G}_m$ for the set of all connected clusters.

\begin{definition}[Super-Interaction Graph]
\label{Def:superinteractiongraph}
    Suppose we have $L$ clusters $\bm V_1,\bm V_2,\cdots, \bm V_L$, we define the super-interaction graph $G^L_{\bm V_1,\cdots,\bm V_L}$ composed by interaction graphs $G_{\bm V_1},G_{\bm V_2},\cdots G_{\bm V_L}$, where vertices $\{h_X\}_{X\subset S_1\cup S_2\cdots\cup S_L}$ inherit from $G_{\bm V_1},G_{\bm V_2},\cdots G_{\bm V_L}$  and vertices $h_X$ and $h_Y$ are connected if ${\rm supp}(h_X)\cap{\rm supp}(h_Y)\neq\emptyset$. 
\end{definition}

In our paper, the super-interaction graph is generally induced by a Hamiltonian series, say $\{H^{(1)},\cdots,H^{(L)}\}$. From the above definition, we know that the super-interaction graph $G^L_{\bm V_1,\cdots,\bm V_L}$ contains $\sum_{k=1}^L\abs{G_{\bm V_k}}$ vertices. 

\begin{definition}[Connected Super-Interaction Graph]
\label{Def:consuperinteractiongraph}
The super-interaction-graph $G^L_{\bm V_1,\cdots,\bm V_L}$ is connected if and only if 
the super cluster $\bm V=(\bm V_1,\bm V_2,\cdots, \bm V_L)$ is connected. All $m$-sized connected super-interaction graphs are denoted by $\mathcal{G}_m^L$, with $m=\sum_{k=1}^L\abs{\bm V_k}$.
\end{definition}

Specifically, we denote $\mathcal{G}_m^{L,O_i}$ as the set of all $m$-sized connected super-interaction graphs which connects to $O_i$.

\subsection{Cluster Expansion}
\label{sec:cluster expansion}
We first consider a simple case, that is the cluster expansion of the single-step Hamiltonian dynamics $e^{iHt}O_ie^{-iHt}$~\cite{wild2023classical}. For any cluster $\bm{V}\in\mathcal{C}_m$, we can write $\bm{V}=(X_1,\cdots, X_m)$. This notation helps us to write the function $e^{iHt}O_ie^{-iHt}$ as the multivariate Taylor-series expansion by using the cluster expansion method. Here, we fix the parameter $O_i$, but considering $\{t,\lambda_X\}$ as variables. As a result, we have
\begin{eqnarray}
\begin{split}
    e^{iHt}O_ie^{-iHt}=&\sum\limits_{m=0}^{+\infty}\frac{t^m}{m!}\left(\frac{\partial^m[e^{iHt}O_ie^{-iHt}]}{\partial t^m}\right)_{t=0}
    \label{Eq:taylor}
\end{split}
\end{eqnarray}
Recall that the Hamiltonian $H=\sum_{X\in S}\lambda_Xh_X$, then we assign $z_X=-it\lambda_X$. This results in 
\begin{align}
    \frac{\partial[e^{iHt}O_ie^{-iHt}]}{\partial t}=\sum_{X\in S}\frac{\partial[e^{iHt}O_ie^{-iHt}]}{\partial z_X}\frac{\partial z_X}{\partial t}=\sum_{X\in S}(-i)\lambda_X\frac{\partial[e^{iHt}O_ie^{-iHt}]}{\partial z_X}.
\end{align}
Taking above derivative function into Eq.~\ref{Eq:taylor}, we have
\begin{eqnarray}
    \begin{split}
         e^{iHt}O_ie^{-iHt}=&\sum\limits_{m=0}^{+\infty}\frac{(-it)^m}{m!}\sum\limits_{X_1,\cdots,X_m}\lambda_{X_1}\cdots\lambda_{X_m}\left(\frac{\partial^m[e^{iHt}O_ie^{-iHt}]}{\partial z_{X_1}\cdots\partial z_{X_m}}\right)_{z=(0,\cdots, 0)}\\
         =&\sum\limits_{m=0}^{+\infty}(-it)^m\sum\limits_{\bm V\in\mathcal{C}_m,V=(X_1,\cdots,X_m)}\frac{\bm\lambda^{\bm V}}{\bm V!}\left(\frac{\partial^m[e^{iHt}O_ie^{-iHt}]}{\partial z_{X_1}\cdots\partial z_{X_m}}\right)_{z=(0,\cdots, 0)}
    \end{split}
\end{eqnarray}
where $\bm\lambda^{\bm V}=\prod_{X\in S}\lambda_X^{\mu_{\bm V}(X)}$ and $\bm V!=\prod_{X\in S}\mu_{\bm V}(X)!$. Finally, we utilize the BCH expansion to compute
\begin{eqnarray}
    \begin{split}
        \left(\frac{\partial^m[e^{iHt}O_ie^{-iHt}]}{\partial z_{X_1}\cdots\partial z_{X_m}}\right)_{z=(0,\cdots, 0)}=&\frac{\partial^m}{\partial z_{X_1}\cdots\partial z_{X_m}}\sum\limits_{j=0}^{\infty}\frac{(-it)^j}{j!}[H,O_i]_j\big|_{z=(0,\cdots,0)}\\
        =&\frac{(-it)^m}{m!}\frac{\partial^m}{\partial z_{X_1}\cdots\partial z_{X_m}}[\underbrace{H,[H,\cdots [H}_{m}, O_i]\cdots]]\big|_{z=(0,\cdots,0)}\\
        =&\frac{(-it)^m}{m!}\sum_{\sigma\in \mathcal{P}_m}[\partial_{z_{X_{\sigma(1)}}}H,\cdots[\partial_{z_{X_{\sigma(m)}}}H,O_i]\cdots]\big|_{z=(0,\cdots,0)}\\
        =&\frac{1}{m!}\sum_{\sigma\in \mathcal{P}_m}[h_{X_{\sigma(1)}},\cdots[h_{X_{\sigma(m)}},O_i]\cdots].
    \end{split}
\end{eqnarray}
As a result, the cluster expansion of the single-step Hamiltonian dynamics can be written as
\begin{eqnarray}
\begin{split}
     e^{iHt}O_ie^{-iHt}&=\sum_{m\geq 0}^{+\infty}\sum_{\bm V\in\mathcal{C}_m}\frac{\bm\lambda^{\bm V}}{\bm V!}\frac{(-it)^m}{m!}\sum\limits_{\sigma\in \mathcal{P}_m}\left[h_{V_{\sigma(1)}},\cdots [h_{V_{\sigma(m)}},O_i]\right]\\
     &=\sum_{m\geq 0}^{+\infty}\sum_{\bm V\in\mathcal{C}_m}\frac{\bm\lambda^{\bm V}}{\bm V!}D_{\bm V}[e^{iHt}O_ie^{-iHt}].
\end{split}
\end{eqnarray}
Here $\mathcal{P}_m$ represents the permutation group on the set $\{1,\cdots,m\}$, and we denote the cluster derivative
\begin{align}
    D_{\bm V}[e^{iHt}O_ie^{-iHt}]=\frac{(-it)^m}{m!}\sum\limits_{\sigma\in \mathcal{P}_m}\left[h_{V_{\sigma(1)}},\cdots [h_{V_{\sigma(m)}},O_i]\right].
    \label{Eq:gradient}
\end{align}
From Eq.~\ref{Eq:gradient}, we know that $V_{\sigma(1)}\cap V_{\sigma(2)}\cap\cdots\cap V_{\sigma(m)}\cap {\rm supp}(O_i)=\emptyset$ may result in $D_{\bm V}\left[e^{iHt}O_ie^{-iHt}\right]=0$~\cite{haah2024learning, wild2023classical}. This property dramatically reduces the computational complexity in approximating $e^{iHt}O_ie^{-iHt}$, which only needs to consider connected clusters $\bm V$ with bounded size. 

\subsection{Heisenberg-picture Dyson series}
\label{sec:Heisenberg Dyson series}
In this section, we provide the derivation of the Heisenberg-picture Dyson series. Given a time-dependent Hamiltonian evolution operator $U(t)=\mathcal{T} e^{-i \int_0^t H(s) \mathrm{d} s}$ and an operator $O$, we are interested in expanding the time-evolved operator $O(t)=U^\dagger(t)OU(t)$ in the Heisenberg picture using the BCH formula. 
To tackle the time dependency, a standard technique is the Dyson series. We begin by reviewing the properties of the Dyson series, which gives 
\begin{equation}
    \mathcal{T}\left[e^{-i \int_0^t H(s) \mathrm{d} s}\right]=\sum_{k=0}^{\infty} \frac{(-i)^k}{k!} \int_0^t \cdots \int_0^t \mathcal{T}\left[H\left(t_k\right) \cdots H\left(t_1\right)\right] \mathrm{d}^k t,
\end{equation}
where $\mathcal{T}$ is the time-ordering operator. The time-ordering operator will put time-dependent operators in non-decreasing order according to the time variables, i.e., $\mathcal{T}(H(t_k)\cdots H(t_2) H(t_1))=H(t_{\sigma(k)})\cdots H(t_{\sigma(2)}) H(t_{\sigma(1)})$ such that $\sigma$ is a permutation satisfying $t_{\sigma(1)} \leq t_{\sigma(2)} \leq \cdots \leq t_{\sigma(k)}$. Following Ref.~\cite{low2018hamiltonian}, we can discretize the integrals for Riemann integrable $H(t)$ as
\begin{equation}\label{eq:discretized dyson}
\mathcal{T}\left[e^{-i \int_0^t H(s) \mathrm{d} s}\right]=\lim _{M \rightarrow \infty} \sum_{k=0}^{\infty} \frac{(-i t)^k}{k!M^k} \tilde{B}_k, \quad \tilde{B}_k=\sum_{m_1, \cdots, m_k=0}^{M-1} \mathcal{T}\left[H\left(m_k \Delta\right) \cdots H\left(m_1 \Delta\right)\right],
\end{equation}
where $\Delta=t/M$.

At first glance, as we need to resort to the Dyson series for approximating $U(t)$, it seems challenging to expand $O(t)$ in the Heisenberg picture using the BCH formula for the derivation of the cluster expansion as what is done in the time-independent case (see Appendix \ref{sec:cluster expansion}). Here, we propose the first Heisenberg-picture Dyson series, i.e., combined with the discretized BCH formula. To this end, we mildly generalize the definition of the time-ordering operator $\mathcal{T}$. That is when acting on a product of operators both time-dependent and time-dependent ones, $\mathcal{T}$ will independently act on each consecutive segment of time-dependent operators and leave the time-independent one out. For instance, $\mathcal{T}(H(t_{a_k})\cdots H(t_{a_2}) H(t_{a_2}) P H(t_{b_l})\cdots H(t_{b_2}) H(t_{b_2}))=\mathcal{T}(H(t_{a_k})\cdots H(t_{a_2}) H(t_{a_2})) P \mathcal{T}(H(t_{b_l})\cdots H(t_{b_2}) H(t_{b_2}))$, where $P$ is some time-independent operator.

Using the generalized time-ordering operator, we define
\begin{equation}\label{eq:time-dependent ob}
    \tilde{O}_{r,s}:=\sum_{m_1,\cdots,m_r;z_1,\cdots,z_s=0}^{M-1}\mathcal{T}[H_{m_r}\cdots H_{m_1}OH_{z_s}\cdots H_{z_1}]= \sum_{m_1,\cdots,m_r;z_1,\cdots,z_s=0}^{M-1}\mathcal{T}[H_{m_r}\cdots H_{m_1}]O\mathcal{T}[H_{z_s}\cdots H_{z_1}],
\end{equation}
where we have used the abbreviation $H(m_r\Delta)$ for $H_{m_r}$, and the second equation is due to our generalization of the time-ordering operator. Besides, we use the adjoint notation to denote the commutator as $\mathrm{ad}_X(Y)=[X,Y]$. We can then write the nested commutator of time-dependent Hamiltonians into the form:
\begin{equation}\label{eq:ad_comm}
    \mathcal{T}[\mathrm{ad}^j_H(O)]:=\mathcal{T}\left[\sum_{m_j,\cdots,m_1=0}^{M-1}\left[H_{m_j}, \cdots,\left[H_{m_1}, O  \right] \cdots\right]\right].
\end{equation} 
We propose a helpful lemma for the nested commutator to deduce the BCH formula. 
\begin{lemma}\label{lemma:adjoint expansion}
Let $\mathcal{T}[\mathrm{ad}^j_H(O)]$ be a nested commutator defined by Eq.~\eqref{eq:ad_comm}. Then, we have
\begin{equation}\label{eq:adjoint expansion}
    \mathcal{T}[\mathrm{ad}^{j}_H(O)]=\sum_{i=0}^j (-1)^i \binom{j}{i}\tilde{O}_{j-i,i},
\end{equation}
where $\tilde{O}_{j-i,i}$ is given by Eq.~\eqref{eq:time-dependent ob}.
\end{lemma}
\begin{proof}
This is proved by induction. First, we inspect that the first-order expansion is valid for our formula. Then, we proceed by assuming that $j$-th expansion is feasible and check the $(j+1)$-th term:
\begin{equation}
\begin{aligned}
    \mathcal{T}[\mathrm{ad}^{j+1}_H(O)]&=\sum_{i=0}^j (-1)^i \binom{j}{i}\widetilde{[H_{j+1},O]}_{j-i,i}\\
    &=\sum_{i=0}^j (-1)^i \binom{j}{i}\tilde{O}_{j+1-i,i} - \sum_{i=0}^j (-1)^i \binom{j}{i}\tilde{O}_{j-i,i+1}\\
    &=\tilde{O}_{j+1,0}+\sum_{i=1}^j (-1)^i \binom{j}{i}\tilde{O}_{j+1-i,i}-\sum_{i^\prime=1}^j (-1)^{i^\prime-1} \binom{j}{i^\prime-1}\tilde{O}_{j+1-i^\prime,i^\prime} +(-1)^{j+1} \tilde{O}_{0,j+1}\\
    &=\sum_{i=0}^{j+1} (-1)^i \binom{j+1}{i}\tilde{O}_{j+1-i,i},
\end{aligned}
\end{equation}
where in the third line we have substituted $i$ into $i^\prime=i+1$ for the second summation in the second line; and in the last line we have used that fact that $\binom{j}{i}+\binom{j}{i-1}=\binom{j+1}{i}$.
\end{proof}

We then apply the discretized Dyson expansion provided in Eq.~\eqref{eq:discretized dyson} to the Heisenberg picture formula, and thereby obtain the following Heisenberg-picture (discretized) Dyson series, which is summarized as
\begin{lemma}\label{lemma:dyson-bch}
Let $O(t)=\mathcal{T} e^{i \int_0^t H(s) \mathrm{d} s} O \mathcal{T} e^{-i \int_0^t H(s) \mathrm{d} s}$ be an observable evolved by a time-dependent evolution in the Heisenberg picture. Then, we can expand $O(t)$ as
\begin{equation}\label{eq:dyson-bch}
    O(t)=\lim_{M\rightarrow\infty}\sum_{j=0}^\infty \frac{(it)^j}{j!M^j}\mathcal{T}[\mathrm{ad}_H^j(O)].
\end{equation}
\end{lemma}

\begin{proof}
We prove by simply applying the result given by Lemma \ref{lemma:adjoint expansion}:
\begin{equation}\label{eq:heisenberg dyson proof}
\begin{aligned}
    O(t)&=\mathcal{T} e^{i \int_0^t H(s) \mathrm{d} s} O \mathcal{T} e^{-i \int_0^t H(s) \mathrm{d} s}\\
    &=\lim _{M \rightarrow \infty} \sum_{p,q=0}^{\infty} (-1)^q\frac{(i t)^{p+q}}{p!q!M^{p+q}} \tilde{B}_p O  \tilde{B}_q\\
    &=\lim _{M \rightarrow \infty} \sum_{s}^{\infty}\sum_{d=0}^s (-1)^d\left(\frac{it}{M}\right)^s \frac{\tilde{O}_{s-d,d}}{(s-d)!d!}=\lim_{M\rightarrow\infty}\sum_{j=0}^\infty \frac{(it)^j}{j!M^j}\mathcal{T}[\mathrm{ad}_H^j(O)].
\end{aligned}
\end{equation}
Here, in the second line, we insert the discretized Dyson expansion. We then substitute $p,q$ into $s-d,d$ and also note that $\tilde{O}_{s-d,d}=\tilde{B}_{s-d}O\tilde{B}_d$ as prescribed by Eq.~\eqref{eq:time-dependent ob}. Eventually, we obtain the final result by applying the conclusion on the nested commutator given by Lemma \ref{lemma:adjoint expansion}.
\end{proof}

\subsection{Cluster expansion of time-dependent Hamiltonian dynamics}
\label{sec:time-ClusterExpansion}
In this section, we deduce cluster expansion of time-dependent Hamiltonian dynamics following the ideas from Sec.~\ref{sec:cluster expansion} with a special interest in adiabatic evolution. Our aim is to derive the multi-variate Taylor expansion for $O(t)=U^\dagger(t)OU(t)$ with $U(t)=\mathcal{T} e^{-i \int_0^t H(s) \mathrm{d} s}$. We also assume that the derivative of $H(t)$ higher than order $1$ vanishes, which is reasonable for typical adiabatic dynamics.
We note that for each term in the time-dependent Hamiltonian, the time-dependent part resides in the coefficients. 
\begin{theorem}[Time-dependent cluster expansion]\label{thm:time-dependent cluster expansion}
Given $O(t)=\mathcal{T} e^{i \int_0^t H(s) \mathrm{d} s} O \mathcal{T} e^{-i \int_0^t H(s) \mathrm{d} s}$, its cluster expansion is
\begin{equation}\label{eq:time-dependent cluster expansion}
\begin{split}
    O(t)=\lim_{M\rightarrow\infty}\sum\limits_{m=0}^{+\infty}\frac{(-it)^{m}}{m!M^m}\sum\limits_{\bm V\in\mathcal{C}_m,V=(X_1,\cdots,X_m)}\frac{\tilde{\bm\lambda}^{\bm V}}{\bm V!}&\Bigg(\mathcal{T}\Bigg[ \sum_{n_m,\cdots,n_1=0}^{M-1}\sum_{\sigma\in \mathcal{P}_m}\Big[h_{X_{\sigma(m)}} f(n_m,t,X_{\sigma(m)}),\cdots,\\
    &\left[h_{X_{\sigma(1)}} f(n_1,t,X_{\sigma(1)}), O\Big],\cdots\right]\Bigg]\Bigg)_{z=(0,\cdots, 0)},
\end{split}
\end{equation}
where $\tilde{\lambda}_X(t)=\lambda_X(t)+t\lambda^\prime_X(t)$ with $\lambda^\prime_X(t)=\frac{\mathrm{d}\lambda_X(t)}{\mathrm{d}t}$, $\tilde{\bm\lambda}^{\bm V}=\prod_{X\in S}\tilde{\lambda}_X^{\mu_{\bm V}(X)}$, $\bm V!=\prod_{X\in S}\mu_{\bm V}(X)!$, $f(n_m,t,X):=\frac{\partial Z_{X}(n_m)}{\partial Z_{X}(t)}$ and $Z_X(t)=-it\lambda_X(t)$.
\end{theorem}

Especially, as we consider the adiabatic Hamiltonian of the form $H(t)=(1-t)H_0+tH_1$, we could write the Hamiltonian as a sum of local operators $H(t)=\sum_{X\in S} \lambda_X(t)h_X$. We have the following as the Taylor expansion of $O(t)$ with $\{t,\lambda_X(t)\}$:
\begin{eqnarray}
\begin{split}
    O(t)=&\sum\limits_{m=0}^{+\infty}\frac{t^m}{m!}\left(\frac{\partial^m[\mathcal{T} e^{i \int_0^t H(s) \mathrm{d} s}O\mathcal{T} e^{-i \int_0^t H(s) \mathrm{d} s}]}{\partial t^m}\right)_{t=0}.
    \label{Eq:time_taylor}
\end{split}
\end{eqnarray}
Then, by changing of variable as $Z_X(t)=-it\lambda_X(t)$, we have
\begin{equation}
    \frac{\mathrm{d}}{\mathrm{d}t}O(t)=\sum_X -i\tilde{\lambda}_X(t)\frac{\partial[\mathcal{T} e^{i \int_0^t H(s) \mathrm{d} s}O\mathcal{T} e^{-i \int_0^t H(s) \mathrm{d} s}]}{\partial Z_X(t)},
\end{equation}
where $\tilde{\lambda}_X(t)=\lambda_X(t)+t\lambda^\prime_X(t)$ 
 with $\lambda^\prime_X(t)=\frac{\mathrm{d}\lambda_X(t)}{\mathrm{d}t}$. Similarly, we take the above result to Eq.~\eqref{Eq:time_taylor}, which results in
\begin{eqnarray}\label{eq:time-cluster1}
    \begin{split}
         O(t)=&\sum\limits_{m=0}^{+\infty}\frac{(-it)^m}{m!}\sum\limits_{X_1,\cdots,X_m}\tilde{\lambda}_{X_1}(t)\cdots\tilde{\lambda}_{X_m}(t)\left(\frac{\partial^m[\mathcal{T} e^{i \int_0^t H(s) \mathrm{d} s}O\mathcal{T} e^{-i \int_0^t H(s) \mathrm{d} s}]}{\partial z_{X_1}(t)\cdots\partial z_{X_m}(t)}\right)_{z=(0,\cdots, 0)}\\
         =&\sum\limits_{m=0}^{+\infty}(-it)^m\sum\limits_{\bm V\in\mathcal{C}_m,V=(X_1,\cdots,X_m)}\frac{\tilde{\bm\lambda}^{\bm V}}{\bm V!}\left(\frac{\partial^m[\mathcal{T} e^{i \int_0^t H(s) \mathrm{d} s}O\mathcal{T} e^{-i \int_0^t H(s) \mathrm{d} s}]}{\partial z_{X_1}(t)\cdots\partial z_{X_m}(t)}\right)_{z=(0,\cdots, 0)}\\
         =&\lim_{M\rightarrow\infty}\sum\limits_{m=0}^{+\infty}(-it)^m\sum\limits_{\bm V\in\mathcal{C}_m,V=(X_1,\cdots,X_m)}\frac{\tilde{\bm\lambda}^{\bm V}}{\bm V!}\left(\frac{\partial^m}{\partial z_{X_1}(t)\cdots\partial z_{X_m}(t)}\sum_{j=0}^\infty \frac{(it)^j}{j!M^j}\mathcal{T}[\mathrm{ad}_H^j(O)]\right)_{z=(0,\cdots, 0)}\\
         =&\lim_{M\rightarrow\infty}\sum\limits_{m=0}^{+\infty}\frac{t^{2m}}{m!M^m}\sum\limits_{\bm V\in\mathcal{C}_m,V=(X_1,\cdots,X_m)}\frac{\tilde{\bm\lambda}^{\bm V}}{\bm V!}\left(\mathcal{T}\left[ \sum_{n_m,\cdots,n_1}\sum_{\sigma\in \mathcal{P}_m}\left[\frac{\partial H_{n_m}}{\partial Z_{X_{\sigma(m)}}(t)},\cdots,\left[\frac{\partial H_{n_1}}{\partial Z_{X_{\sigma(m)}}(t)}, O\right],\cdots\right] \right]\right)_{z=(0,\cdots, 0)}.
    \end{split}
\end{eqnarray}
Here, in the second line, we apply $\tilde{\bm\lambda}^{\bm V}=\prod_{X\in S}\tilde{\lambda}_X^{\mu_{\bm V}(X)}$ and $\bm V!=\prod_{X\in S}\mu_{\bm V}(X)!$. Then, we substitute in the BCH expansion of $O(t)$ that is derived by Eq.~\eqref{eq:dyson-bch}. In this last line, we note that the only non-zero term in the summation over $j$ is $j=m$. This is because, for $j<m$ ones, the derivative vanishes because of our assumption about the derivative of the Hamiltonian; for $j>m$ ones, there will be some Hamiltonian left untouched by the partial differentiation operation, which also vanishes because when $Z=0$, the corresponding Hamiltonian becomes zero.

Then, we remark that because each $H_{n_m}=(1-\frac{n_m t}{M})H_0+\frac{n_m t}{M}H_1=\sum_X \lambda_X(n_m)h_X$, we have
\begin{equation}
    \frac{\partial H_{n_m}}{\partial Z_{X_{\sigma(m)}}(t)}=\frac{\partial H_{n_m}}{\partial Z_{X_{\sigma(m)}}(n_m)}\cdot \frac{\partial Z_{X_{\sigma(m)}}(n_m)}{\partial Z_{X_{\sigma(m)}}(t)}=\frac{h_{X_{\sigma(m)}}}{it} f(n_m,t,X_{\sigma(m)}),
\end{equation}
where we have taken the abbreviation $f(n_m,t,X):=\frac{\partial Z_{X}(n_m)}{\partial Z_{X}(t)}$. Taking this result to Eq.~\eqref{eq:time-cluster1}, we finally obtain
\begin{equation}\label{eq:time-dependent final}
\begin{split}
    O(t)=\lim_{M\rightarrow\infty}\sum\limits_{m=0}^{+\infty}\frac{(-it)^{m}}{m!M^m}\sum\limits_{\bm V\in\mathcal{C}_m,V=(X_1,\cdots,X_m)}\frac{\tilde{\bm\lambda}^{\bm V}}{\bm V!}&\Bigg(\mathcal{T}\Bigg[ \sum_{n_m,\cdots,n_1}\sum_{\sigma\in \mathcal{P}_m}\Big[h_{X_{\sigma(m)}} f(n_m,t,X_{\sigma(m)}),\cdots,\\
    &\left[h_{X_{\sigma(1)}} f(n_1,t,X_{\sigma(1)}), O\Big],\cdots\right]\Bigg]\Bigg)_{z=(0,\cdots, 0)}.
\end{split}
\end{equation}
\subsection{Polynomial Function for unitary 2-design}
\begin{lemma}[Lemma~3 of Supp~Mat. in Ref.~\cite{cerezo2021cost}]
    Let $W\in {\rm SU}(d)$ form a unitary $t$-design with $t\geq 2$ and let $A, B, C, D$ be arbitrary linear operators, then
    \begin{eqnarray}
        \begin{split}
            \int {\rm d}\mu(W){\rm Tr}[WAW^{\dagger}B]{\rm Tr}[WCW^{\dagger}D]&=\frac{1}{d^2-1}\left({\rm Tr}[A]{\rm Tr}[B]{\rm Tr}[C]{\rm Tr}[D]+{\rm Tr}[AC]{\rm Tr}[BD]\right)\\
            &-\frac{1}{d(d^2-1)}\left[{\rm Tr}[AC]{\rm Tr}[B]{\rm Tr}[D]+{\rm Tr}[A]{\rm Tr}[C]{\rm Tr}[BD]\right],
        \end{split}
    \end{eqnarray}
    where $d=2^n$.
    \label{lemma:tdesign}
\end{lemma}

\section{Cluster expansion of $L$-step Hamiltonian dynamics}
\label{App:LStepClusterExpansion}
We consider the $L$-step scenario driven by $\{H^{(1)},\cdots,H^{(L)}\}$ and corresponding time parameters $\{t_1,\cdots,t_L\}$. According to the linear property of the commutator net, for any Hermitian operator $A$, we have 
\begin{align}
    \left[A,\sum\limits_{\sigma\in \mathcal{P}_m}[h_{V_{\sigma(1)}},\cdots[h_{V_{\sigma(m)}},O_i]]\right]= \sum\limits_{\sigma\in \mathcal{P}_m}\left[A,[h_{V_{\sigma(1)}},\cdots[h_{V_{\sigma(m)}},O_i]]\right].
    \label{Eq:linearproperty}
\end{align}

Given this observation, we first consider the cluster expansion of $2$-step Hamiltonian dynamics
\begin{eqnarray}
    \begin{split}
        &e^{iH^{(2)}t_2}e^{iH^{(1)}t_1}O_ie^{-iH^{(1)}t_1}e^{-iH^{(2)}t_2}\\
        =&\sum\limits_{m_2\geq 0}\sum\limits_{\bm V_2\in\mathcal{C}_{m_2}}\frac{\bm\lambda^{\bm V_2}}{\bm V_2!}D_{\bm V_2}\left[e^{iH^{(2)}t_2}\sum_{m_1\geq 0}^{+\infty}\sum_{\bm V_1\in\mathcal{C}_{m_1}}\frac{\bm\lambda_1^{\bm V_1}}{\bm V_1!}\frac{(-it)^{m_1}}{m_1!}\sum\limits_{\sigma\in \mathcal{P}_{m_1}}\left[h_{V_{\sigma(1)}},\cdots [h_{V_{\sigma(m_1)}},O_i]\right]e^{-iH^{(2)}t_2}\right]\\
        =&\sum\limits_{m_2\geq 0}\sum\limits_{\bm V_2\in\mathcal{C}_{m_2}}\frac{\bm\lambda^{\bm V_2}}{\bm V_2!}\frac{(-it_2)^{m_2}}{m_2!}\sum\limits_{\sigma_2\in \mathcal{P}_{m_2}}\left[h_{V_{\sigma_2(1)}},\cdots\left[h_{V_{\sigma_2(m_2)}},\sum_{m_1\geq 0}^{+\infty}\sum_{\bm V_1\in\mathcal{C}_{m_1}}\frac{\bm\lambda_1^{\bm V_1}}{\bm V_1!}\frac{(-it)^{m_1}}{m_1!}\sum\limits_{\sigma_1\in \mathcal{P}_{m_1}}\left[h_{V_{\sigma_1(1)}},\cdots [h_{V_{\sigma_1(m_1)}},O_i]\right]\right]\right]\\
        =&\sum\limits_{m_1,m_2\geq 0}\sum\limits_{(\bm V_1,\bm V_2)}\frac{\bm\lambda^{\bm V_1}\bm\lambda^{\bm V_2}}{\bm V_1!\bm V_2!}\frac{(-it_1)^{m_1}(-it_2)^{m_2}}{m_1!m_2!}\sum\limits_{\substack{\sigma_2\in \mathcal{P}_{m_2}\\\sigma_1\in \mathcal{P}_{m_1}}}\left[h_{V_{\sigma_2(1)}},\cdots \left[h_{V_{\sigma_2(m_2)}}\cdots\left[h_{V_{\sigma_1(m_1)}},O_i\right]\right]\right],
    \end{split}
\end{eqnarray}
where the second equality comes from the relationship given by Eq.~\ref{Eq:linearproperty}. Repeat above process for $L$ times, we have the cluster expansion of $L$-step Hamiltonian dynamics, that is
\begin{align}
    U_i(\vec{t})=\sum\limits_{\substack{m_1\geq0\\\cdots\\m_L\geq 0}}\sum\limits_{\substack{\bm V_1\in\mathcal{C}_{m_1}\\\cdots\\\bm V_L\in\mathcal{C}_{m_L}}}\frac{\prod_{k=1}^L(\bm\lambda^{\bm V_k}(-it_k)^{m_k})}{\prod_{k=1}^L\bm V_k!m_k!}\sum\limits_{\substack{\sigma_1\in \mathcal{P}_{m_1}\\\cdots\\\sigma_L\in \mathcal{P}_{m_L}}}\left[h_{V_{\sigma_L(1)}},\cdots \left[h_{V_{\sigma_L(m_L)}},\cdots\left[h_{V_{\sigma_1(m_1)}},O_i\right]\right]\right].
    \label{Eq:Ut}
\end{align}
Here, notations $\sigma_1,\cdots,\sigma_L$ represent $L$ permutations, and $\mathcal{P}_{m_1}\cdots,\mathcal{P}_{m_L}$ represents corresponding permutation groups.

Similar to the single-step Hamiltonian dynamics, we know that if clusters $\bm V_1,\bm V_2,\cdots, \bm V_L$ and $O_i$ are disconnected, then the commute net $\left[h_{V_{\sigma_L(1)}},\cdots [h_{V_{\sigma_1(m_1)}},O_i]\right]=0$, which can be summarized as the following lemma.
\begin{lemma}
    Given clusters $\bm V_1,\bm V_2,\cdots,\bm V_L$ and an observable $O_i$,
    if the supper-interaction graph induced by $\bm V=(\bm V_1,\bm V_2,\cdots,\bm V_L,O_i)$ is disconnected, then the commute net $$\left[h_{V_{\sigma_L(1)}},\cdots[h_{V_{\sigma_L(m_L)}},\cdots[h_{V_{\sigma_1(1)}}\cdots[h_{V_{\sigma_1(m_1)}},O_i]]]\right]=0,$$
    where $\abs{\bm V_k}=m_k$ and $\sigma_k(1),\sigma_k(2),\cdots,\sigma_k(m_k)$ represents an entry of the permutation group $\mathcal{P}_{m_k}$.
\end{lemma}
\begin{proof}
    Denote all connected super-interaction graph as  $\mathcal{G}_{\bm V_1,\bm V_2,\cdots,\bm V_L,O_i}^L$. Consider a cluster $\bm W\notin \mathcal{G}_{\bm V_1,\bm V_2,\cdots,\bm V_L,O_i}^L$. For every permutation series $(\sigma_1(1),\cdots,\sigma_1(m_1),\cdots,\sigma_L(m_L))$, there exists an index $\sigma_k(s)$ such that $\bm W_{\sigma_k(s)}$ and $\bm W_{\sigma_k(s+1)}\cup\cdots\cup\bm W_{\sigma_L(m_L)}\cup {\rm supp}(O_i)$ does not have an overlap. This directly results in 
    \begin{align}
        \left[h_{W_{\sigma_{k(s)}}},\cdots[h_{W_{\sigma_{k(m_k)}}}\cdots[h_{W_{\sigma_L(1)}}\cdots[h_{V_{\sigma_L(m_L)}},O_i]]]\right]=0,
    \end{align}
    and the concerned commutator vanishes.
\end{proof}

Using this property, we may rewrite the above expression by introducing the connected cluster set $\mathcal{G}_m^{L,O_i}$ composed by all connected super-interaction graphs $G^L_{\bm V_1,\cdots, \bm V_L}$ (connected to $O_i$) with size $$m=\abs{\bm V_1}+\cdots+\abs{\bm V_L}.$$ Here, $O_i$ is a single-qubit operator non-trivially acts on qubit $i$, then $\{\bm V_1,\cdots,\bm V_L,O_i\}$ are connected implies ${\rm supp}(O_i)\in \bm V_k$ for some $k\in[L]$. Such observation enables us to only consider summation over $\mathcal{G}_m^{L,O_i}$, meanwhile truncate the cluster expansion up to $M$ order, that is
\begin{align}
    V_i(\vec{t})=\sum\limits_{\substack{m_1\geq0\\\cdots\\m_L\geq 0}}^M\sum\limits_{(\bm V_1\cdots, \bm V_L)\in\mathcal{G}_m^{L,O_i}}\frac{\prod_{k=1}^L(\bm\lambda^{\bm V_k}(-it_k)^{m_k})}{\prod_{k=1}^L\bm V_k!m_k!}\sum\limits_{\substack{\sigma_1\in \mathcal{P}_{m_1}\\\cdots\\\sigma_L\in \mathcal{P}_{m_L}}}\left[h_{V_{\sigma_L(1)}},\cdots [h_{V_{\sigma_1(m_1)}},O_i]\right].
    \label{Eq:clusterexpansion}
\end{align}
Given above knowledge, we can outline our algorithm and provide the running time complexity analysis.

\section{Methods Outline}
\label{Method}

\subsection{Algorithm Outline}
\label{AlgOutline}
We consider to decompose the Hamiltonian dynamics onto Pauli basis. For any Hermitian operator $A$, we suppose its Pauli decomposition is written by $A=\sum_{P}\alpha_PP$. Then we implement the inter-product on the operator $P^{\prime}$ on both sides, that is
\begin{align}
    {\rm Tr}[AP^{\prime}]=\sum_P\alpha_P{\rm Tr}[PP^{\prime}]=\sum_P\alpha_P2^n\delta_{PP^{\prime}}.
\end{align}
This give rises to the coefficient $\alpha_{P^{\prime}}=\frac{{\rm Tr}[AP^{\prime}]}{2^n}$.

Specifically, let $U_O(t)=e^{iHt}(O_1\otimes \cdots\otimes O_n)e^{-iHt}$, and we have
\begin{align}
     U_O(t)=\frac{1}{2^n}\sum\limits_{P\in\{I,X,Y,Z\}^{\otimes n}}{\rm Tr}\left[e^{iHt}Oe^{-iHt}P\right]P=\frac{1}{2^n}\sum\limits_{P\in\{I,X,Y,Z\}^{\otimes n}}{\rm Tr}\left[Oe^{-iHt}Pe^{iHt}\right]P.
\end{align}
Now we denote 
\begin{align}
    U_{\rm cut}(t)=\frac{1}{2^n}\sum\limits_{|P|<k}{\rm Tr}\left[Oe^{-iHt}Pe^{iHt}\right]P,
\end{align}
and 
\begin{align}
    V(t)=\frac{1}{2^n}\sum\limits_{|P|<k}{\rm Tr}\left[OV_P(t)\right]P,
\end{align}
where $V_P(t)$ is given by the Eq.~\ref{Eq:clusterexpansion} such that $\|V_P(t)-e^{-iHt}Pe^{iHt}\|\leq\epsilon$ for local Pauli operator $P$. In the following sections, we show that $V(t)$ would be an estimator to $U_O(t)$ for most of input states and global observables.

\subsection{Average-Case error}
For any quantum state $|\psi\rangle=C|0^n\rangle$ with the Clifford circuit $C\in{\rm Cl}(2^n)$, we have
\begin{align}
   \abs{\langle\psi|(V(t)-U_{\rm cut})|\psi\rangle}\leq \left\|V(t)-U_{\rm cut}\right\|\leq \frac{1}{2^n}\sum\limits_{|P|<k}\abs{{\rm Tr}\left[O(V_P(t)-e^{-iHt}Pe^{iHt})\right]\|P\|}\leq \epsilon n^k,
\end{align}
where we utilize the inequality $|{\rm Tr}[AB]|\leq \|A\|_1\|B\|$. 

Then we consider the truncate error
\begin{eqnarray}
    \begin{split}
        \mathbb{E}_{|\psi\rangle}\abs{\langle\psi|U_O(t)-U_{\rm cut}(t)|\psi\rangle}^2&=\frac{1}{4^n}\sum\limits_{|P|,|Q|\geq k}{\rm Tr}\left[Oe^{-iHt}Pe^{iHt}\right]{\rm Tr}\left[Oe^{-iHt}Qe^{iHt}\right]\mathbb{E}_{|\psi\rangle}\langle\psi|P|\psi\rangle\langle\psi|Q|\psi\rangle\\
        &=\frac{1}{4^n}\sum\limits_{|P|\geq k}\abs{{\rm Tr}\left[Oe^{-iHt}Pe^{iHt}\right]}^2\frac{1}{2^n+1}.
    \end{split}
\end{eqnarray}
Since $|P|\geq k$ and $|Q|\geq k$, we know that $P, Q\neq I$. Here, we utilize the result promised by lemma~\ref{lemma:tdesign}
\begin{eqnarray}
    \begin{split}
        \mathbb{E}_{\psi}\langle\psi|P|\psi\rangle\langle\psi|Q|\psi\rangle&=\frac{1}{d^2-1}\left[{\rm Tr}[P]{\rm Tr}[Q]+{\rm Tr}[PQ]\right]-\frac{1}{d(d^2-1)}\left[{\rm Tr}[P]{\rm Tr}[Q]+{\rm Tr}[PQ]\right]\\
        &=\frac{\delta_{PQ}}{(d+1)},
    \end{split}
\end{eqnarray}
    where $d=2^n$ and $\delta_{PQ}=1$ iff $P=Q$, else $\delta_{PQ}=0$.

\noindent(i)~When the operator $O$ represents a constant-rank projector, such as $O=|\psi\rangle\langle\psi|$ and $O=|\psi\rangle\langle\phi|$, we have 
\begin{align}
    \abs{{\rm Tr}\left[Oe^{-iHt}Pe^{iHt}\right]}\leq \|O\|_1\|e^{-iHt}Pe^{iHt}\|=\mathcal{O}(1),
    \label{Eq:D9}
\end{align}
which implies
\begin{align}
    \mathbb{E}_{\psi}\abs{\langle\psi|U_O(t)-U_{\rm cut}(t)|\psi\rangle}^2=\frac{1}{4^n(2^n+1)}\sum\limits_{|P|\geq k}\abs{{\rm Tr}\left[Oe^{-iHt}Pe^{iHt}\right]}^2\leq \frac{4^n-n^k}{4^n(2^n+1)}\leq \mathcal{O}\left(2^{-n}\right).
    \label{Eq:D10}
\end{align}
Using the Markov inequality ${\rm Pr}[|X|\leq\delta]={\rm Pr}[|X|^2\leq\delta^2]\geq 1-\mathbb{E}[|X|^2]/\delta$, it is shown that
$\abs{\langle\psi|U_O(t)-U_{\rm cut}(t)|\psi\rangle}\leq\epsilon n^k$ with probability larger than $1-1/2^n(\epsilon^2 n^{2k})$. Combine all together, we conclude that 
\begin{align}
    \abs{\langle\psi|(V(t)-U_O(t))|\psi\rangle}\leq \mathcal{O}\left(\epsilon n^k\right)
    \label{Eq:D11}
\end{align}
with nearly unit probability.

\noindent(ii)~When the operator $O$ represents a Pauli operator, we consider an average case scenario, that is $O_W=WO_oW^{\dagger}$ with $W\in {\rm Cl}(2^n)$ and $O_0$ represents some $n$-qubit Pauli operators with ${\rm Tr}[O_0]=0$. According to lemma~\ref{lemma:tdesign}, we have
\begin{align}
   & \mathbb{E}_W \abs{{\rm Tr}\left[O_We^{-iHt}Pe^{iHt}\right]}^2\\=& \mathbb{E}_W{\rm Tr}\left[WO_0W^{\dagger}e^{-iHt}Pe^{iHt}\right]{\rm Tr}\left[WO_0W^{\dagger}e^{-iHt}Pe^{iHt}\right]\\
    =&\frac{1}{d^2-1}\left[{\rm Tr}[O_0]{\rm Tr}[e^{-iHt}Pe^{iHt}]{\rm Tr}[O_0]{\rm Tr}[e^{-iHt}Pe^{iHt}]+{\rm Tr}[O^2_0]{\rm Tr}[(e^{-iHt}Pe^{iHt})(e^{-iHt}Pe^{iHt})]\right]\\
    &-\frac{1}{d(d^2-1)}\left[{\rm Tr}[O_0^2]{\rm Tr}[e^{-iHt}Pe^{iHt}]{\rm Tr}[e^{-iHt}Pe^{iHt}]+{\rm Tr}[O_0]{\rm Tr}[O_0]{\rm Tr}[(e^{-iHt}Pe^{iHt})(e^{-iHt}Pe^{iHt})]\right]\\
    =&\frac{4^n}{4^n-1}\approx 1.
\end{align}
The first equality comes from ${\rm Tr}\left[\left(WO_0W^{\dagger}e^{-iHt}Pe^{iHt}\right)^{\dagger}\right]={\rm Tr}\left[WO_0W^{\dagger}e^{-iHt}Pe^{iHt}\right]$, and the
last equality utilizes results ${\rm Tr}[O_0^2]=d$, ${\rm Tr}[e^{-iHt}Pe^{iHt}]={\rm Tr}[P]=0$, and ${\rm Tr}[(e^{-iHt}Pe^{iHt})(e^{-iHt}Pe^{iHt})]={\rm Tr}[P^2]=d$.

This further gives rise to
\begin{align}
    \mathbb{E}_{\psi,W}\abs{\langle\psi|U_O(t)-U_{\rm cut}(t)|\psi\rangle}^2=\frac{1}{4^n(2^n+1)}\sum\limits_{|P|\geq k}\mathbb{E}_W\abs{{\rm Tr}\left[O_We^{-iHt}Pe^{iHt}\right]}^2= \mathcal{O}(1/2^n).
    \label{Eq:D17}
\end{align}
Still using the Markov inequality, one obtains $\abs{\langle\psi|U_O(t)-U_{\rm cut}(t)|\psi\rangle}\leq\epsilon n^k$ with probability larger than $1-1/2^n(\epsilon^2 n^{2k})$ over the choice of input state $|\psi\rangle$ and observable $O\in\mathcal{P}$.

\subsection{Global observable ensemble}
Here, we consider another ensemble for $O$, enabling $O$ possesses the ``global'' property. Specifically, we consider $O_0$ be a global Pauli operator whose Hamming weight is $n$. Without loss of generality, we assume $n=0 (mod 2)$. Let 
\begin{align}
    \mathcal{O}_{\rm ob}=\{O|O=(W_{12}\otimes W_{3,4}\cdots W_{n-1,n})O_0(W_{12}\otimes W_{3,4}\cdots W_{n-1,n})^{\dagger}\},
\end{align}
where each $W_{j,j+1}$ represents a random $2$-qubit Clifford gate. Since $O_0$'s Hamming weight is $n$, and each $W_{j,j+1}$ reduces at most $1$ Hamming weight (unitary evolution does not change the spectral distribution), the Hamming weight of the resulting $O$ is at least $n/2$, which can still be considered as a ``global'' observable.

Before we give our result, we require the following lemma.

\begin{lemma}
    Suppose an operator $D$ acting on spaces $A\otimes B$, then $\|{\rm Tr}_B(D)\|_2^2\leq {\rm dim}(B)\|D\|_2^2$, where the partial trace is defined by ${\rm Tr}_B[\cdot]=\sum_{i\in B}(I_A\otimes \langle i|)D(I_A\otimes |i\rangle)$.
\end{lemma}
\begin{proof}
    Suppose the orthogonal basis $\{|j\rangle_A\}$ are defined on $A$, and $\{|\alpha\rangle_B\}$ are defined on $B$. Then we have ${\rm Tr}_B[D]=\sum\limits_{\alpha}(I_A\otimes\langle\alpha|)D(I_A\otimes |\alpha\rangle)$, and the squared Hilbert-Schmidt norm is 
    \begin{align}
        \|{\rm Tr}_B[D]\|_2^2={\rm Tr}_A\left[{\rm Tr}^{\dagger}_B[D]{\rm Tr}_B[D]\right]=\sum\limits_{j,k\in A}\langle j|{\rm Tr}^{\dagger}_B[D]|k\rangle\langle k|{\rm Tr}_B[D]|j\rangle.
    \end{align}
This can be further upper bounded by the Cauchy-Schwarz inequality, that is
\begin{eqnarray}
    \begin{split}
        &\sum\limits_{j,k\in A}\abs{\langle j|{\rm Tr}_B[D]|k\rangle}^2=\sum\limits_{j,k\in A}\abs{\sum\limits_{\alpha\in B}\langle j\alpha|D|k\alpha\rangle}^2\\
        \leq &\sum\limits_{j,k\in A}\left(\sum\limits_{\alpha\in B}\abs{\langle j\alpha|D|k\alpha\rangle}^2\right)\left(\sum\limits_{\alpha\in B}1^2\right)
        ={\rm dim}(B)\sum\limits_{j,k\in A;\alpha\in B}\abs{\langle j\alpha|D|k\alpha\rangle}^2.
    \end{split}
\end{eqnarray}
On other hand, we have
\begin{align}
    \|D\|_2^2={\rm Tr}[D^{\dagger}D]=\sum\limits_{p,q}\langle p|D^{\dagger}|q\rangle\langle q|D|p\rangle=\sum\limits_{p,q}\abs{\langle p|D|q\rangle}^2=\sum\limits_{j,k\in A;\alpha, t\in B}\abs{\langle j\alpha|D|k t\rangle}^2\geq \sum\limits_{j,k\in A;\alpha\in B}\abs{\langle j\alpha|D|k\alpha\rangle}^2.
\end{align}
Above two results finally yield
\begin{align}
    \|{\rm Tr}_B[D]\|_2^2\leq {\rm dim}(B)\sum\limits_{j,k\in A;\alpha\in B}\abs{\langle j\alpha|D|k\alpha\rangle}^2\leq {\rm dim}(B)\sum\limits_{j,k\in A;\alpha, t\in B}\abs{\langle j\alpha|D|k t\rangle}^2={\rm dim}(B)\|D\|_2^2.
\end{align}
This completes the proof.
\end{proof}

We consider the same metric $\mathbb{E}_W \abs{{\rm Tr}\left[O_We^{-iHt}Pe^{iHt}\right]}^2$, with $W=W_{1,2}\otimes W_{3,4}\cdots W_{n-1,n}$. Since $W_{j,j+1}\in {\rm Cl}(4)$, which represents a local $2$-design ensemble. This implies the property
\begin{align}
    \mathbb{E}_{W_{j,j+1}} W_{j,j+1}^{\otimes 2}(O_0(j)\otimes O_0(j+1))^{\otimes 2}W^{\dagger,\otimes 2}_{j,j+1}=\frac{-1}{15}I_4+\frac{4}{15}S_{2},
\end{align}
where the swap operator $S_{2}$ swaps qubits between the first copy and second copy of the involved operator. This expression further gives rise to 
\begin{align}
    \bigotimes\limits_{j=1}^{n/2}\left(\mathbb{E}_{W_{2j-1,2j}} W_{2j-1,2j}^{\otimes 2}(O_0(2j-1)\otimes O_0(2j))^{\otimes 2}W^{\dagger,\otimes 2}_{2j-1,2j}\right)=\sum\limits_{K\subset [n], |K| \text{even}}\left(\frac{-1}{15}\right)^{n-|K|}\left(\frac{4}{15}\right)^{|K|}S_K.
\end{align}
On other hand, we note that 
\begin{eqnarray}
    \begin{split}
        \mathbb{E}_W \abs{{\rm Tr}\left[O_We^{-iHt}Pe^{iHt}\right]}^2
    &=\mathbb{E}_W\abs{{\rm Tr}\left[WO_0W^{\dagger}U_P(t)\right]}^2\\
        =&{\rm Tr}\left[(U_P(t)\otimes U_P(t))\mathbb{E}_W\left(W^{\otimes 2}O_0^{\otimes 2}W^{\dagger,\otimes 2}\right)\right]\\
        =&\sum\limits_{K\subset [n], |K| \text{even}}\left(\frac{-1}{15}\right)^{n-|K|}\left(\frac{4}{15}\right)^{|K|}{\rm Tr}[(U_P(t)\otimes U_P(t))S_K]\\
        =&\sum\limits_{K\subset [n], |K| \text{even}}\left(\frac{-1}{15}\right)^{n-|K|}\left(\frac{4}{15}\right)^{|K|}\|{\rm Tr}_{n-K}[U_P(t)]\|_2^2.
    \end{split}
\end{eqnarray}
We note that 
\begin{align}
    \sum\limits_{K\subset [n], |K| \text{even}}\left(\frac{-1}{15}\right)^{n-|K|}\left(\frac{4}{15}\right)^{|K|}\|{\rm Tr}_{n-K}[U_P(t)]\|_2^2=\sum\limits_{k=0}^{n/2}\binom{n}{2k}\left(\frac{-1}{15}\right)^{n-2k}\left(\frac{4}{15}\right)^{2k}\|{\rm Tr}_{n-K}[U_P(t)]\|_2^2.
\end{align}
This thus gives rise to
\begin{align}
    \mathbb{E}_W \abs{{\rm Tr}\left[O_We^{iHt}Pe^{-iHt}\right]}^2\leq \sum\limits_{k=0}^{n/2}\binom{n}{2k}\left(\frac{1}{15}\right)^{n-2k}\left(\frac{4}{15}\right)^{2k}2^{n-2k}\|U_P(t)\|_2^2\leq \frac{1}{2}\left(\frac{2}{5}\right)^n.
\end{align}

Finally, 
\begin{align}
    \mathbb{E}_{\psi,O}\abs{\langle\psi|U_O(t)-U_{\rm cut}(t)|\psi\rangle}^2=\frac{1}{4^n(2^n+1)}\sum\limits_{|P|\geq k}\mathbb{E}_W\abs{{\rm Tr}\left[O_We^{-iHt}Pe^{iHt}\right]}^2= \mathcal{O}(1/5^n).
    \label{Eq:D17}
\end{align}
Using the Markov inequality, one obtains $\abs{\langle\psi|U_O(t)-U_{\rm cut}(t)|\psi\rangle}\leq\epsilon n^k$ with probability larger than $1-1/5^n(\epsilon^2 n^{2k})$ over the choice of input state $|\psi\rangle$ and observable $O\in\mathcal{O}$.

\subsection{Approximate $V_P(\vec{t})$}
The following lemma proves the convergence of the $L$-step cluster expansion and the support of $V_P(\vec{t})$ for constant times.

\begin{lemma}[Informal]
\label{lemma:clustersupp}
   Given a $k$-local Pauli operator $P$, then for any $L$-step quantum dynamics driven by $\{H^{(1)},\cdots,H^{(L)}\}$ and corresponding constant time parameters $\{t_1,\cdots,t_L\}$, the operator $U_P(\vec{t})=\prod_{k=1}^Le^{iH^{(k)}t_k}P\prod_{k=1}^Le^{-iH^{(k)}t_k}$ can be approximated by an operator $V_P(\vec{t})$ such that $\|U_P(\vec{t})-V_P(\vec{t})\|\leq\epsilon$. Here, 
  $V_P(\vec{t})$ represents a $M=\mathcal{O}\left(e^{\pi teL|P|\mathfrak{d}}\log(e^{\pi teL|P|\mathfrak{d}}/\epsilon)\right)$-order truncated cluster expansion given by Eq.~\ref{Eq:clusterexpansion}, $\mathfrak{d}$ represents the maximum degree of interaction graphs induced by Hamiltonians $\{H^{(1)},\cdots,H^{(L)}\}$ and $t=\max\{\abs{t_k}\}$.
 \end{lemma}
We leave proof details in Appendix~\ref{AppendixA}. We note that $V_P(\vec{t})$ can be efficiently computed by a polynomial classical algorithm.

\begin{lemma}
   Given a $k$-local Pauli operator $P$, the operator $U_P(\vec{t})=U(\vec{t})PU^{\dagger}(\vec{t})$ can be approximated by an operator $V_P(\vec{t})$ such that $\|U_P(\vec{t})-V_P(\vec{t})\|\leq\epsilon^{\prime}$ with $\mathcal{O}\left((e^{tL|P|\mathfrak{d}}/\epsilon^{\prime})^{e^{tL|P|\mathfrak{d}}}\right)$ running time.
   \label{lemma:VTCompute}
 \end{lemma}

The proof details refer to Appendix~\ref{AppendixC}.

\section{Proof of theorem~\ref{theorem1}}
\begin{theorem}[Formal Version of theorem~\ref{theorem1}]
   Let $\{H^{(1)},\cdots,H^{(L)}\}$ be local Hamiltonians acting on an $n$-qubit system, time series $\vec{t}=(t_1,\cdots,t_L)$ with $t=\max\{|t_k|\}_{k=1}^L\leq\mathcal{O}(1)$, Hamiltonian dynamics $U(\vec{t})=\prod_{k=1}^Le^{-iH^{(k)}t_k}$. Then for
   any observable like $O=|\phi\rangle\langle\phi|$ or $n$-qubit Pauli observable $O\in \mathcal{P}$, and the classical input state $|\psi\rangle\in {\rm stabilizer}(2^n)$, the quantum dynamics mean value is defined by 
   \begin{align}
       \mu(\vec{t})=\langle\psi|U(\vec{t})^{\dagger}O U(\vec{t})|\psi\rangle.
   \end{align}
   There exists a classical algorithm that produces an estimation $\hat{\mu}(\vec{t})$ such that $\abs{\mu(\vec{t})-\hat{\mu}(\vec{t})}\leq\epsilon$ in runtime $$\mathcal{O}(n^k(e^{\pi teL\mathfrak{d}}n^k/\epsilon)^{e^{\pi teL\mathfrak{d}}})$$ with success probability at least $\geq 1-1/(2^n\epsilon^2 n^{2k})$, where $k=\mathcal{O}(1)$.
    \label{theorem_formal}
\end{theorem}

\noindent\emph{Added Note:} When the observable $O=|\phi\rangle\langle\phi|$ or $|\phi\rangle\langle\phi^{\prime}|$, the success probability is defined only over the input state set $|\psi\rangle$. If $O\in \mathcal{P}$, the success probability is defined over both stabilizer state $|\psi\rangle$ and $O\in\mathcal{P}$.

\begin{proof}
Given above results, we can analyze the running time meanwhile prove Theorem~\ref{theorem1}. From section~\ref{AlgOutline}, it is shown that 
\begin{align}
    \hat{\mu}(\vec{t})=\langle\psi|V(\vec{t})|\psi\rangle=\frac{1}{2^n}\sum\limits_{|P|<k}{\rm Tr}\left[OV_P(t)\right]\langle\psi|P|\psi\rangle
\end{align}
is a good estimator enabling $\abs{\mu(\vec{t})-\hat{\mu}(\vec{t})}\leq\epsilon$ with nearly unit probability. As a result, on only requires to compute each $V_P(t)$ for $|P|\leq k$.

Lemma~\ref{lemma:VTCompute} implies that each $V_P(t)$ can be exactly computed in $\mathcal{O}((e^{\pi teL|P|\mathfrak{d}}/\epsilon^{\prime})^{e^{\pi teL\mathfrak{d}|P|}})$ running time, where $\epsilon^{\prime}$ characterizes $\|U_P(t)-V_i(t
)\|\leq\epsilon^{\prime}$. Assign $\epsilon=n^k\epsilon^{\prime}$, each $V_P(t)$ may require $\mathcal{O}((n^ke^{\pi teL|P|\mathfrak{d}}/\epsilon)^{e^{\pi teL\mathfrak{d}|P|}})$ running time, and the resulting $\hat{\mu}(\vec{t})$ satisfies $\abs{\langle\psi|U_{\rm cut}(t)|\psi\rangle-\hat{\mu}(\vec{t})}\leq\epsilon$. The estimator contains $\mathcal{O}(n^k)$ terms, and the the resulting computational complexity would be $\mathcal{O}(n^k(n^ke^{\pi teL|P|\mathfrak{d}}/\epsilon)^{e^{\pi teL\mathfrak{d}|P|}})$.

When $O=|\phi\rangle\langle\phi|$ or $O=|\phi\rangle\langle\phi^{\prime}|$ (or even $O=|\psi\rangle\langle\phi|+|\phi\rangle\langle\psi|$), Eqs~\ref{Eq:D9}-\ref{Eq:D11} imply
\begin{align}
    \mathbb{E}_{\psi}\abs{\langle\psi|U_O(t)-U_{\rm cut}(t)|\psi\rangle}^2\leq \mathcal{O}(1/2^n).
\end{align}
Using the Markov inequality, above result implies that for any stabilizer state $|\psi\rangle$, the expectation value error $\abs{\langle\psi|U_O(t)-U_{\rm cut}(t)|\psi\rangle}\leq \epsilon n^{k}$ with success probability at least $1-1/(2^n\epsilon ^2n^{2k})$ over the input state $|\psi\rangle$. Combine all together, $\hat{\mu}(\vec{t})$ is a $\epsilon$-close estimation to $\mu(\vec{t})$ with nearly unit probability.

When $O\in\mathcal{P}$, we need to consider the average-case expectation value difference given by Eq.~\ref{Eq:D17}. It is shown that 
\begin{align}
    \mathbb{E}_{\psi,O}\abs{\langle\psi|U_O(t)-U_{\rm cut}(t)|\psi\rangle}^2=\frac{1}{4^n(2^n+1)}\sum\limits_{|P|\geq k}\mathbb{E}_O\abs{{\rm Tr}\left[Oe^{-iHt}Pe^{iHt}\right]}^2= \mathcal{O}(1/2^n),
\end{align}
implying $\abs{\langle\psi|U_O(t)-U_{\rm cut}(t)|\psi\rangle}\leq \epsilon n^{k}$ with success probability at least $1-1/(2^n\epsilon ^2n^{2k})$ over the input state $|\psi\rangle$ and $O\in\mathcal{P}$.
This completes the proof.
\end{proof}

\section{Proof of Lemma~\ref{lemma:clustersupp}}
\label{AppendixA}
The following lemma proves the convergence of the $L$-step cluster expansion and the support of $V_i(\vec{t})$ for constant times.
\begin{theorem}
[Formal version of  Lemma~\ref{lemma:clustersupp}]
   Given a $k$-local Pauli operator $P$, then for any $L$-step quantum dynamics driven by $\{H^{(1)},\cdots,H^{(L)}\}$ and corresponding constant time parameters $\vec{t}=\{t_1,\cdots,t_L\}$, the operator $U_P(\vec{t})=\prod_{k=1}^Le^{iH^{(k)}t_k}P\prod_{k=1}^Le^{-iH^{(k)}t_k}$ can be approximated by $V_P(\vec{t})$ such that $\|U_P(\vec{t})-V_P(\vec{t})\|\leq\epsilon\|P\|$, where 
   \begin{align}
       V_P(\vec{t})=\sum\limits_{\substack{m_1\geq0\\\cdots\\m_L\geq 0}}^M\sum\limits_{\bm V_1\cdots, \bm V_L\in\mathcal{G}_m^{L,P}}\frac{\prod_{k=1}^L(\bm\lambda^{\bm V_k}(-it_k)^{m_k})}{\prod_{k=1}^L\bm V_k!m_k!}\sum\limits_{\substack{\sigma_1\in \mathcal{P}_{m_1}\\\cdots\\\sigma_L\in \mathcal{P}_{m_L}}}\left[h_{V_{\sigma_L(1)}},\cdots [h_{V_{\sigma_1(m_1)}},P]\right].
   \end{align}
   Superficially, if the evolution time $t<1/(2eL(k\mathfrak{d}))$, the number of involved cluster terms 
   \begin{align}
        M=\frac{\log(1/\epsilon)-L\log(1-2teL(k\mathfrak{d}))}{L\log(1/(2teL(k\mathfrak{d})))},
   \end{align}
   otherwise we have
   \begin{align}
   M=e^{\pi teL(k\mathfrak{d})/\kappa}\log\left[\frac{1}{\epsilon}\frac{e^{\pi teL(k\mathfrak{d})/\kappa}-1}{(1-\kappa)^L}\right],
   \end{align}
   where the parameter $\kappa\in\mathcal{O}(1)$. 
\end{theorem}

\subsection{Short time Hamiltonian dynamics}\label{sec:short time}
   Let $t=\max_{k\in[L]}\{\abs{t_k}\}$, we first consider the scenario $\abs{t}\leq 1/(2eL(k\mathfrak{d})$, where the constant $\mathfrak{d}$ represents the maximum degree of the Hamiltonian interaction graph. ($L\mathfrak{d}$ represents the maximum degree of the $m$-sized graph $G^L_{\bm V_1\cdots\bm V_L}$ which connects to $P$.) Now we study the convergence of the cluster expansion
  $$U_P(\vec{t})=\sum\limits_{\substack{m_1\geq0\\\cdots\\m_L\geq 0}}\sum\limits_{\bm V_1\cdots, \bm V_L\in\mathcal{G}_m^{L,P}}\frac{\prod_{k=1}^L(\bm\lambda^{\bm V_k}(-it_k)^{m_k})}{\prod_{k=1}^L\bm V_k!m_k!}\sum\limits_{\substack{\sigma_1\in \mathcal{P}_{m_1}\\\cdots\\\sigma_L\in \mathcal{P}_{m_L}}}\left[h_{V_{\sigma_L(1)}},\cdots \left[h_{V_{\sigma_1(m_1)}}\cdots\left[h_{V_{\sigma_1(m_1)}},P\right]\right]\right]$$
up to index $m_1,m_2,\cdots m_L\leq M$. Specifically, let $m=m_1+\cdots+m_L$, we have
 \begin{eqnarray}
        \begin{split}
       \epsilon_M(\vec{t})=&\|U_P(\vec{t})-V_P(\vec{t})\|\\
       =&\left\|\sum\limits_{\substack{m_1\geq M+1\\\cdots\\m_L\geq M+1}}\sum\limits_{\bm V_1\cdots, \bm V_L\in\mathcal{G}_m^{L,P}}\frac{\prod_{k=1}^L(\bm\lambda^{\bm V_k}(-it_k)^{m_k})}{\prod_{k=1}^L\bm V_k!m_k!}\sum\limits_{\substack{\sigma_1\in \mathcal{P}_{m_1}\\\cdots\\\sigma_L\in \mathcal{P}_{m_L}}}\left[h_{V_{\sigma_L(1)}},\cdots [h_{V_{\sigma_1(m_1)}},P]\right]\right\|\\
        \leq &\sum\limits_{m_1,\cdots,m_L\geq M+1}\sum\limits_{\bm V_1\cdots, \bm V_L\in\mathcal{G}_m^{L,P}}\frac{{\bm\lambda}^{\bm V_1}\cdots {\bm\lambda}^{\bm V_L}(2t_1)^{m_1}\cdots (2t_L)^{m_L}}{(\bm V_1!\cdots\bm V_L!)}\left\|P\right\|\\
        \leq&\sum\limits_{m_1,\cdots,m_L\geq M+1}(2t_1)^{m_1}\cdots(2t_L)^{m_L}\abs{\mathcal{G}_{m}^{L,P}}\|P\|\\
        \leq&\|P\|\sum\limits_{m_1,\cdots,m_L\geq M+1}(2t_1)^{m_1}\cdots(2t_L)^{m_L}\abs{eL(k\mathfrak{d})}^{m_1+\cdots m_L}\\
        \leq &\|P\|\left[\sum\limits_{l\geq M+1}(2teL(k\mathfrak{d}))^{l}\right]^L,
        \end{split}
        \label{Eq:app}
    \end{eqnarray}
where $t=\max_{k\in[L]}\{\abs{t_k}\}$. The second line is valid since $\left\|\left[h_{V_{\sigma_L(1)}},\cdots [h_{V_{\sigma_1(m_1)}},P\right]\right\|\leq 2^{m_1+\cdots +m_L}\max\|h_i\|\|P\|\leq 2^m\|P\|$, and the fifth line comes from $\abs{\mathcal{G}_m^{L,P}}\leq (eL(k\mathfrak{d}))^m$~(refers to proposition~3.6 in Ref.~\cite{haah2024learning}).

As a result, when $t\leq 1/(2eL\mathfrak{d})$, we have 
\begin{align}
    \epsilon_M(\vec{t})\leq \|P\|\frac{(2teLk\mathfrak{d})^{L(M+1)}}{(1-2teLk\mathfrak{d})^L}.
\end{align}
Let $\epsilon=\frac{(2teLk\mathfrak{d})^{L(M+1)}}{(1-2teLk\mathfrak{d})^L}$, this results in 
\begin{align}
    M=\frac{\log(1/\epsilon)-L\log(1-2teLk\mathfrak{d})}{L\log(1/(2teLk\mathfrak{d}))}.
\end{align}

\subsection{Arbitrary constant time Hamiltonian dynamics}
\label{Sec:meanvalueanalytic}

    Noting that above process can be further generalized to an arbitrary constant time $t$ by means of analytic continuation. Consider the radius of a disk $R>1$, the analytic continuation can be achieved by using the map $t\mapsto t\phi(z)$, where the complex function $$\phi(z)=\frac{\log(1-z/R^{\prime})}{\log(1-1/R^{\prime})}$$ maps a disk onto an elongated region along the real axis~\cite{wild2023classical}. Here, the parameter $R^{\prime}>R$, and $\phi(z)$ is analytic on the closed desk $D_R=\{z\in\mathbb{C}:\abs{z}\leq R\}$. Meanwhile, $\phi(z)$ satisfies $\phi(0)=0$, $\phi(1)=1$ and we select the branch ${\rm Im}(\phi(z))\leq -\pi/(2\log(1-1/R^{\prime}))$.
    
    We consider the function 
     \begin{align}
         f(z)=\prod_{k=1}^Le^{iH^{(k)}t_k\phi(z)}P\prod_{k=1}^Le^{-iH^{(k)}t_k\phi(z)}
     \end{align}
     on the region $\abs{z}\leq sR$ where $s\in(0,1)$. Consider a curve ${\mathcal{C}}^{\prime}=\{\abs{w}=R\}$, according to the Cauchy integral method, we have
    \begin{eqnarray}
    \begin{split}
         f(z)=&\frac{1}{2\pi i}\oint_{{\mathcal{C}}^{\prime}}\frac{f(w)}{w-z}dw\\
         =&\frac{1}{2\pi i}\oint_{{\mathcal{C}}^{\prime}}\frac{f(w)}{w}\left(1-\frac{z}{w}\right)^{-1}dw\\
         =&\frac{1}{2\pi i}\oint_{{\mathcal{C}}^{\prime}}\frac{f(w)}{w}\left(\sum\limits_{k=0}^M\left(\frac{z}{w}\right)^k+\left(\frac{z}{w}\right)^M\left(1-\frac{z}{w}\right)^{-1}\right)dw\\
         =&\sum\limits_{k=0}^M\frac{1}{2\pi i}\oint_{{\mathcal{C}}^{\prime}}\frac{f(w)}{w^k}z^k+\frac{1}{2\pi i}\oint_{{\mathcal{C}}^{\prime}}\frac{f(w)}{w-z}\left(\frac{z}{w}\right)^{M+1}dw\\
         =&\sum\limits_{k=0}^M\frac{f^{(k)}(0)}{k!}z^k+\frac{1}{2\pi i}\oint_{{\mathcal{C}}^{\prime}}\frac{f(w)}{w-z}\left(\frac{z}{w}\right)^{M+1}dw.
    \end{split}
    \end{eqnarray}
As a result, the truncated error can be upper bounded by
 \begin{eqnarray}
    \begin{split}
    \left\|f(z)-\sum\limits_{k=0}^M\frac{f^{(k)}(0)}{k!}z^k\right\|_2=& \left\|\frac{1}{2\pi i}\oint_{{\mathcal{C}}^{\prime}}\frac{f(w)}{w-z}\left(\frac{z}{w}\right)^{M+1}dw\right\|_2\\
    \leq&\frac{1}{2\pi}\oint_{{\mathcal{C}}^{\prime}}\frac{\|f(w)\|_2}{\|w-z\|}\left\|\frac{z}{w}\right\|^{M+1}dw.
\end{split}
\end{eqnarray}
We require the following result to evaluate the upper bound of $\|f(w)\|_2$.

\begin{definition}[Multi-variable complex analytic function]
    Suppose $g:D\mapsto\mathbb{C}$ be a function on the domain $D\subset\mathbb{C}^L$, if for any vector $ t\in D$, there exists a $r$-radius cylinder centered on $ t$, such that 
    \begin{align}
        g(\vec{w})=\sum\limits_{\alpha_1,\cdots,\alpha_L\geq 0}c_{\vec{\alpha}}(w_1- t_1)^{\alpha_1}\cdots(w_L- t_L)^{\alpha_L},
    \end{align}
    then $g$ is analytic on the point $ t=( t_1,\cdots, t_L)$.
\end{definition}

\begin{lemma}
   Given complex values $\vec{w}=(w_1,\cdots,w_L) \in\mathbb{C}^L$, if $\abs{{\rm Im}(w_l)}\leq 1/(2eLk\mathfrak{d})$ for all $l\in[L]$, we have
    \begin{align}
        \|U_P(\vec{w})\|\leq \frac{\|P\|}{(1-2\abs{\max_l{\rm Im}(w_l)}eL(k\mathfrak{d}))^L},
    \end{align}
    where $\mathfrak{d}$ represents the maximum degree of the interaction graph induced by Hamiltonian $H$, and $k$ represents the locality of Pauli operator $P$. 
    \label{lemma:upperboundnorm}
\end{lemma}
\begin{proof}
     Eq.~\ref{Eq:app} provides an approximation to $U_P(\vec{t})$ when $\max_l\abs{t_l}\leq1/(2eL(k\mathfrak{d}))$, in other word, $U_P(\vec{t})$ remains analytic for all complex values $t_l\in\mathbb{C}$ in the range $\abs{t_l}<1/(2eLk\mathfrak{d})$. In the following, we substitute $w_k$ by $t_k$, and consider 
     $U_P(\vec{w})=\prod_{l=1}^Le^{iH^{(l)}(w_l-{\rm Re}(w_l))}e^{iH^{(l)} {\rm Re}(w_l)}P\prod_{l=1}^Le^{-iH^{(l)}(w_l-{\rm Re}(w_l))}e^{-iH^{(l)} {\rm Re}(w_l)}$. Equivalently, we have
     \begin{align}
         U_P(\vec{t})=\sum\limits_{l_1\cdots,l_L\geq 0}u_{l_1,\cdots,l_L}(w_l-{\rm Re}(w_l))^{l_1}\cdots(w_l-{\rm Re}(w_l))^{l_L}
     \end{align}
     for some operators $u_{l_1,\cdots,l_L}$, which naturally implies $U_P(\vec{w})$ is analytic for all complex values of $\vec{w}$ on a disk in the complex plane of radius $1/(2eLk\mathfrak{d})$ around any point on the real axis. 

Since the operator norm is unitary invariant, we have 
\begin{eqnarray}
\begin{split}
     \|e^{-iw_1H^{(1)}}Pe^{iw_1H^{(1)}}\|&=\|e^{-i{\rm Re}(w_1)H^{(1)}}e^{-i(w_1-{\rm Re}(w_1))H^{(1)}}Pe^{i(w_1-{\rm Re}(w_1))H^{(1)}}e^{i{\rm Re}(w_1)H^{(1)}}\|\\
     &=\|e^{-i(w_1-{\rm Re}(w_1))H^{(1)}}Pe^{i(w_1-{\rm Re}(w_1))H^{(1)}}\|.
\end{split}
\end{eqnarray}
When $\abs{{\rm Im}(w_1)}<1/(eLk\mathfrak{d})$, the support of the above operator is limited into a small region which only covers constant number of qubits.
According the cluster expansion method, we have
\begin{eqnarray}
    \begin{split}
        &\left\|e^{-iw_2H^{(2)}}\left(e^{-i(w_1-{\rm Re}(w_1))H^{(1)}}Pe^{i(w_1-{\rm Re}(w_1))H^{(1)}}\right)e^{iw_2H^{(2)}}\right\|\\
        &
    \left\|e^{-i(w_2-{\rm Re}(w_2))H^{(2)}}\left(e^{-i(w_1-{\rm Re}(w_1))H^{(1)}}Pe^{i(w_1-{\rm Re}(w_1))H^{(1)}}\right)e^{i(w_2-{\rm Re}(w_2))H^{(2)}}\right\|\\\leq& \frac{1}{1-2\abs{{\rm Im}(w_2)}ek\mathfrak{d}}\left\|e^{-i(w_1-{\rm Re}(w_1))H^{(1)}}Pe^{i(w_1-{\rm Re}(w_1))H^{(1)}}\right\|.
    \end{split}
\end{eqnarray}
Repeat this process $L-1$ times, we finally obtain the result
\begin{eqnarray}
    \begin{split}
       \|U_P(\vec{w})\|=&\Big\|e^{-i(w_L-{\rm Re}(w_L))H^{(L)}}e^{-i{\rm Re}(w_L)H^{(L)}}\cdots e^{-i(w_1-{\rm Re}(w_1))H^{(1)}}e^{-i{\rm Re}(w_1)H^{(1)}}P\\
        &e^{i{\rm Re}(w_1)H^{(1)}}e^{i(w_1-{\rm Re}(w_1))H^{(1)}}\cdots e^{i{\rm Re}(w_L)H^{(L)}}e^{i(w_L-{\rm Re}(w_L))H^{(L)}}\Big\|\\
        \leq&\left(\frac{1}{1-(2\max_k\abs{{\rm Im}(w_k)}ek\mathfrak{d})}\right)^{L-1}\|e^{-i(w_1-{\rm Re}(w_1))H^{(1)}}Pe^{i(w_1-{\rm Re}(w_1))H^{(1)}}\|.
    \end{split}
    \label{Eq:upbound1}
\end{eqnarray}
For square matrices $A$ and $B$, the BCH expansion enables us to write the cluster expansion to $e^{tA}Be^{-tA}$~\cite{haah2024learning} for $t\in\mathbb{R}$. As a result, we have
  \begin{eqnarray}
      \begin{split}
           \|U_P(\vec{w})\|\leq\frac{\|P\|}{(1-2\abs{\max_k{\rm Im}(w_k)}ekL\mathfrak{d})^L}.
      \end{split}
      \label{Eq:upbound2}
  \end{eqnarray}
\emph{Added Note:} We note that inequality~\ref{Eq:upbound1} only requires the condition $\max_k\abs{{\rm Im}(w_k)}<1/(2ek\mathfrak{d})$, however, Eq.~\ref{Eq:app} provides an approximation to $U_P(\vec{t})$ when $\max_l\abs{t_l}\leq1/(2eL(k\mathfrak{d}))$. Taking the intersection of these two regimes, we finally obtain inequality~\ref{Eq:upbound2}, which is the upper bound to inequality~~\ref{Eq:upbound1}.
\end{proof}

Recall that 
$$f(w)=\prod_{k=1}^Le^{iH^{(k)}t_k\phi(w)}P\prod_{k=1}^Le^{-iH^{(k)}t_k\phi(w)}$$
where $\vec{t}\in\mathbb{R}^L$ and ${\rm Im}(\phi(w))\leq -\pi/(2\log(1-1/R^{\prime}))$. Assign $\vec{t}\phi(w)$ to $\vec{w}$ given in Lemma~\ref{lemma:upperboundnorm}, then Lemma~\ref{lemma:upperboundnorm} implies 
\begin{eqnarray}
\begin{split}
     \|f(w)\|=\|U_P(\phi(w)\vec{t})\|&\leq \frac{\|P\|}{(1-2\abs{\max_k{\rm Im}(t_k\phi(w))}eLk\mathfrak{d})^L}\\
     &\leq \frac{\|P\|}{(1+\pi teLk\mathfrak{d}/(\log(1-1/R^{\prime})))^L}
     \label{Eq:fupperbound}
\end{split}
\end{eqnarray}
for all $w\in C^{\prime}=\{\abs{w}=R\}$. This further results in 
\begin{eqnarray}
    \begin{split}
    \left\|f(z)-\sum\limits_{k=0}^M\frac{f^{(k)}(0)}{k!}z^k\right\|=& \left\|\frac{1}{2\pi i}\oint_{{\mathcal{C}}^{\prime}}\frac{f(w)}{w-z}\left(\frac{z}{w}\right)^{M+1}dw\right\|\\
    \leq&\frac{1}{2\pi}\oint_{{\mathcal{C}}^{\prime}}\frac{\|f(w)\|}{\|w-z\|}\left\|\frac{z}{w}\right\|^{M+1}dw\\
    \leq&\max\{\|f(w)\|\}\frac{s^{M+1}}{(1-s)}
    \label{Eq:analyticerror}
\end{split}
\end{eqnarray}
where the last line follow from the fact that $\abs{w-z}\geq R(1-s)$, $\abs{z}\leq sR$ and $\|w\|=R$. Combine inequalities~\ref{Eq:fupperbound} and \ref{Eq:analyticerror}, we have
\begin{align}
    \left\|f(z)-\sum\limits_{k=0}^M\frac{f^{(k)}(0)}{k!}z^k\right\|\leq \frac{\|P\|s^{M+1}}{(1+\pi teLk\mathfrak{d}/(\log(1-1/R^{\prime})))^L(1-s)}.
    \label{Eq:error1}
\end{align}
Let $\kappa=\frac{-\pi teLk\mathfrak{d}}{\log(1-1/R^{\prime})}$, $R^{\prime}$ can be further expressed by
\begin{align}
    \frac{1}{R^{\prime}}=1-e^{-\pi teLk\mathfrak{d}/\kappa}.
\end{align}
Since the parameter $R^{\prime}>R$, we can always select $R$ such that $(R^{\prime})^M(R^{\prime}-1)=2R^M(R-1)$ holds. Substitute this relationship into the approximation upper bound given by~\ref{Eq:error1} and assign $s=1/R$, we finally obtain
\begin{align}
    \epsilon= \frac{s^{M+1}}{\left(1+\frac{\pi teLk\mathfrak{d}}{\log(1-1/R^{\prime})}\right)^L(1-s)}=\frac{1}{(1-\kappa)^L}\left(1-e^{-\pi teLk\mathfrak{d}/\kappa}\right)^M\left(e^{\pi teLk\mathfrak{d}/\kappa}-1\right).
\end{align}
This implies truncating at order
\begin{align}
    M(t)=\frac{\log\left[\frac{1}{\epsilon}\frac{e^{\pi teLk\mathfrak{d}/\kappa}-1}{(1-\kappa)^L}\right]}{\log\left[e^{\pi teLk\mathfrak{d}/\kappa}/\left(e^{\pi teLk\mathfrak{d}/\kappa}-1\right)\right]}\approx e^{\pi teLk\mathfrak{d}/\kappa}\log\left[\frac{1}{\epsilon}\frac{e^{\pi teLk\mathfrak{d}/\kappa}-1}{(1-\kappa)^L}\right].
\end{align}

\section{Proof of Lemma~\ref{lemma:VTCompute}}
\label{AppendixC}
\begin{lemma}[Lemma~\ref{lemma:VTCompute} in Appendix~\ref{Method}]
    There exists a classical algorithm that can exactly output $V_P(\vec{t})$ in 
    $$\mathcal{O}((e^{\pi teL|P|\mathfrak{d}}/\epsilon)^{e^{\pi teL|P|\mathfrak{d}}})$$ running time such that $\|V_P(\vec{t})-U_P(\vec{t})\|\leq\epsilon$.
\end{lemma}

\begin{proof}
We consider the polynomial expression of the function
\begin{eqnarray}
    \begin{split}
        f(z)&=\prod_{k=1}^Le^{iH^{(k)}t_k\phi(z)}P\prod_{k=1}^Le^{-iH^{(k)}t_k\phi(z)}\\
        &=\sum\limits_{\substack{m_1\geq0\\\cdots\\m_L\geq 0}}\sum\limits_{\bm V_1\cdots, \bm V_L\in\mathcal{G}_m^{L,P}}\frac{\prod_{k=1}^L(\bm\lambda^{\bm V_k}(-it_k\phi(z))^{m_k})}{\prod_{k=1}^L\bm V_k!m_k!}\sum\limits_{\substack{\sigma_1\in \mathcal{P}_{m_1}\\\cdots\\\sigma_L\in \mathcal{P}_{m_L}}}\left[h_{V_{\sigma_1(1)}},\cdots \left[h_{V_{\sigma_1(m_1)}}\cdots\left[h_{V_{\sigma_L(m_L)}},P\right]\right]\right]\\
        &=\sum\limits_{l_1=0,\cdots,l_L=0}^{+\infty}A_{l_1,\cdots,l_L}t_{1}^{l_1}\cdots t_L^{l_L}[\phi(z)]^{l_1+\cdots+l_L}
    \end{split}
\end{eqnarray}
where the first equality comes from Eq.~\ref{Eq:Ut} combined with discarding disconnected clusters, and $A_{l_1,\cdots,l_L}$ represents the operator which is independent to variables $\{t_k,\phi(z)\}_{k=1}^L$. In Appendix~\ref{AppendixA}, we have known that $V_P(\vec{t})$ is essentially the approximation to $f(1)$ up to $M$ degree, that is
\begin{align}
    V_P(\vec{t})=\sum\limits_{m=0}^M\frac{f^{(m)}(0)}{m!}.
\end{align}
As a result, computing gradient functions $f^{(m)}(0)$ for index $m\in[M]$ suffice to exactly compute $V_P(\vec{t})$.

Recall that $\phi(z)=\log((1-z/R^{\prime})/(1-1/R^{\prime}))=\sum_{l=0}^{+\infty}\phi_lz^l$, where $\phi_l=\frac{1}{l(R^{\prime})^l\log(1-1/R^{\prime})}$ for $l\geq 1$ and $\phi_0=0$. This enables us to compute 
\begin{eqnarray}
\begin{split}
    &\frac{{\rm d}^mf(z)}{{\rm d}z^m}\big|_{z=0}\\
    =&\sum\limits_{l_1,\cdots,l_L\geq 0}A_{l_1\cdots l_L}t_1^{l_1}\cdots t_L^{l_L}\frac{{\rm d}^m}{{\rm d}z^m}\left[\sum\limits_{s=0}^{\infty}\phi_sz^s\right]^{l_1+\cdots l_L}\big|_{z=0}\\
    =&\sum\limits_{l_1,\cdots,l_L\geq 0}A_{l_1\cdots l_L}t_1^{l_1}\cdots t_L^{l_L}\left[\sum\limits_{s_1,\cdots,s_l\geq 1}^{\infty}\phi_{s_1}\phi_{s_2}\cdots\phi_{s_l}(s_1+\cdots+s_l)\cdots(s_1+\cdots+s_l-m+1)z^{s_1+\cdots+s_l-m}\right]\big|_{z=0}
    \\=&\sum\limits_{l_1+\cdots +l_L=1}^{m}A_{l_1,\cdots,l_L}t_1^{l_1}\cdots t_L^{l_L}\sum\limits_{\substack{s_1,\cdots,s_l\geq 1\\ s_1+\cdots+s_l=m}}\phi_{s_1}\cdots\phi_{s_l}m!,
\end{split}
\end{eqnarray}
where the index $l=l_1+\cdots+l_L$. Now let us explain the third equality. When $z=0$, only terms with $s_1+\cdots +s_l=m$ may not vanish. Meanwhile $l>m$ may result in some non-negative index $s_{l^*}=0$ for $l^*\in[l]$ which implies $\phi_{s_{l^*}}=0$, as a consequence, we have $l\leq m$. Noting that the nested commutator $A_{l_1,\cdots,l_L}t_1^{l_1}\cdots t_L^{l_L}$ can be numerically evaluated in time $e^{\mathcal{O}(m)}$. We refer readers to Appendix~A in Ref.~\cite{wild2023classical} to find more details on computing $A_{l_1,\cdots,l_L}$. The derivative function $\frac{1}{M!}\frac{{\rm d}^Mf(z)}{{\rm d}z^M}\big|_{z=0}$ thus can be exactly computed in $\mathcal{O}\left(\exp(M)\right)$ classical running time, with 
$$   M=e^{\pi teL|P|\mathfrak{d}/\kappa}\log\left[\frac{1}{\epsilon}\frac{e^{\pi teL|P|\mathfrak{d}/\kappa}-1}{(1-\kappa)^L}\right].$$ 

\end{proof}

\section{Dequantization on Guided Local Hamiltonian Problem}
\label{App:groundstate}

Here, our fundamental idea is to construct a filter function in terms of the guiding state $|\psi_c\rangle$, that is
\begin{align}
    C(x)=(F_\sigma*P)(x)=\sum\limits_{j=0}^{2^n-1}p_jF_{\sigma}(x-E_j)=\int_{-\infty}^{\infty}\hat{F}_{\sigma}(t)\langle\psi_c|e^{-i2\pi t(H-xI)}|\psi_c\rangle dt,
\end{align}
where $P(x)=\sum_{j=0}^{2^n-1}p_j\delta(x -E_j)$ is the spectral function of the guiding state $|\psi_c\rangle$, $p_j=\abs{\langle\phi_j|\psi_c\rangle}^2$, $\delta(\cdot)$ is the Dirac delta function. The Gaussian derivative filter function $F_{\sigma}(t)=-te^{-t^2/(2\sigma^2)}/(\sigma^3\sqrt{2\pi})$,  $\hat{F}_{\sigma}(t)$ represents its Fourier transform, and the variance parameter will be defined latter.

Since the Gaussian derivative filter has an exponentially-decaying tail, $(F_{\sigma}*P)(x)$ is dominated by $\gamma F_{\sigma}(x-E_0)$ in the neighborhood of $E_0$, which decades monotonically to zero when $x\leq E_0$ approaches to $E_0$. Hence, our scheme to pinpoint the ground-state energy is to estimate $C(x)$ beginning at $x=E_a$, a lower bound to $E_0$, and an interval of $\varepsilon$. The algorithm outputs the estimation $\hat{E}_0$ when $C(x)$ decays below the termination threshold $\mathcal{O}(\gamma\epsilon\sigma^{-3})$, outlined as following lemmas~\cite{wang2023quantum}.

\begin{lemma}[Ref.~\cite{wang2023quantum}]
    Let $c=\mathcal{O}(1)$, and suppose $\epsilon>0$ such that $\epsilon\leq c\min\left(0.9\Delta/\sqrt{2\ln(9\Delta\epsilon^{-1}\gamma^{-1})},0.2\Delta\right)$. Then for
    \begin{align}
        \sigma=\min\left(0.9\Delta/\sqrt{2\ln(9\Delta\epsilon^{-1}\gamma^{-1})},0.2\Delta\right)
    \end{align}
    we have 
    \begin{align}
        \abs{(F_{\sigma}*P)(x)}\leq \frac{0.6\epsilon\gamma}{\sqrt{2\pi}\sigma^3}
    \end{align}
    for any $x\in[E_0-0.5\epsilon, E_0+0.5\epsilon]$.
    Meanwhile
    \begin{align}
        \abs{(F_{\sigma}*P)(x)}> \frac{0.8\epsilon \gamma}{\sqrt{2\pi}\sigma^3}
    \end{align}
    for any  $x\in[E_0-0.5\sigma, E_0-0.5\epsilon)\cup (E_0+\epsilon, E_0+0.5\sigma]$.
    \label{lemma:accuracy1}
\end{lemma}

Above results indicate that when $\sigma=\tilde{\mathcal{O}}(\Delta)$, estimating the filter function $C(x)$ up to $\epsilon_1=\mathcal{O}(\epsilon \gamma\Delta^{-3})$ sufficing to provide a $\epsilon$-approximation to the ground state energy. To efficiently compute the filter function $C(x)$, we need to truncate the infinity integral up to a finite order $T$, which is promised by the following result.

\begin{lemma}[Lemma~A.4 in Ref.~\cite{wang2023quantum}]
    Let $\epsilon_1>0$, then for 
    \begin{align}
        T=\pi^{-1}\sigma^{-1}\sqrt{2\ln(8\pi^{-1}\epsilon_1^{-1}\sigma^{-2})},
    \end{align}
    we have
    \begin{align}
        \abs{\int_{-\infty}^{+\infty}\hat{F}_{\sigma}(t)e^{2\pi ixt}dt-\int_{-T}^{+T}\hat{F}_{\sigma}(t)e^{2\pi ixt}dt}\leq \frac{\epsilon_1}{2}
    \end{align}
    for any $x\in\mathbb{R}$.
    \label{lemma:time}
\end{lemma}

\subsection{Evaluating the filter function}

Given the  filter function $$F_{\sigma,T}(x)=\int_{-T}^T\hat{F}_{\sigma}(t)e^{2\pi itx}dt,$$
we define a probability distribution 
\begin{align}
    q(t)=\frac{|\hat{F}_{\sigma}(t)|}{\|\hat{F}_{\sigma}(t)\|_1}
\end{align}
for any $t\in[-T,T]$. Let the phase function $\phi(t)$ be the phase of $\hat{f}_T(t)$, that is $\hat{f}_T(t)=|\hat{f}_T(t)|e^{2\pi\phi(t)}$. Then $F_{\sigma,T}(x)$ can be rewritten as an expectation value 
\begin{align}
    F_{\sigma,T}(x)=\int_{-T}^T\|\hat{F}_{\sigma}(t)\|_1e^{2\pi i(tx+\phi(t))}q(t)dt.
\end{align}

Now we define the random variable 
\begin{align}
    Z(x)=\|\hat{F}_{\sigma}(t)\|_1e^{2\pi i(tx+\phi(t))}Z_t,
\end{align}
where the random variable $Z_t$ satisfies $\mathbb{E}[Z_t]=\langle\psi_c|e^{-iHt}|\psi_c\rangle$. It is shown that 
\begin{align}
    \mathbb{E}[Z(x)]=(F_{\sigma,T}*P)(x)
\end{align}
for any $x\in\mathbb{R}$.

\begin{lemma}[\cite{wang2023quantum}]
    Let $\{(t^{(i)},Z^{(i)})\}_{i=1}^N$ be $N$ i.i.d samples such that $t^{(i)}\sim q(t)$, $Z^{(i)}\sim Z_{t_i}$, and $x_1,x_2,\cdots x_M\in\mathbb{R}$ be arbitrary. For each $x_m$, if the sample complexity 
    \begin{align}
        S=\mathcal{O}\left(\frac{\|\hat{F}_{\sigma,T}\|_1^2\log(N/\delta)}{\epsilon_1^2}\right)
    \end{align}
    the estimator 
    \begin{align}
        Z_m=\frac{\|\hat{F}_{\sigma,T}\|_1}{S}\sum\limits_{l=1}^Se^{2\pi i(t^{(l)}x_m+\phi(t^{(i)}))}Z^{(i)}
    \label{filterestimator}
    \end{align}
    satisfies 
    \begin{align}
        {\rm Pr}\left[\abs{Z_m-(F_{\sigma,T}*P)(x_m)}\leq\epsilon_1\right]\geq 1-\delta
    \end{align}
    for any $x_1,x_2,\cdots x_N$.
    \label{lemma:samplecomplexity}
\end{lemma}

Now we can summarize the whole process in evaluating the ground state energy. Suppose a ground state lower bound $E_a$ is provided such that $\abs{E_a-E_0}\leq \mathcal{O}(\sigma)$, then a $\epsilon$-step length search is performed from the initial grid $x\in\{E_a+\epsilon, E_a+2\epsilon,\cdots\}$ until the condition $ \abs{(F_{\sigma}*P)(x)}\leq \frac{0.6\epsilon \gamma}{\sqrt{2\pi}\sigma^3}$ is achieved. For each trail energy $x$, the filter function $(F_{\sigma}*P)(x)$ can be estimated by $Z_x$ (given by Eq.~\ref{filterestimator}). Given the fact that $\|\hat{F}_{\sigma,T}\|_1=\mathcal{O}(\sigma^{-2})$, let $\epsilon_1= \frac{0.6\epsilon\gamma}{\sqrt{2\pi}\sigma^3}$, and $N=\mathcal{O}(\sigma/\epsilon)$, the sample complexity can be estimated by $\mathcal{O}(\gamma^{-2}\epsilon^{-3}\Delta^3)$, the maximum evolution time 
\begin{align}
    T=\tilde{\mathcal{O}}\left(\Delta^{-1}\sqrt{2\ln(8\pi^{-1}\epsilon_1^{-1}\Delta^{-2})}\right)
\end{align}
(lemma~\ref{lemma:time}). As a result, the dequantization algorithm only needs to generate a random variable $Z_t$ such that $\mathbb{E}[Z_t]=\langle\psi_c|e^{-iHt}|\psi_c\rangle$, in other words, provide an estimation to $\langle\psi_c|e^{-iHt}|\psi_c\rangle$ within a $\epsilon^{\prime}=\epsilon_1/\|\hat{F}_{\sigma,T}\|_1=\mathcal{O}(\epsilon_1\sigma^2)$ additive error.

Before giving the proof details, we summarize the errors:
\begin{itemize}
    \item $\epsilon$: the accuracy to estimate $E_0$;
    \item $\epsilon_1$: the accuracy to estimate $C(x)$
    \item $\epsilon^{\prime}$: the accuracy to estimate $\langle\psi_c|e^{-iHt}|\psi_c\rangle$
\end{itemize}

Lemmas~\ref{lemma:accuracy1},~\ref{lemma:time},~\ref{lemma:samplecomplexity} provide relationships between above accuracy parameters.

\section{Classical Algorithm for bounded additive error}
\label{sec:small_cluster}
We first consider the cluster expansion of $D_{ t}(H)=\langle\psi_c|e^{-iHt}|\psi_c\rangle$. Using the Taylor expansion formula, we have
\begin{align}
    e^{-iHt}=\sum\limits_{m\geq 0}\frac{t^m}{m!}\left(\frac{\partial^m e^{-iHt}}{\partial t^m}\right)_{t=0}.
    \label{Eq:origional}
\end{align}
Recall that $H=\sum_X\lambda_Xh_X$, then we define $Z_X=-it\lambda_X$ and $Z=(Z_{X_1},Z_{X_2},\cdots)$. As a result, we may write 
\begin{eqnarray}
\begin{split}
     \left(\frac{\partial^m e^{-iHt}}{\partial t^m}\right)_{t=0}&=\sum\limits_{X_1,\cdots,X_m}\left(\frac{\partial Z_{X_1}}{\partial t}\right)\cdots \left(\frac{\partial Z_{X_m}}{\partial t}\right)\frac{\partial^m e^{-iHt}}{\partial Z_{X_1}\cdots\partial Z_{X_m}}\big|_{Z=(0,0,\cdots,0)}\\
     &=\sum\limits_{X_1,\cdots,X_m}(-i)^m\lambda_{X_1}\cdots\lambda_{X_m}\frac{\partial^m}{\partial Z_{X_1}\cdots\partial Z_{X_m}}\sum\limits_{j\geq0}\frac{(- it)^j}{j!}H^j\big|_{Z=(0,0,\cdots,0)}\\
     &=\sum\limits_{X_1,\cdots,X_m}(-i)^m\lambda_{X_1}\cdots\lambda_{X_m}\frac{\partial^m}{\partial Z_{X_1}\cdots\partial Z_{X_m}}\sum\limits_{j\geq0}\frac{(- it)^j}{j!}\left(\sum\limits_{\lambda_1,\cdots,\lambda_j}\lambda_1\cdots\lambda_jh_{X_1}\cdots h_{X_j}\right)\big|_{Z=(0,0,\cdots,0)}\\
     &=\sum\limits_{X_1,\cdots,X_m}(-i)^m\lambda_{X_1}\cdots\lambda_{X_m}\frac{1}{m!}\sum\limits_{\sigma\in\mathcal{P}_m}h_{X_{\sigma(1)}}\cdots h_{X_{\sigma(m)}}.
     \label{Eq:partialderevate}
\end{split}
\end{eqnarray}
Taking Eq.~\ref{Eq:partialderevate} into Eq.~\ref{Eq:origional}, we finally obtain the cluster expansion formula
\begin{eqnarray}
    \begin{split}
         D_{t}(H)&=\sum\limits_{m\geq 0}\frac{(t^m}{m!}\sum\limits_{X_1,\cdots,X_m}(-i)^m\lambda_{X_1}\cdots\lambda_{X_m}\frac{1}{m!}\sum\limits_{\sigma\in\mathcal{P}_m}{\rm Tr}\left[|\psi_c\rangle\langle\psi_c|h_{X_{\sigma(1)}}\cdots h_{X_{\sigma(m)}}\right]\\
         &=\sum\limits_{m\geq 0}\sum\limits_{\bm W\in\mathcal{C}_m}\frac{(- it)^{\abs{\bm W}}{\bm\lambda}^{\bm W}}{\bm W!}\frac{1}{m!}\sum\limits_{\sigma\in\mathcal{P}_m}{\rm Tr}\left[|\psi_c\rangle\langle\psi_c|h_{X_{\sigma(1)}}\cdots h_{X_{\sigma(m)}}\right]\\
         &=1+\sum\limits_{m\geq 1}\sum\limits_{\bm W\in\mathcal{C}_m}\frac{(- it)^{\abs{\bm W}}{\bm\lambda}^{\bm W}}{\bm W!}\frac{1}{m!}\sum\limits_{\sigma\in\mathcal{P}_m}{\rm Tr}\left[|\psi_c\rangle\langle\psi_c|h_{X_{\sigma(1)}}\cdots h_{X_{\sigma(m)}}\right],
    \end{split}
\end{eqnarray}
where $\bm\lambda^{\bm W}=\prod_{X\in S}\lambda_W^{\mu_{\bm W}(X)}$ and $\bm W!=\prod_{X\in S}\mu_{\bm W}(X)!$. 

From the above expression, it is shown that each cluster $\bm W=(h_{X_1},\cdots,h_{X_m})$, and ${\rm Tr}\left[|\psi_c\rangle\langle\psi_c|h_{X_{\sigma(1)}}\cdots h_{X_{\sigma(m)}}\right]$ can be factorized when $\bm W$ is disconnected and $|\psi_c\rangle$ be the product state. In general, the classical initial state $|\psi_c\rangle$ may not represent a product state, but a superposition quantum state with $R\leq {\rm poly}(n)$ configurations. Suppose $|\psi_c\rangle=\sum_xa_x|x\rangle$ where $|x\rangle$ represents a tensor product state and coefficients satisfy $\sum_x\abs{a_x}^2=1$. As a result, the filter function can be decomposed by $$D_{ t}(H)=\sum_{x,y}a_xa_y^{*}\langle y|e^{-iHt}|x\rangle.$$ In the following sections, we focus on estimating $d_{x,y, t}(H)=\langle y|e^{-iHt}|x\rangle$.

When the cluster $\bm W$ is disconnected, the related term in the cluster expansion of $d_{x,y, t}(H)$ is given by
$$\frac{1}{m!}\sum_{\sigma\in\mathcal{P}_m}{\rm Tr}\left[|x\rangle\langle y|h_{X_{\sigma(1)}}\cdots h_{X_{\sigma(m)}}\right]=\prod\limits_{\bm V\in P_{c,\max}(\bm W)}\langle y|h^{\bm V}|x\rangle,$$
with $P_{c,\max}(\bm W)$ representing the partition of the cluster $\bm W$ into its maximal connected components. For the convenience of the following description, we denote $\mathcal{D}_{\bm W}(d_{x,y, t}(H))=(- t)^{\bm W}\prod\limits_{\bm V\in P_{c,\max}(\bm W)}\langle y|h^{\bm V}|x\rangle$, which naturally give rises to
\begin{align}
    d_{x,y, t}(H)=1+\sum\limits_{m\geq 1}\sum\limits_{\bm W\in\mathcal{C}_m}\prod\limits_{\bm V\in P_{c,\max}(\bm W)}\frac{\bm\lambda^{\bm V}}{\bm V!}\mathcal{D}_{\bm V}(d_{x,y, t}(H)).
\end{align}

Directly compute $d_{x,y, t}(H)$ is generally hard, however, we can compute the function $\ln(d_{x,y, t}(H))$ when the inverse temperature $ t$ is smaller than a threshold. We consider the formal Taylor series 
\begin{align}
    \ln(1+z)=\sum\limits_{k=1}^{\infty}\frac{(-1)^{k-1}}{k}z^k.
\end{align}
Taking $d_{x,y, t}(H)$ into above series may yield
\begin{eqnarray}
    \begin{split}
        \ln\left(d_{x,y, t}(H)\right)=\sum\limits_{k=1}^{\infty}\frac{(-1)^k}{k}\sum\limits_{m_1,\cdots,m_k}\sum\limits_{\bm W_1\in\mathcal{C}_{m_1},\cdots,\bm W_k\in\mathcal{C}_{m_k}}&\left(\prod\limits_{\bm V_1\in P_{c,\max}(\bm W_1)}\frac{\lambda^{\bm V_1}}{\bm V_1!}\mathcal{D}_{\bm V_1}(d_{x,y, t}(H))\right)\cdots\\
        &\cdots\left(\prod\limits_{\bm V_k\in P_{c,\max}(\bm W_k)}\frac{\lambda^{\bm V_k}}{\bm V_1!}\mathcal{D}_{\bm V_k}(d_{x,y, t}(H))\right).
    \end{split}
    \label{Eq:logcluster}
\end{eqnarray}
Using a similar method mentioned in Ref.~\cite{wild2023classical}, we can re-group the sums over all clusters, that is
\begin{align}
    \ln\left(d_{x,y, t}(H)\right)=\sum\limits_{m\geq 1}\sum\limits_{\bm W\in\mathcal{G}_m}\sum\limits_{P\in\mathcal{P}_c(\bm W)}C(P)\prod\limits_{\bm V\in P}\frac{\lambda^{\bm V}}{\bm V!}\mathcal{D}_{\bm V}(d_{x,y, t}(H)).
\end{align}
Here, $\mathcal{G}_m$ represents the set of all connected clusters $\bm W$ with size $m$, $\mathcal{P}_c(\bm W)$ represents the set of all partitions $P$ to the cluster $\bm W$, and the coefficient $C(P)$ can be determined by considering the different ways in which the partition $P$ can be generated by the clusters $\bm W_1,\cdots,\bm W_k$ given by Eq.~\ref{Eq:logcluster}.

\begin{lemma}[Proposition~9 in Ref.~\cite{wild2023classical}]
    Let the cluster $\bm W\in\mathcal{G}_m$, then we have
    \begin{align}
        \abs{\sum\limits_{P\in\mathcal{P}_c(\bm W)}C(P)\prod\limits_{\bm V\in P}\frac{\lambda^{\bm V}}{\bm V!}\mathcal{D}_{\bm V}(d_{x,y, t}(H))}\leq \left[2e(\mathfrak{d}+1)\abs{ t}\right]^m.
    \end{align}
    \label{lemma:error1}
\end{lemma}
Above lemma enables us to approximate $\ln\left(d_{x,y, t}(H)\right)$ by truncating the cluster expansion up to $M\leq\mathcal{O}(\ln(\abs{S}/\epsilon))$ order when $\abs{ t}<1/2e^2\mathfrak{d}(\mathfrak{d}+1)$, that is
\begin{eqnarray}
    \begin{split}
        &\abs{\ln\left(d_{x,y, t}(H)\right)-\sum\limits_{m=1}^M\sum\limits_{\bm W\in\mathcal{G}_m}\sum\limits_{P\in\mathcal{P}_c(\bm W)}C(P)\prod\limits_{\bm V\in P}\frac{\lambda^{\bm V}}{\bm V!}\mathcal{D}_{\bm V}(d_{x,y, t}(H))}\\
        =&\abs{\sum\limits_{m\geq M+1}\sum\limits_{\bm W\in\mathcal{G}_m}\sum\limits_{P\in\mathcal{P}_c(\bm W)}C(P)\prod\limits_{\bm V\in P}\frac{\lambda^{\bm V}}{\bm V!}\mathcal{D}_{\bm V}(d_{x,y, t}(H))}\\
        \leq&\sum\limits_{m\geq M+1}\sum\limits_{\bm W\in\mathcal{G}_m}\left[2e(\mathfrak{d}+1)\abs{ t}\right]^m\\
        \leq&\sum\limits_{m\geq M+1}\abs{S}\left[2e^2\mathfrak{d}(\mathfrak{d}+1)\abs{ t}\right]^m\\
        =&\frac{\abs{S}\left[2e^2\mathfrak{d}(\mathfrak{d}+1)\abs{ t}\right]^{M+1}}{1-\left[2e^2\mathfrak{d}(\mathfrak{d}+1)\abs{ t}\right]}.
    \end{split}
\end{eqnarray}
The third line comes from Lemma~\ref{lemma:error1}, and the fourth line is valid since $\abs{\mathcal{G}_m}\leq \abs{S}(e\mathfrak{d})^m$, with $\abs{S}$ represents the number of local terms in the Hamiltonian $H$. Let 
\begin{align}
    \epsilon^{\prime}=\frac{\abs{S}\left[2e^2\mathfrak{d}(\mathfrak{d}+1)\abs{ t}\right]^{M+1}}{1-\left[2e^2\mathfrak{d}(\mathfrak{d}+1)\abs{ t}\right]},
\end{align}
and this directly yields 
\begin{align}
    M\leq\frac{\ln\left(\frac{\abs{S}}{\epsilon^{\prime}[1-\left[2e^2\mathfrak{d}(\mathfrak{d}+1)\abs{ t}\right]]}\right)}{\ln(1/\left[2e^2\mathfrak{d}(\mathfrak{d}+1)\abs{ t}\right])}.
\end{align}
As a result, we can evaluate the running time complexity for computing the truncated $M$-order Taylor series. It is shown that the connected cluster set satisfies $\abs{\mathcal{G}_m}\leq \abs{S}(e\mathfrak{d})^m$. Recall $\mathcal{P}_c(\bm W)$ represents all connected filters given by the cluster $\bm W=(\bm W_1,\cdots,\bm W_m)$, and enumerating all filters of $\bm W$ into connected subclusters takes time $\exp(\mathcal{O}(M))$ (Details refer to Proposition~11 in Ref.~\cite{wild2023classical}). Finally, the coefficient $C(P)$ can be computed in time $\exp(\mathcal{O}(\abs{P}))$ by using the algorithm given by Ref.~\cite{bjorklund2008fast}. Taking all together, we can summarize that there exists a $$\abs{S}\exp(\mathcal{O}(M))=\abs{S}{\rm poly}\left[\left(\frac{\abs{S}}{\epsilon^{\prime}[1-\left[2e^2\mathfrak{d}(\mathfrak{d}+1)\abs{ t}\right]}\right)^{-1/\ln\left(\left[2e^2\mathfrak{d}(\mathfrak{d}+1)\abs{ t}\right]\right)}\right]$$ running time classical algorithm that can output an approximation to $\ln\left(d_{x,y, t}(H)\right)$ within $\epsilon$ additive error.

Equivalently, we can output an approximation $\hat{d}_{x,y, t}(H)$ such that 
\begin{align}
    e^{-\epsilon^{\prime}}\abs{\hat{d}_{x,y, t}(H)}\leq\abs{d_{x,y, t}(H)}\leq e^{\epsilon^{\prime}}\abs{\hat{d}_{x,y, t}(H)}.
\end{align}
This implies 
\begin{align}
    \frac{\abs{\hat{d}_{x,y, t}(H)-d_{x,y, t}(H)}}{\abs{d_{x,y, t}(H)}}=\abs{\frac{\hat{d}_{x,y, t}(H)}{d_{x,y, t}(H)}-1}\approx\abs{\ln\left(\frac{\hat{d}_{x,y, t}(H)}{d_{x,y, t}(H)}\right)}=\abs{\ln(\hat{d}_{x,y, t}(H))-\ln(d_{x,y, t}(H))}\leq\epsilon^{\prime},
\end{align}
where the second approximation is valid since $e^{-\epsilon^{\prime}}\leq \abs{\frac{\hat{d}_{x,y, t}(H)}{d_{x,y, t}(H)}}\leq e^{\epsilon^{\prime}}$. As a result, we can utilize estimators $\hat{d}_{x,y, t}(H)$ to construct an estimator to $D_{ t}(H)$, such that
\begin{eqnarray}
    \begin{split}
        \abs{D_{ t}(H)-\hat{D}_{ t}(H)}&=\abs{\sum\limits_{x,y}a_xa^*_y\left(d_{x,y, t}(H)-\hat{d}_{x,y, t}(H)\right)}\\
        &\leq\sum\limits_{x,y}\abs{a_xa_y^*}\abs{d_{x,y, t}(H)-\hat{d}_{x,y, t}(H)}\\
        &\leq \epsilon^{\prime}\sum\limits_{x,y}\abs{a_xa_y^*}\abs{d_{x,y, t}(H)}\\
        &\leq \epsilon^{\prime}\max_{x,y}\abs{d_{x,y, t}(H)}\\
        &\leq \epsilon^{\prime}.
    \end{split}
\end{eqnarray}
We note that above result depends on the condition $T\leq 1/(2e^2\mathfrak{d}(\mathfrak{d}+1))$, equivalently, we have 
\begin{align}
    \epsilon\geq e^{-\Delta^2/(4e^2\mathfrak{d}^4)}\Delta \gamma^{-1},
\end{align}

Let $\epsilon^{\prime}=\epsilon_1\Delta^2$, above results can be summarized as follows:

\begin{theorem}\label{thm:deq_rqite}
Suppose an $R$-configurational semi-classical guiding state $\ket{\psi_c}$ is given, there exists a classical algorithm to solve the GSEE problem with a runtime of
\begin{align}\label{eq:GSEE_limited}
        R^2\abs{S}{\rm poly}\left[\left(\frac{\abs{S}}{\epsilon \gamma[\Delta-\left[2e^2\mathfrak{d}(\mathfrak{d}+1)\right]}\right)^{1/\ln\left(\left[\Delta(2e^2\mathfrak{d}(\mathfrak{d}+1))^{-1}\right]\right)}\right]
    \end{align}
    that outputs a $\epsilon$-approximation to the ground state energy $E_0$, with corresponding to the accuracy limit:
\begin{eqnarray}
    \epsilon\geq e^{-\Delta^2/(4e^2\mathfrak{d}^4)}\Delta \gamma^{-1},
\end{eqnarray}
\end{theorem}

\begin{proof}
    To estimate each $\ln d_{x,y,t}(H)$ within $\epsilon^{\prime}$ additive error, the computational complexity of $|S|\exp(\mathcal{O}(M))$ is required. We note that $\epsilon^{\prime}=\epsilon_1/\|\hat{F}_{\sigma, T}\|_1=\mathcal{O}(\epsilon_1\Delta^2)$, meanwhile $\epsilon_1=\mathcal{O}(\epsilon \gamma\Delta^{-3})$, we finally obtain $\epsilon^{\prime}=\mathcal{O}(\epsilon \gamma\Delta^{-1})$ which gives rise to the above computational complexity. 
\end{proof}

\section{Proof of Theorem~\ref{theorem:dequantize}}
\begin{theorem}[Formal version of Theorem~\ref{theorem:dequantize}]
Suppose an $R={\rm poly}(n)$-configurational classical guiding state $\ket{\psi_c}$ is given, which has $\gamma$ overlap to the ground sate of $H$. Then there exists a classical algorithm to solve the GLH problem with a runtime of 
\begin{align}
    \mathcal{O}\left({\rm poly}(n)\left(\frac{e^{ \mathfrak{d}\sqrt{\log(1/(\epsilon\gamma))}/\Delta}{\rm poly}(n)R^{5/2}\Delta}{\epsilon\gamma}\right)^{e^{\mathfrak{d}\sqrt{\log(1/(\epsilon\gamma))}/\Delta}}\right)
\end{align}
that outputs a $\epsilon$-approximation to the ground state energy $E_0$, with success probability $\geq 1-{\rm poly}(R,1/\Delta)/(2^n\epsilon^2)$
\label{theorem:dequantizeformal}
\end{theorem}
Before giving the proof details, we still define the error notations:
\begin{itemize}
    \item $\epsilon$: the accuracy to estimate $E_0$;
    \item $\epsilon_1$: the accuracy to estimate $C(x)$
    \item $\epsilon^{\prime}$: the accuracy to estimate $\langle\psi_c|e^{-iHt}|\psi_c\rangle$
\end{itemize}

Lemmas~\ref{lemma:accuracy1},~\ref{lemma:time},~\ref{lemma:samplecomplexity} provide relationships between above accuracy parameters.

\begin{proof}
Here, we consider to design a classical algorithm in simulating the Hadamard Test algorithm when the target problem has the particle number preserved property. Specifically, we focus on a class of Hamiltonians (quantum lattice model and electronic structure model) $H$ with the property $[H,\hat{n}]=0$, where $\hat{n}=\sum_in_i$ and $n_i$ represents the particle number operator on the $i$-th site. Suppose the quantum system has $n$ spin orbitals, and the initial state is given by a semi-classical state $|\psi_c\rangle$ (used in Refs.~\cite{gharibian2022dequantizing, cade2022improved}) with the particle number $\hat{n}\geq 0$. Our target is to estimate both real part and imaginary part of $\langle \psi_c|e^{-iHt}|\psi_c\rangle$. In general, the quantum Hadamard Test algorithm requires the controlled $e^{-iHt}$ operation, however, the particle number preserving property enables us to bypass the requirement of ancilla qubit.

In detail, let the vacuum state $|\Omega\rangle=|0^n\rangle$, then the particle number symmetry enables the relationship
\begin{align}
    e^{-iHt}|\Omega\rangle=|\Omega\rangle.
\end{align}
Starting from the quantum state 
\begin{align}
    |\psi_1\rangle=\frac{1}{\sqrt{2}}\left(|\Omega\rangle+|\psi_c\rangle\right),
\end{align}
apply the operator $e^{-iHt}$ to the quantum system, then the quantum system becomes to
\begin{align}
    e^{-iHt}|\psi_1\rangle=\frac{1}{\sqrt{2}}\left(|\Omega\rangle+e^{-iHt}|\psi_c\rangle\right).
\end{align}
As a result, we have
\begin{align}
    {\rm Re}\left[\langle\psi_c|e^{-iHt}|\psi_c\rangle\right]=\langle\psi_1|e^{iHt}\left(|\Omega\rangle\langle\psi_c|+|\psi_c\rangle\langle\Omega|\right)e^{-iHt}|\psi_1\rangle=:\langle\psi_1|e^{iHt}(O_1+O_2)e^{-iHt}|\psi_1\rangle,
\end{align}
\begin{align}
    {\rm Im}\left[\langle\psi_c|e^{-iHt}|\psi_c\rangle\right]=i\langle\psi_1|e^{iHt}\left(|\psi_c\rangle\langle\Omega|-|\Omega\rangle\langle\psi_c|\right)e^{-iHt}|\psi_1\rangle=:i\langle\psi_1|e^{iHt}(O_2-O_1)e^{-iHt}|\psi_1\rangle.
\end{align}

As a result, simulating $\langle\psi_1|e^{iHt}O_1e^{-iHt}|\psi_1\rangle$ and $\langle\psi_1|e^{iHt}O_2e^{-iHt}|\psi_1\rangle$ suffices to simulate the the Loschmidt echo $\langle\psi_c|e^{-ijH}|\psi_c\rangle$. Taking $\langle\psi_1|e^{iHt}O_2e^{-iHt}|\psi_1\rangle$ as an example, we suppose the classical state $|\psi_c\rangle=\sum_{j=1}^Ra_{j}|\bm j\rangle$ has $R={\rm poly}(n)$ configurations where each configuration (product state) $|\bm j\rangle$ has $P$ particles, and amplitude $\abs{a_{ j}}=1/{\rm poly}(n)$ with $\sum_{j=1}^R\abs{a_{j}}^2=1$. This implies $\max_{j\in[R]}\{\abs{a_j}\}\geq 1/\sqrt{R}$. The expectation value
\begin{align}
    \langle\psi_1|e^{iHt}O_2e^{-iHt}|\psi_1\rangle=\sum_{j=1}^Ra_{j}\langle \psi_1|e^{-Ht}|\bm j\rangle\langle\Omega|e^{-iHt}|\psi_1\rangle,
\end{align}
and an $\mathcal{O}(\epsilon^{\prime}/\sqrt{R})$-approximation to each term $\langle \psi_1|e^{-Ht}|\bm j\rangle\langle\Omega|e^{-iHt}|\psi_1\rangle$ may give rise to a $\epsilon^{\prime}$-approximation to $\langle\psi_1|e^{iHt}O_2e^{-iHt}|\psi_1\rangle$. To approximate $\langle \psi_1|e^{-Ht}|\bm j\rangle\langle\Omega|e^{-iHt}|\psi_1\rangle$, we assume $|\psi_1\rangle=\sum_{j=1}^{R+1}b_j|j\rangle$ with $b_0=1/\sqrt{2}$ and $b_j=a_j/\sqrt{2}$ for $j\geq 2$. Then the expectation value $\langle \psi_1|e^{-Ht}|\bm j\rangle\langle\Omega|e^{-iHt}|\psi_1\rangle=\sum_{x,x^{\prime}=1}^{R+1}b_xb^*_{x^{\prime}}\langle \bm x|e^{-Ht}|\bm j\rangle\langle\Omega|e^{-iHt}|\bm x^{\prime}\rangle$, where each configuration $|\bm x\rangle$ can be considered as a sample from the Clifford ensemble. As a result, one requires to estimate $\langle \bm x|e^{-Ht}|\bm j\rangle\langle\Omega|e^{-iHt}|\bm x^{\prime}\rangle$ within $\epsilon^{\prime}/R^{5/2}$ additive error, sufficing to give a $\epsilon^{\prime}/\sqrt{R}$ additive error to $\langle \psi_1|e^{-Ht}|\bm j\rangle\langle\Omega|e^{-iHt}|\psi_1\rangle$.

Considering the evolution time $t\leq \tilde{\mathcal{O}}\left(\Delta^{-1}\sqrt{2\ln(8\pi^{-1}\epsilon_1^{-1}\Delta^{-2})}\right)$ and $\epsilon^{\prime}=\epsilon\gamma\Delta^{-1}$, by theorem~\ref{theorem_formal}, this takes the computational complexity $$\mathcal{O}\left({\rm poly}(n)\left(\frac{e^{ \mathfrak{d}\sqrt{2\ln(8\pi^{-1}(\epsilon\gamma)^{-1}\Delta^{-2})}/\Delta}{\rm poly}(n)R^{5/2}\Delta}{\epsilon\gamma}\right)^{e^{\mathfrak{d}\sqrt{2\ln(8\pi^{-1}(\epsilon\gamma)^{-1}\Delta^{-2})}/\Delta}}\right)$$ with success probability $\geq 1-{\rm poly}(R)/(2^n(\epsilon^{\prime})^2)$.

Repeat above process for each $\langle \psi_1|e^{-Ht}|\bm j\rangle\langle\Omega|e^{-iHt}|\psi_1\rangle$, one obtains an estimation to $\langle\psi_1|e^{iHt}O_2e^{-iHt}|\psi_1\rangle$ and $\langle\psi_1|e^{iHt}O_1e^{-iHt}|\psi_1\rangle$. To estimate the function $C(x)$ within $\epsilon_1=\mathcal{O}(\epsilon \gamma\Delta^{-3})$ additive error, it is required $\mathcal{O}(\Delta^{-2}\log(\Delta/\epsilon)\epsilon_1^{-2})$  samples. For each sample, we require the corresponding estimation all successful, the final success probability is $\geq 1-{\rm poly}(R,1/\Delta)/(2^n\epsilon^2)$.
\end{proof}

\section{2D GLH problem by Ancilla-Free Hadamard Test}
We suppose the classical state $|\psi_c\rangle=\sum_{j=1}^Ra_{j}|j\rangle$ has $R={\rm poly}(n)$ configurations where each configuration (product state) $|j\rangle$ has $P$ particles, and amplitude $\abs{a_{ j}}=1/{\rm poly}(n)$ with $\sum_{j=1}^R\abs{a_{j}}^2=1$. This implies $\max_{j\in[R]}\{\abs{a_j}\}\geq 1/\sqrt{R}$. The expectation value
\begin{align}
    \langle\psi_1|e^{iHt}O_2e^{-iHt}|\psi_1\rangle=\sum_{j=1}^Ra_{j}\langle \psi_1|e^{iHt}|\bm j\rangle\langle\Omega|e^{-iHt}|\psi_1\rangle,
\end{align}
and an $\mathcal{O}(\epsilon^{\prime}/\sqrt{R})$-approximation to each term $\langle \psi_1|e^{-Ht}|\bm j\rangle\langle\Omega|e^{-iHt}|\psi_1\rangle$ may give rise to a $\epsilon^{\prime}$-approximation to the target quantum mean value.

Denote $O=|\bm j\rangle\langle\Omega|$, the quantum mean value can be equivalently computed by
\begin{eqnarray}
    \begin{split}
         \mu(t)&=\langle\psi_1|e^{iHt}\left(O_1\otimes\dots\otimes O_n\right)e^{-iHt}|\psi_1\rangle\\
         &=
         \langle\psi_1|\left(e^{iHt}O_1e^{-iHt}\right)\left(e^{iHt}O_2e^{-iHt}\right)\cdots\left(e^{iHt}O_ne^{-iHt}\right)|\psi_1\rangle\\    
    &=\langle\psi_1|U_1(t)U_2(t)\cdots U_n(t)|\psi_1\rangle,
    \end{split}
\end{eqnarray}
where $U_i(t)=e^{iHt}O_ie^{-iHt}$.  
Our first step aims to approximate $U_i(t)$ by $V_i(t)$ such that $\|U_i(t)-V_i(t)\|\leq \mathcal{O}(\epsilon/2n)$ via using the cluster expansion method which is given by lemma~\ref{lemma:VTCompute}. Here, the operator $V_i(t)$ is essentially a linear combination of ${\rm poly}(n)$ matrices which nontrivial act on at most $\mathcal{O}(e^{\mathfrak{d}t}\log(2n/\epsilon))$ qubits. After approximating $U_i(t)$ by operator $V_i(t)$ for index $i\in[n]$, the mean value $\mu(\vec{t})$ can be approximated by $ \hat{\mu}(t)=\langle\psi_1|V_1(t)\cdots V_n(t)|\psi_1\rangle$ such that $\abs{\mu(t)-\hat{\mu}(t)}\leq \epsilon/2$.

The second step applies the causality principle and the lightcone of $V_i(t)$ to assign $\{V_i(t)\}_{i=1}^n$ into two different groups, which are denoted by $V(R_1)$ and $V(R_2)$. This method is first studied in Ref.~\cite{bravyi2021classical} to simulate constant 2D digital quantum circuits. It is shown that each region ($R_1$ or $R_2$) consists of $\sqrt{n}/4M$ sub-regions which are separated by $\geq 2M$ distance. This property enables operators $V(R_1)$ and $V(R_2)$ are easy to simulate classically, and the quantum dynamics mean value has the form $\hat{\mu}(t)=\langle\psi_1|V(R_1)V(R_2)|\psi_1\rangle$. Then the classical Monte Carlo algorithm can be used to approximate $\hat{\mu}(t)$. Noting that operators $V(R_1)$ and $V(R_2)$ are not always unitary matrices, they have to be normalized in advance, such that $\gamma_i=\|V(R_i)|0^n\rangle\|^2\leq1$ for $i\in\{1,2\}$. This step can be implemented efficiently since both $V(R_1)$ and $V(R_2)$ are the product of some local operators $V_i(\vec{t})$ which can be normalized easily. As a result, as a mean value of $$F(x)=\frac{\gamma_1\langle x|V(R_2)|0^n\rangle}{\langle x|V^{\dagger}(R_1)|0^n\rangle}$$ with $x$ samples from $$p(x)=\gamma_1^{-1}\abs{\langle0^n|V(R_1)|x\rangle}^2,$$ we have
\begin{eqnarray}
         \hat{\mu}(t)=\sum\limits_{x}\langle0^n|V(R_1)|x\rangle\langle x|V(R_2)|0^n\rangle
         =\sum\limits_{x}p(x)\frac{\gamma_1\langle x|V(R_2)|0^n\rangle}{\langle x|V^{\dagger}(R_1)|0^n\rangle},
         \label{Eq:MCMC}
\end{eqnarray}
and the variance of $F(x)$ is given by ${\rm Var}(F)=\sum_{x}p(x)\left\|\frac{\gamma_1\langle x|V(R_2)|0^n\rangle}{\langle x|V^{\dagger}(R_1)|0^n\rangle}\right\|^2-\hat{\mu}^2(t)=\gamma_1\gamma_2-\hat{\mu}^2(t)\leq 1$.

As a result, $\mathcal{O}(4/\epsilon^2)$ samples $x$ suffice to provide an estimation to $\hat{\mu}(\vec{t})$ within $\mathcal{O}(\epsilon/2)$ additive error. Combining the above two steps together, a $\epsilon$ approximation to the quantum mean value problem is provided. 

\begin{algorithm}
\label{Alg}
\caption{Classical Algorithm for estimating $\langle \psi_1|e^{-Ht}|\bm j\rangle\langle\Omega|e^{-iHt}|\psi_1\rangle$ in 2D topology}
\textbf{Input:} Hamiltonian set $\{H^{(1)},\cdots,H^{(K)}\}$, time series $\{t_1,\cdots,t_K\}$, global observable $O$, accuracy $\epsilon$;\\
\textbf{Output:} Mean value estimation $\hat{\mu}(t)$;\\
{\textbf{for}} $i=1,\cdots, n$\\
\quad \quad \textbf{Compute} $V_i(\vec{t})$ given by Eq.~\ref{Eq:clusterexpansion} via using Lemma~\ref{lemma:VTCompute}.\\

    {\textbf{End for}}\\   
{\textbf{Grouping}} $\{V_i(\vec{t})\}$ into $V(R_1)$ and $V(R_2)$;\\
{\textbf{for}} $j=1,\cdots, J=[4/\epsilon^2]$\\
\quad \quad \textbf{Sample} $x_j\sim p(x)=\abs{\langle0^n|V(R_1)|x\rangle}^2$, {\textbf{Compute}} $F(x_j)$ by using Light-cone arguement;\\
    {\textbf{End for}}\\
\textbf{Output} $\hat{\mu}(\vec{t})=\frac{1}{J}\sum_{j=1}^{J}F(x_j)$.
\end{algorithm}

We summarize the above results as follows.

\begin{theorem}
Suppose a 2D geometrically local Hamiltonian, and
an $R$-configurational semi-classical guiding state $\ket{\psi_c}$, there exists a classical algorithm to solve the GLH problem with a runtime of
 \begin{align}
         \mathcal{O}\left(\frac{nR}{(\epsilon\gamma)^2}\left(\frac{\sqrt{R}ne^{\mathfrak{d}\sqrt{2\ln(8\pi^{-1}(\epsilon\gamma)^{-1}\Delta^{-2})}/\Delta}}{\epsilon\gamma}\right)^{e^{\mathfrak{d}\sqrt{2\ln(8\pi^{-1}(\epsilon\gamma)^{-1}\Delta^{-2})}/\Delta}\log(\sqrt{R}n\Delta/(\epsilon\gamma))}\right),
    \end{align}
    that outputs a $\epsilon$-approximation to the ground state energy $E_0$.
\end{theorem}

\subsection{Compute $F(x)$}
\begin{definition}[$2$-coloring of $2$-dimensional lattice with distance $r$] Consider a graph representing a $2$-dimensional lattice, where each vertex is assigned a color, and the entire lattice is divided into many small regions with different colors. A $2$-coloring of $2$-dimensional lattice with distance $r$ satisfies the following properties:
\begin{itemize}
\item There are $2$ colors in total;
\item The distance between two regions with the same color is at least $r$.
\end{itemize}
\label{Def:color}
\end{definition}

Suppose the $n$-qubit $2$-dimensional lattice has been colored by $2$ regions, denoted by $R_1,R_2$, then each region can be divided into 
\begin{align}
    R_j=\cup_{l=1}^{\sqrt{n}/r}R_j(l),
\end{align}
where each zone $R_j(l)$ contains $r\sqrt{n}$ qubits. It is easy to check $R_j$ contains $S=[\sqrt{n}/r]$ zones separated by distance at least $\geq r$.

Let the zone distance $r=2M$ (given in Def.~\ref{Def:color} and lemma~\ref{lemma:clustersupp}), then for any region index $j\in[2]$ and $l\in[S]$, we have 
\begin{align}
    \abs{R_j(l)}=2M\sqrt{n}.
\end{align}

Recall that $\hat{\mu}(t)=\langle0^n|V_1(t)V_2(t)\cdots V_n(t)|0^n\rangle$, then we can assign operators $\{V_i(t)\}_{i=1}^n$ into groups $R_1,R_2$ marked by different colors, where each group is denoted by
\begin{align}
    V(R_j)=\bigotimes\limits_{i \in R_j}V_i(t)=\bigotimes\limits_{l=1}^S\left(\bigotimes_{i\in R_j(l)}V_i(t)\right)=\bigotimes\limits_{l=1}^SV_{R_j(l)}(t),
    \label{Eq:clustercircuit}
\end{align}
where $V_{R_j(l)}(t)$ contains all $V_i(t)$ if the qubit index $i$ belongs to the zone $R_j(l)$. 
To deeply understand whether $V(R_j)$ can be classically simulated, we need to evaluate the support of $V_{R_j(l)}(t)$.

\begin{lemma}
    Given the operator $V_{R_j(l)}(t)$ defined as Eq.~\ref{Eq:clustercircuit}, we have
\begin{align}
    {\rm supp}\left(V_{R_j(l)}(t)\right)\leq 4M\sqrt{n},
\end{align}
where $M=\mathcal{O}\left(e^{\pi teK\mathfrak{d}}\log(2n/\epsilon)\right)$,
meanwhile
\begin{align}
  {\rm supp}\left(V_{R_j(l)}(t)\right)\cap {\rm supp}\left(V_{R_j(q)}(t)\right)=\emptyset
\end{align}
for all indexes $l\neq q\in[S]$.
\label{lemma:Lightconesize}
\end{lemma}

\begin{proof}
In the two-dimensional lattice, the proposed coloring method enables the region $R_j(l)$ being a ($\sqrt{n}\times 2M$) rectangle area, where each qubit relates to an operator $V_i(t)$. According to the lemma~\ref{lemma:clustersupp}, it is shown that the support of $V_i(t)$ is upper bounded by $\mathcal{O}(M^2)$. As a result, most of ${\rm supp}(V_i(t))$ have the overlap and can thus be contracted. Finally, $ {\rm supp}\left(V_{R_j(l)}(t)\right)$ is only characterized by operators $\{V_i(t),i\in \partial R_j(l)\}$ combined with $R_j(l)$, that is 
    \begin{align}
         {\rm supp}\left(V_{R_j(l)}(t)\right)=\bigcup\limits_{i\in \partial R_j(l)}{\rm supp}\left(V_i(t)\right)+\abs{R_j(l)}\leq 4M\sqrt{n},
    \end{align}
where $\partial R_j(l)$ is the boundary of the region $R_j(l)$. The visualization of the above statement is provided by Fig.~1 in the main file.

Recall that ${\rm supp}\left(V_{R_j(l)}(t)\right)$ is essentially a quasi-1D-rectangular in the two-dimensional lattice, then the short-side length of ${\rm supp}\left(V_{R_j(l)}(t)\right)$ is upper bounded by $\leq 4M$. Furthermore, from Def.~\ref{Def:color}, we know that
\begin{align}
    {\rm dist}\left(R_j(l),R_{j}(q)\right)\geq r=2M,
\end{align}
which naturally implies 
$${\rm supp}\left(V_{R_j(l)}(t)\right)\cap {\rm supp}\left(V_{R_j(q)}(t)\right)=\emptyset.$$
for all $l\neq q$.
\end{proof}

The above property enables us to decouple $V(R_j)$ into a series of operators $V_{R_j(l)}(t)$ whose support region does not have the overlap, and this naturally provides a classical method in simulating $\langle x|V_{R_j(l)}(t)|0^n\rangle$. 

\begin{lemma}
    Given the operator $V_{R_j(l)}(t)$ defined as Eq.~\ref{Eq:clustercircuit}, for any for $x\in\{0,1\}^n$, there exists a classical algorithm that can deterministically output $\langle x|V_{R_j(l)}(t)|0^n\rangle$ within 
     $C(n)\leq\mathcal{O}\left(\sqrt{n}2^{4M^2}\right)$ running time.
    \label{lemma:samplecompute}
\end{lemma}

\begin{proof}
  We first label all qubits contained in the region $V_{R_j(l)}$ by $(q_0,\cdots,q_{4M\sqrt{n}-1})$ using the row major order. Then we consider a $M\times 4M$ window $W$ swiping along the row index. According to the result given in Lemma~\ref{lemma:clustersupp}, it is shown that 
   $\abs{{\rm supp}\left(V_{i}(t)\right)}\leq 4M^2$. At the initial stage, suppose the window $W$ only covers qubit set $\mathcal{W}=\{q_0,\cdots,q_{4M^2-1}\}$, then the support size of $V_{i}(t)$ implies if $q_i\notin\mathcal{W}$, then $V_{q_i}(t)$ may not affect the measurement result of $l_0=\{q_0,\cdots q_{4M-1}\}$ which represents the first row within the window $W$. This further implies that computing $\langle x_{l_0}|V_{l_0}(t)|0_{l_0}\rangle$ can be fixed into a small subspace, and the state vector simulator has runtime approximately $\mathcal{O}(2^{4M^2})$ per each gate or one-qubit measurement~\cite{bravyi2021classical}. Swiping the window $W$ from $l_0$ to $l_{\sqrt{n}-1}$, $\langle x|V_{R_j(l)}(t)|0^n\rangle$ can be deterministically computed by a classical algorithm with $\tilde{\mathcal{O}}\left(\sqrt{n}2^{4M^2}\right)$ running time. 
\end{proof}

\section{Limitations of NISQ Algorithms on current quantum devices}
Here, we study the Hamiltonian simulation algorithm on near-term quantum devices in the context of a noisy environment. We suppose each quantum gate is affected by a local Pauli channel. For the sake of clarity, we begin by presenting the definitions of the local Pauli channel.

\begin{definition}[Local Pauli channel]
\label{Def:localchannel}
Let $\mathcal{N}_i$ denote a local Pauli channel and the action of $\mathcal{N}_i$ is random local Pauli operators $P$ acting on the $i$-th qubit according to a specific channel parameter $\{q(P)\}$, where $P\in\{\mathbb{I}, \sigma^x,\sigma^y,\sigma^z\}$. Specifically, the action of $\mathcal{N}_i$ is given by 
\begin{equation}\label{eq:Local Pauli Noise Channel}
    \mathcal{N}_i(P)=q(P)P
\end{equation}
for $P\in\{\mathbb{I},\sigma^x,\sigma^y, \sigma^z\}$, where $q(P)\in(-1,1)$. The noise strength in this model is represented by a single parameter $q=\max_{P\in\{\sigma^x,\sigma^y, \sigma^z\}}\abs{q(P)}$.
\end{definition}

\begin{definition}[Quantum Circuit affected by Pauli Channel]
\label{Eq:densitymatrix}
We assume that the noise in the quantum device is modeled by a Pauli channel $\mathcal{N}_i$ with strength $q$. Let $\mathcal{U}$ be a causal slice, and let $\mathcal{N} \circ \Ucal=\left(\otimes_{i=1}^n\mathcal{N}_i\right)\circ \Ucal$ be the representation of a noisy circuit layer. We define the $d$-depth noisy quantum state with noise strength $q$ as
\begin{align}
 \rho_{q,d}=\mathcal{N}\circ \mathcal{U}_{d}\circ \mathcal{N}\circ \mathcal{U}_{d-1}\circ\cdots\circ\mathcal{N}\circ \mathcal{U}_{1}(|0^n\rangle\langle0^n|).
 \label{Eq:density_matrix}
\end{align} 
\end{definition}

Quantum error mitigation is necessary due to imperfections in quantum devices to correct the bias caused by noise. The fundamental concept is to correct the impact of quantum noise through classical post-processing of measurement results, without mid-circuit measurements and adaptive gates as in standard error correction. Here, we argue that the existing error mitigation strategies might require a number of samples $\rho_{q,d}$ that scales exponentially with the number of gates in the light-cone of the observable of interest. This thus losses the original quantum advantages compared to the proposed classical simulation algorithm which only requires quasi-polynomial time. We extend previous results given by Ref~\cite{quek2022exponentially} to a more general Pauli channel. We first review some related lemmas, then give the generalized result and main result (given in the main file) on the quantum error-mitigation overheads.



\subsection{Involved Lemmas}
We require following lemmas to support our proof.
\begin{lemma}[Ref.~\cite{quek2022exponentially}]
    Let $\gamma,\cdots,P_N$ be probability measures on some state space $X$ such that 
    \begin{align}
        \frac{1}{N+1}\sum\limits_{k=0}^ND(P_k\|P_0)\leq\alpha\log(N)
    \end{align}
    for $0<\alpha<1$. Then the minimum average probability of error over tests $\psi:X\mapsto \{0,1,\cdots, N\}$ that distinguish the probability distributions $P_0,\cdots, P_N$ which we define as
    \begin{align}
        \bar{p}_{e,N}={\rm inf}_{\psi}\frac{1}{N+1}\sum\limits_{j=0}^NP_j(\psi\neq j)
    \end{align}
    satisfies
    \begin{align}
         \bar{p}_{e,N}\geq \frac{\log(N+1)-\log(2)}{\log(N)}-\alpha.
    \end{align}
    \label{lemma:Fano}
\end{lemma}

\begin{lemma}[Lemma~6 in~\cite{wang2021noise}]
    Consider a single instanoise channel $\mathcal{N}=\mathcal{N}_1\otimes\cdots\otimes\mathcal{N}_n$ where each local noise channel $\{\mathcal{N}_j\}_{j=1}^n$ is a Pauli noise channel that satisfies $\mathcal{N}_j(\sigma)=q_{\sigma}\sigma$ for $\sigma\in\{X,Y,Z\}$ and $q_{\sigma}$ be the Pauli strength. Then we have
    \begin{align}
        D_2\left(\mathcal{N}(\rho)\|\frac{I^{\otimes n}}{2^n}\right)\leq q^{2c}D_2\left(\rho\|\frac{I^{\otimes n}}{2^n}\right),
    \end{align}
    where $D_2(\cdot\|\cdot)$ represents the $2$-Renyi relative entropy, $q=\max_{\sigma}q_{\sigma}$ and $c=1/(2\ln 2)$.
    \label{lemma:renyiineq}
\end{lemma}

\begin{lemma}
    Given an arbitrary $n$-qubit density matrix and maximally mixed state $I^{\otimes n}/2^n$, we have
    \begin{align}
        D\left(\rho\|I^{\otimes n}/2^n\right)\leq D_2\left(\rho\|I^{\otimes n}/2^n\right),
    \end{align}
    where $D(\cdot\|\cdot)$ denotes the relative entropy and $D_2(\cdot\|\cdot)$ denotes the $2$-Renyi relative entropy.
    \label{lemma:renyi}
\end{lemma}
\begin{proof}
    Given quantum states $\rho$ and $\sigma$, the quantum $2$-Renyi entropy 
    \begin{align}
        D_2(\rho\|\sigma)=\log{\rm Tr}\left[\left(\sigma^{-1/4}\rho\sigma^{-1/4}\right)^2\right].
    \end{align}
    When $\sigma=I^{\otimes n}/2^n$, we have $D_2(\rho\|I^{\otimes n}/2^n)=\log{\rm Tr}\left[\left((I^{\otimes n}/2^n)^{-1}\rho^2\right)\right]=n+\log{\rm Tr}[\rho^2]$. Noting that the function $y=x^2-x\log x\geq 0$ when $x\in[0,1]$, and this implies ${\rm Tr}(\rho^2)\geq {\rm Tr}(\rho\log \rho)$. Finally, we have
    \begin{align}
        D\left(\rho\|I^{\otimes n}/2^n\right)=n+{\rm Tr}\left[\rho\log\rho\right]+n\leq {\rm Tr}\left[\rho^2\right]+n=D_2\left(\rho\|I^{\otimes n}/2^n\right).
    \end{align}
\end{proof}

\begin{lemma}[Ref.~\cite{haah2016sample}]
    Let $\epsilon\in(0,1)$ and $\delta\in(0,1)$. Suppose there exists a POVM $\{M_{\sigma}d\sigma\}$ on $(\mathbb{C}^{2^n})^{\otimes m}$ such that for any quantum state $\rho$,
    \begin{align}
        \int_{d_{\rm tr}(\sigma,\rho)\leq \epsilon}d\sigma{\rm Tr}\left[M_{\sigma}\rho^{\otimes m}\right]\geq 1-\delta,
        \label{Eq:measurement}
    \end{align}
    This implies the sample complexity lower bound
    \begin{align}
        m\geq \Omega\left(\frac{2^{3n}(1-\epsilon)^2}{\epsilon^2}\right).
    \end{align}
    \label{lemma:samplelowerbound}
\end{lemma}


\subsection{Generalize the Theorem~1 in Ref~\cite{quek2022exponentially} to Pauli channel}
\begin{fact}[Generalized result to Ref.~\cite{quek2022exponentially}]
    Let $\mathcal{A}$ be an error mitigation algorithm that takes as input $m$ noisy quantum state copies prepared by a $d$-depth noisy quantum circuit that affected by local Pauli noise channels with strength $q$, and a set of Hermitian observables. The error mitigation algorithm $\mathcal{A}$ requires $m\geq\Omega(q^{-2d})$ copies of noisy states in the worst-case scenario over the choice of observable sets.
     \label{Theorem:limitation}
\end{fact}

The basic idea is to construct a polynomial reduction to the quantum state discrimination problem~\cite{quek2022exponentially}. Let us consider an error mitigation problem. Given the quantum state set $F_{\rho}=\{\rho_0,\rho_1,\cdots,\rho_N\}$, where $\rho_x=|x\rangle\langle x|$ when $x<N$ and $\rho_N=I_n/2^n$ with $N=2^n$, as the input of a noiseless quantum circuit $C$, and utilize a set of observables $\{CZ_iC^{\dagger}\}_{i=1}^n$ to measure the output states $C(\rho_x)$. The quantum error mitigation algorithm should output the estimation $o_j$ such that $\abs{o_j-{\rm Tr}(C(\rho_x)CZ_jC^{\dagger})}\leq\epsilon$. Now we show that a noisy state identification problem can be solved by quantum error mitigation algorithm. Consider an arbitrary $\rho_x\in F_{\rho}$, we may have two scenarios:
\begin{itemize}
    \item If the unknown quantum state $\rho_x$ whose index satisfies $x<N$, we have $y_j={\rm Tr}(C(\rho_x)CZ_jC^{\dagger})=1-2x_j$, where $x_j$ represents the $j$-th bit within $x$;
    \item Else $\rho_N=I_n/2^n$ resulting in $y_j={\rm Tr}(C(\rho_N)CZ_jC^{\dagger})=0$.
\end{itemize}

Randomly sample a quantum state $\rho_x\in F_{\rho}$, we denote $\hat{y}=(y_1,\cdots,y_n)$ and $P_x(\hat{y})$ represents the probability distribution on measuring the result $\hat{y}$. As a result, if a quantum error mitigation algorithm can successfully recover every $o_j$ for $j\in[n]$, this enables us to uniquely identify the unknown quantum state $\rho_x$ from the distribution $P_x(\hat{y})$. Then we can utilize Fano's lower bound for quantum state identification problem (Lemma~\ref{lemma:Fano}). Specifically, we have 
\begin{eqnarray}
\begin{split}
    \frac{1}{N+1}\sum\limits_{k=0}^ND(P_k\|P_0)\leq& \frac{1}{N+1}\sum\limits_{k=0}^ND\left(\Phi^{\otimes m}_{C,q}(\rho_k)\|(I_n/2^n)^{\otimes m}\right)\\
    \leq &\frac{1}{N+1}\sum\limits_{k=0}^ND_2\left(\Phi^{\otimes m}_{C,q}(\rho_k)\|(I_n/2^n)^{\otimes m}\right)\\
    \leq &\frac{1}{N+1}\sum\limits_{k=0}^Nmq^{2cd}D_2\left(\rho_k\|I_n/2^n\right)\\
    = & q^{2cd}mn\\
    = & q^{2cd}m\log N,
\end{split}
\end{eqnarray}
where $d$ represents the depth of quantum circuit $C$. Let $\alpha=q^{2cd}m$, then in order for the test to have a constant failure probability $\delta$, it takes at least $m\geq q^{-2cd}(1-\delta)$ copies.

\subsection{A sample complexity lower bound related to approximation error and circuit depth}
\label{App:proofofQEM}

\begin{problem}
    Consider a pure quantum state packing net $\{\rho_0,\cdots,\rho_{\abs{\Omega}}\}$ such that for $\frac{1}{2}\|\rho_i-\rho_j\|\geq 2\epsilon$ for any $i\neq j$, and a $d$-depth quantum circuit $\mathcal{C}$ affected by Pauli channel $\mathcal{N}$. Suppose that a distinguisher has knowledge of $\mathcal{C}$ and $\mathcal{N}$, and is given access to copies of the quantum state $\Phi_{\mathcal{C},q}(\rho_i)$, with some unknown index $i\in[\abs{\Omega}]$. What is the fewest number of copies of $\Phi_{\mathcal{C},q}(\rho_i)$ sufficing to identify $i\in[\abs{\Omega}]$ with high probability?
    \label{problem2}
\end{problem}

Now we discuss how to utilize the quantum error mitigation algorithm to solve the above problem. Suppose the noisy state $\Phi_{\mathcal{C},q}(\rho_i)$ is provided, we focus on its quantum mean value on observables 
$$\{C^{\dagger}\rho_0C,\cdots,C^{\dagger}\rho_{\abs{\Omega}}C\}$$
that is to estimate $\{{\rm Tr}\left(\Phi_{\mathcal{C},q}(\rho_i)C^{\dagger}(\rho_j)C\right)\}$ for $j\in[N]$. If a quantum error mitigation algorithm $\mathcal{A}$ can recover the quantum mean value, then we have the map
\begin{align}
   \left\{{\rm Tr}\left(\Phi_{\mathcal{C},q}(\rho_i)C^{\dagger}(\rho_j)C\right)\right\}\mapsto\{{\rm Tr}\left[\rho_i\rho_j\right]\}.
\end{align}
According to our assumption, all quantum states $\rho_i$ comes from a packing-net, then for any $i\neq j$, we have ${\rm Tr}(\rho_i\rho_j)=\sqrt{1-d^2_{tr}(\rho_i,\rho_j)}\leq\sqrt{1-4\epsilon^2}\leq 1-2\epsilon^2$. Otherwise we have ${\rm Tr}(\rho_i\rho_i)\geq 1-\epsilon^2$. As a result, a quantum error mitigation algorithm can be used to identify the index $i$ hidden in the noisy state $\Phi_{\mathcal{C},q}(\rho_i)$, which thus can solve Problem~\ref{problem2}. The sample complexity of Problem~\ref{problem2} can be used to benchmark the sample complexity lower bound of the quantum error mitigation problem.

\begin{theorem}
     Let $\mathcal{A}$ be an input state-agnostic error mitigation algorithm that takes as input $m$ copies noisy quantum states produced by a $d$-depth quantum circuit $\mathcal{C}$ affected by $q$-strength local Pauli noise channels, and a set of observables $\{O\}$. Suppose the algorithm $\mathcal{A}$ is able to produce estimates $\{\hat{o}\}$ such that $\abs{\hat{o}-\langle o\rangle}\leq\epsilon$. Then the sample complexity 
     \begin{align}
      m\geq\min\left\{\frac{q^{-2cd}(1-\eta)^2}{2n},\frac{2^{3n}(1-\epsilon)^2}{\epsilon^2}\right\}
  \end{align}
     in the worst-case scenario over the choice of the observable set, where $c=1/(2\ln 2)$ and $\eta\in\mathcal{O}(1)$.
\end{theorem}

\begin{proof}
    Randomly select $\rho_i$ and $\rho_j$ from the $\epsilon$-packing net, we consider the sample complexity $m$ in distinguishing quantum states $\Phi_{\mathcal{C},q}(\rho_i)$ and $\Phi_{\mathcal{C},q}(\rho_j)$. When their trace distance is quite large, let $\eta\in(0,1)$ and we have
    \begin{eqnarray}
    \begin{split}
        1-\eta&\leq\frac{1}{2}\left\|\Phi_{\mathcal{C},q}(\rho_i)^{\otimes m}-\Phi_{\mathcal{C},q}(\rho_j)^{\otimes m}\right\|_1\\
        &\leq \frac{1}{2}\left(\left\|\Phi_{\mathcal{C},q}(\rho_i)^{\otimes m}-(I_n/2^n)^{\otimes m}\right\|_1+\left\|\Phi_{\mathcal{C},q}(\rho_j)^{\otimes m}-(I_n/2^n)^{\otimes m}\right\|_1\right)\\
        &\leq \frac{1}{\sqrt{2}}\left(D^{1/2}\left(\Phi^{\otimes m}_{\mathcal{C},q}(\rho_i)\|(I_n/2^n)^{\otimes m}\right)+D^{1/2}\left(\Phi^{\otimes m}_{\mathcal{C},q}(\rho_i)\|(I_n/2^n)^{\otimes m}\right)\right),
    \end{split}
    \end{eqnarray}
    where the second line comes from the triangle inequality and the third line comes from the Pinsker's inequality. Using Lemmas~\ref{lemma:renyi} and~\ref{lemma:renyiineq}, we have 
    \begin{align}
        1-\eta\leq \frac{1}{\sqrt{2}}\left(D^{1/2}_2\left(\Phi^{\otimes m}_{\mathcal{C},q}(\rho_i)\|(I_n/2^n)^{\otimes m}\right)+D^{1/2}_2\left(\Phi^{\otimes m}_{\mathcal{C},q}(\rho_i)\|(I_n/2^n)^{\otimes m}\right)\right)\leq \sqrt{2nm}q^{cd},
    \end{align}
  where $d$ represents the quantum circuit depth of $\mathcal{C}$. As a result we have
  \begin{align}
      m\geq\frac{q^{-2cd}(1-\eta)^2}{2n}.
  \end{align}
  On other hand, when quantum states $\Phi_{\mathcal{C},q}(\rho_i)$ and $\Phi_{\mathcal{C},q}(\rho_j)$ are very close, that is $\frac{1}{2}\|\Phi_{\mathcal{C},q}(\rho_i)-\Phi_{\mathcal{C},q}(\rho_j)\|_1\leq \epsilon$ (this is possible since a CPTP map reduces the trace distance), Lemma~\ref{lemma:samplelowerbound} implies the sample complexity
  \begin{align}
      m\geq \frac{2^{3n}(1-\epsilon)^2}{\epsilon^2}.
  \end{align}
  Combine above inequalities together, we finally have
  \begin{align}
      m\geq\min\left\{\frac{q^{-2cd}(1-\eta)^2}{2n},\frac{2^{3n}(1-\epsilon)^2}{\epsilon^2}\right\}.
  \end{align}
\end{proof}

\section{Classical Simulation for $2$D Fermi-Hubbard model}
\label{App:VQE}

The Fermionic Hubbard model has served as a paradigmatic example for strongly correlated problems. Specifically, its Hamiltonian is given by
\begin{align}
    H_{FH}=-\tau\sum\limits_{(i,j)\in E,\sigma\in\{\uparrow,\downarrow\}}(a_{i\sigma}^{\dagger}a_{j\sigma}+a_{j\sigma}^{\dagger}a_{i\sigma})+U\sum\limits_{i\in V}n_{i\uparrow}n_{i\downarrow},
\end{align}
where $\tau, U$ are coupling parameters of the model, $n_{j\sigma}=a^{\dagger}_{j\sigma}a_{j\sigma}$, and the fermionic creation operators $a_{i\sigma}$ satisfy $a_{i\sigma}^{\dagger}a_{j\tau}+a_{j\tau}a_{i\sigma}^{\dagger}=\delta_{ij}\delta_{\sigma\tau}$.

Ref.~\cite{setia2019superfast} introduced the superfast encoding method to encode above Hamiltonian into linear combinations of $\mathcal{O}(1)$-local Pauli operators. Specifically, the superfast encoding introduces an ancillary qubit for every hoping term in $H_{FH}$ defined on a $a\times b$-sized lattice, giving an overall system size of $4ab-2a-2b$ qubits. Let $Z^{\uparrow}_k$ denote a Pauli $Z$ operator applied to the qubit on the vertical edge adjacent to the vertex $k$. Operators on other adjacent edges are defined analogously by using $\{\rightarrow,\leftarrow,\uparrow,\downarrow\}$ superscripts.

Using the above representation, the nearest-neighbor couplings for horizontal edges map to $5$-local operators:
\begin{align}
a_{k+1}^{\dagger}a_k+a^{\dagger}_ka_{k+1}\mapsto \frac{1}{2}Y_k^{\rightarrow}\left(Z_k^{\downarrow}Z_{k+1}^{\uparrow}-Z_k^{\uparrow}Z_k^{\leftarrow}Z_{k+1}^{\rightarrow}Z_{k+1}^{\downarrow}\right), 
\end{align}
while the vertical nearest-neighbour couplings are encoded by $7$-local operators:
\begin{align}
    a_j^{\dagger}a_k+a_k^{\dagger}a_j\mapsto \frac{1}{2}\left(Z_k^{\leftarrow}Z_k^{\rightarrow}Z_k^{\uparrow}Z_j^{\leftarrow}Z_j^{\rightarrow}Z_j^{\downarrow}-I\right).
\end{align}
Finally, the onsite interactions 
\begin{align}
    n_{i\uparrow}n_{i\downarrow}\mapsto \frac{1}{4}\left(I-Z_k^{\leftarrow}Z_k^{\uparrow}Z_k^{\rightarrow}Z_k^{\downarrow}\right)\left(I-Z_{k^{\prime}}^{\leftarrow}Z_{k^{\prime}}^{\uparrow}Z_{k^{\prime}}^{\rightarrow}Z_{k^{\prime}}^{\downarrow}\right),
\end{align}
where the primed indices correspond to fermions in spin down lattice and  the unprimed ones to the sites in the spin up lattice. This implies each onsite term can be represented by a $8$-local Pauli operator. 

\subsection{VQE Algorithm Simulation}
Here, we consider to utilize the Hamiltonian variational~(HV) ansatz to estimate the ground state energy of Fermi-Hubbard model. The HV ansatz is based on intuition from the quantum adiabatic theorem, which states that one can evolve from the ground state of a Hamiltonian $H_A$ to the ground state of another Hamiltonian $H_B$ by applying a sequence of evolutions of the form $e^{-itH_A}$ and $e^{-itH_B}$ for sufficiently small time $t$. In our case, the HV ansatz starts from the ground state of the non-interacting Hubbard Hamiltonian $(U=0)$ which is essentially a slater determinant quantum state. Each layer of the HV ansatz is constructed by
\begin{align}
    e^{-iH_{v}t_v}e^{-iH_{h}t_h}e^{-iH_ot_o},
    \label{Eq:FHansatz}
\end{align}
where time series $\{t_v,t_h,t_o\}$, $H_V$ is the vertical hopping term, $H_h$ is the horizontal hopping term and $H_o$ is the onsite term. Suppose  the HV ansatz contains $p$ layers, the initial quantum state is $|\phi\rangle$, then the VQE algorithm minimizes the energy function
\begin{align}
    E(\vec{t})=\langle\phi|\prod\limits_{j=1}^pe^{iH_{v}t^{(j)}_v}e^{iH_{h}t^{(j)}_h}e^{iH_ot^{(j)}_o}H_{FH}\prod\limits_{j=1}^pe^{-iH_{v}t^{(j)}_v}e^{-iH_{h}t^{(j)}_h}e^{-iH_ot^{(j)}_o}|\phi\rangle
\end{align}
in each optimization step. It is shown that Hamiltonian $H_{\rm FH}$ can be decomposed by linear combinations of local Pauli operators, and so the energy function is a sum of $\mathcal{O}(n^2)$ mean values of local observable.  
\begin{corollary}
\label{corollary:hubbardsimulation}
    Given a two-dimensional Fermi-Hubbard model defined on a $(a\times b)$-sized lattice, a $p$-depth Hamiltonian Variational ansatz (given by Eq.~\ref{Eq:FHansatz}) with parameters $\{t^{(j)}_v,t^{(j)}_h,t^{(j)}_o\}_{j=1}^p\in[-\pi,\pi]^{3p}$ and a slater determinant initial state, then each step of the corresponding VQE program can be simulated by a classical algorithm with a run time 
    \begin{align}
         \mathcal{O}\left(\frac{4ab}{\epsilon^2}\left(\frac{2L}{\epsilon}\right)^{e^{4\pi^2 ep\mathfrak{d}}\log(2L/\epsilon)}\right),
    \end{align}
    where the constant $\mathfrak{d}$ represents the maximum degree of the interaction graph induced by $H_{\rm FH}$ and the locality $L\leq 8$.
\end{corollary}
\begin{proof}
    The superfast encoding method may encode a $(a\times b)$-sized Hamiltonian into a $(2a\times 2b)$-sized Hamiltonian. Then taking $t=2\pi$, $n=4ab$ into Theorem~\ref{theorem1} may conclude the result directly.
\end{proof}
When the HV ansatz depth $p\leq\mathcal{O}(1)$, the above result implies VQE algorithm can be efficiently simulated by a classical algorithm, and this further suggests VQE algorithms may lose exponential speed-up in terms of the system size.
\subsection{Quantum State Property Simulation}
Given a $2$-dimensional Fermi-Hubbard model, determining its quantum phase diagram under specific external parameters is of significance. Suppose the ground state $|\psi_g\rangle$ of $H_{\rm FH}$ has been prepared by a VQE approach, that is
\begin{align}
    |\psi_g\rangle=\prod\limits_{j=1}^pe^{-iH_{v}t^{(j)}_v}e^{-iH_{h}t^{(j)}_h}e^{-iH_ot^{(j)}_o}|\phi\rangle.
\end{align}
The ground state property can be characterized by the value of $\langle\psi_g|O|\psi_g\rangle$, where $O$ represents the target order parameter. For example, observables related to metal-insulator transition, Friedel oscillations and antiferromagnetic orders are general local~\cite{stanisic2022observing}, while observables related to the spin-charge separation, local-gapped phases and other complex topological quantum phases are general global~\cite{montorsi2012nonlocal,barbiero2013hidden}. Our classical algorithm can provide an estimation to $\langle\psi_g|O|\psi_g\rangle$, in both local (symmetry breaking phase) and global (topological phase) scenarios. 

\section{Classical Simulation for QAOA}
\label{sec:QAOA}
In theoretical computational science, constraint satisfaction problems encompass a wide range of typical problems, such as Maximum Cut, Maximum Independent Set, and Graph Coloring~\cite{gross2018graph}. These problems define their constraints as clauses, with a candidate solution represented by a specific assignment of the corresponding binary variables. The objective of these problems is to find an optimal assignment that maximizes the number of satisfied clauses. In other words, solving a constraint satisfaction problem can be reformulated as optimizing a quadratic function involving binary variables. However, finding the exact solution is widely recognized as an 
$\rm NP$-hard problem~\cite{garey1979computers}. Consequently, an alternative approach is to seek an approximate solution.  Inspired by the quantum annealing process~\cite{kadowaki1998quantum}, QAOA was proposed and applied to solve constraint satisfaction problems. Although the prospects of achieving quantum advantages through QAOA remain unclear, it provides a simple paradigm for optimization that can be implemented on near-term quantum devices.

Here, we focus on the MaxCut problem.

\begin{definition}[Maximum Cut problem]
Considering an unweighted $\mathfrak{d}$-regular graph $G=(V, E)$ with the vertices set $V=\{v_1,\cdots,v_n\}$ and the edges set $E=\{e_{i,j}\}$, the Maximum Cut problem aims at dividing all vertices into two disjoint sets such that maximizing the number of edges that connect the two sets. In the context of QAOA, the problem-oriented Hamiltonian $H_A^{\rm {MaxCut}}$ is defined as 
\begin{equation}\label{eq:H1_MaxCut}
  H_A^{\rm {MaxCut}}= \frac{1}{2}\sum_{e_{i,j}\in E}(\mathbb{I}^{\otimes n}-Z_i\otimes Z_j),
\end{equation}
and mixer $H_B=\sum\limits_{i=1}^nX_i$.
\end{definition}

Subsequently, by iteratively applying \(H_A\) and \(H_B\) to the initial state $\rho$ for \(p\) rounds, the QAOA objective function is given by the following expectation value
\begin{align}\label{eq:Objective Function 1}
f(\vec{ t},\vec{\gamma})=\mathrm {Tr}\left[H_AU(\vec{ t},\vec{\gamma})\rho U(\vec{ t},\vec{\gamma})^{\dagger}(\bm \theta)\right],
\end{align}
where $\rho=(|+\rangle\langle+|)^{\otimes n}$ denotes the uniform superposition over computational basis states and the QAOA circuit
\begin{align}\label{eq:Quantum Circuit}
U(\vec{ t},\vec{\gamma})=\prod_{k=1}^pe^{-i t_kH_A}e^{-i\gamma_kH_B}.
\end{align}
The statistical estimation of \(f(\vec{ t},\vec{\gamma})\) can be achieved by repeating the aforementioned process with identical parameters and computational basis measurements. After defining \(f(\vec{ t},\vec{\gamma})\), the next step involves iteratively updating $\vec{ t},\vec{\gamma}$ through classical optimization methods to maximize \(f(\vec{ t},\vec{\gamma})\) and obtain the global maximum point
\begin{align} \label{eq:Max}
(\vec{ t},\vec{\gamma})^{*}=\arg\max_{\bm \vec{ t},\vec{\gamma} \in\mathcal{D}}f(\vec{ t},\vec{\gamma}),
\end{align}
where the domain $\mathcal{D}=[0,2\pi]^{2p}$. 

Since all Pauli terms in $H_A$ are local operators, it is interesting to note that such local property enables our algorithm to bypass the $2$D constraint. Specifically, one can estimate $\langle+^n|U^{\dagger}(\vec{ t},\vec{\gamma})H_AU(\vec{ t},\vec{\gamma})|+^n\rangle$ by computing
\begin{align}
    \langle+^n|U^{\dagger}(\vec{ t},\vec{\gamma})Z_iZ_jU(\vec{ t},\vec{\gamma})|+^n\rangle
\end{align}
for $e_{ij}\in E$. Let $\vec{t}=(\vec{ t},\vec{\gamma})\in[0,2\pi]^{2p}$ and using Eq.~\ref{Eq:clusterexpansion}, we have
\begin{align}
    V_{i,j}(\vec{t})=\sum\limits_{\substack{m_1\geq0\\\cdots\\m_{2p}\geq 0}}^M\sum\limits_{\bm V_1\cdots, \bm V_{2p}\in\mathcal{G}_m^{2p,Z_iZ_j}}\frac{\prod_{k=1}^{2p}(\bm\lambda^{\bm V_k}(-it_k)^{m_k})}{\prod_{k=1}^{2p}\bm V_k!m_k!}\sum\limits_{\substack{\sigma_1\in \mathcal{P}_{m_1}\\\cdots\\\sigma_L\in \mathcal{P}_{m_{2p}}}}\left[h_{V_{\sigma_1(1)}},\cdots [h_{V_{\sigma_{2p}(m_{2p})}},Z_iZ_j]\right],
\end{align}
where $M\leq\mathcal{O}\left(e^{2\pi ep\tau\mathfrak{d}}\log^2(1/\epsilon)\right)$ (according to lemma~\ref{lemma:clustersupp}), with $\tau=\max\{\abs{ t_k},\abs{\gamma_k}\}_{k=1}^p$. Using lemma~\ref{lemma:VTCompute}, a $\epsilon$-approximation to $\langle+^n|V_{ij}(\vec{t})|+^n\rangle$ can be computed in $\tilde{\mathcal{O}}((e^{2\pi e\tau p\mathfrak{d}}/\epsilon)^{e^{2\pi e\tau p\mathfrak{d}}})$ running time. Let $\epsilon$ to $\epsilon/\abs{E}$, the $\epsilon$-approximation to the objective function $f$ can be obtained in
\begin{align}
    \mathcal{O}((e^{2\pi e \tau p\mathfrak{d}}\abs{E}/\epsilon)^{e^{2\pi e\tau p\mathfrak{d}}})
\end{align}
classical running time.

\end{document}